\documentclass[twoside,11pt]{article}

\usepackage{blindtext}
\usepackage{jmlr2e}

\usepackage[utf8]{inputenc} % allow utf-8 input
\usepackage[T1]{fontenc}    % use 8-bit T1 fonts
\usepackage{hyperref}       % hyperlinks
\usepackage{url}            % simple URL typesetting
\usepackage{booktabs}       % professional-quality tables
\usepackage{amsfonts}       % blackboard math symbols
\usepackage{nicefrac}       % compact symbols for 1/2, etc.
\usepackage{microtype}      % microtypography
\usepackage{xcolor}         % colors
\usepackage{dsfont}
\usepackage{mathtools}
\usepackage{amsmath,amssymb,bm}
\usepackage{algorithm}
\usepackage{algcompatible}
\usepackage{graphicx,wrapfig}
\usepackage{textcomp}
\usepackage{stmaryrd}
\usepackage{mathtools}
\usepackage{array,multirow}
\usepackage{empheq}
\usepackage{lipsum}
\usepackage[justification=justified,singlelinecheck=false]{caption}

\newcommand{\bxi}{\bm{\xi}}
\newcommand{\bpsi}{\bm{\psi}}
\newcommand{\bmu}{\bm{\mu}}
\newcommand{\bLambda}{\bm{\Lambda}}

\newcommand{\btheta}{\bm{\theta}}
\newcommand{\e}{\mathbb{E}}

\newcommand{\dkl}{D_{\textup{KL}}}

\newcommand{\bS}{\bm{\lambda}}
\newcommand{\bL}{\bm{x}}
\newcommand{\bLX}{\bm{X}}
\newcommand{\bX}{\mathbf{X}}
\newcommand{\wbS}{\widetilde{\bS}}
\newcommand{\wS}{\widetilde{\lambda}}
\newcommand{\wbL}{\widetilde{\bm{x}}}
\newcommand{\wbmu}{\widetilde{\bm{\mu}}}
\newcommand{\wcg}{C_{m \to k}}
\newcommand{\wcnb}{C_{m \to k}^{\textup{NB}}}
\newcommand{\overbar}[1]{\mkern 1.5mu\overline{\mkern-1.5mu#1\mkern-1.5mu}\mkern 1.5mu}
\newcommand{\wcnbd}{\overbar{C}_{m \to k}^{\textup{NB}}}
\newcommand{\wcasl}{C_{m \to k}^{\textup{ASL}}}
\newcommand{\wcasles}{\widehat{\bm{C}}_{m \to k}^{\textup{ASL}}}
\newcommand{\wcnbsles}{\widehat{\bm{C}}_{m \to k}^{\textup{NB}}}

\newcommand{\qed}{$\hfill\blacksquare$}

\DeclareMathOperator{\asc}{\xrightarrow{\textup{a.s.}}}
\DeclareMathOperator{\idc}{\xrightarrow{\textup{dist.}}}
\DeclareMathOperator{\T}{\mathsf{T}}
\DeclareMathOperator*{\argmax}{arg\,max}

\newtheorem{mycorollary}{Corollary}
\newtheorem{mydefinition}{Definition}
\newtheorem{myremark}{Remark}

\usepackage{lastpage}
\jmlrheading{27}{2026}{1-\pageref{LastPage}}{7/23; Revised 4/26}{5/26}{23-0910}{Mert Kayaalp and Ali H. Sayed} \ShortHeadings{title}{Kayaalp and Sayed}

\ShortHeadings{Causality in Social Learning}{Kayaalp and Sayed}
\firstpageno{1}

\begin{document}

\title{Causal Influences over Social Learning Networks}

\author{\name Mert Kayaalp\thanks{Corresponding author. Part of this work was conducted while the author was at EPFL.} \email \{mert.kayaalp\}@idsia.ch \\ 
\addr  Dalle Molle Institute for Artificial Intelligence (IDSIA USI-SUPSI)\\
       CH-6900, Lugano, Switzerland \AND Ali H. Sayed \email \{ali.sayed\}@epfl.ch \\
       \addr \'{E}cole Polytechnique F\'ed\'erale de Lausanne (EPFL)\\
       CH-1015, Lausanne, Switzerland}

\editor{Christophe Giraud}

\maketitle

\begin{abstract}%

This paper investigates causal influences between agents linked by a social graph and interacting over time. In particular, the work examines the dynamics of social learning models and distributed decision-making protocols, and derives expressions that reveal the causal relations between pairs of agents and explain the flow of influence over the network.  The results turn out to be dependent on the graph topology and the level of information that each agent has about the inference problem they are trying to solve. Using these conclusions, the paper proposes an algorithm to rank the overall influence between agents to discover highly influential agents. It also provides a method to learn the necessary model parameters from raw observational data. The results and the proposed algorithm are illustrated by considering both synthetic data and real social media data.  
\end{abstract}


\section{Introduction}

Identifying influential agents in a network is a crucial problem with wide-ranging applications, such as detecting propaganda-sharing accounts \citep{smith2021automatic} or selecting individuals to advertise to \citep{lagree2018algorithms}. However, over any social network, information usually spreads and mixes, leading to ripple effects that make the discovery of influence rather challenging. For example, the information leaving a source agent may be altered and combined with comments/data from other agents along the path until it reaches its destination. Most prior works measure influence through some network topology-based properties such as the eigenvector centrality of an agent \citep{dablander2019node}, or through some descriptive importance factor depending on the problem at hand \citep{kempe2003maximizing,banerjee2013diffusion,valentina2023discovering}. In comparison, this work treats influence as a {\em causal} quantity and approaches it from the perspective of structural causal models \citep{wright1934method,pearl2009causality}. More specifically, influence will be defined as the change in behavior of the network when interventions occur at individual agents. This is a useful method to discard spurious and non-causal associations, unlike other methods based, for example, on the use of Granger causality \citep{granger1969investigating}. Obviously, conducting interventional experiments may not be always feasible over real world social networks. However, with the help of appropriate representative models, one can rely on the use of raw observational data \citep{pearl2018book}.

To that end, this paper considers some common mathematical models for social learning over graphs \citep{degroot1974,acemoglu_2011,jadbabaie_2012,chamley2013,nedic_2017}. These models are widely used to model opinion formation from a behavioral perspective \citep{mossel2017opinion,krishnamurthy2022dynamics}, and their generalization to real-world scenarios is verified by field experiments \citep{chandrasekhar2020testing}. From an engineering perspective, the models can be utilized to design distributed decision-making systems where a group of agents (e.g., robots, sensors) collaboratively make decisions about the environment \citep{djuric2012}. However, most of the research on social learning models has focused on analyzing their convergence and error behavior under different assumptions, such as verifying whether truth learning occurs in the presence of malicious agents supplying misinformation \citep{hare2019,li2019_malicious}, or whether the models can track drifts in the underlying hypothesis \citep{bordignon_2021}. In contrast, this study utilizes these opinion formation models to examine \emph{causal effects} over social graphs. To the best of our knowledge, this appears to be the first study to do so by studying the diffusion of influence over space and time from a causal inference perspective.

\subsection{Contributions} 
\begin{itemize}
    \item Our work proposes a novel causal impact criterion for dynamic social networks in Sec.~\ref{sec:causal_effect_def}. It quantifies how much an agent $m$ affects another agent $k$ for each agent pair \((m,k)\).
    \item We derive closed-form expressions for these impact factors for two social network models: non-Bayesian social learning \citep{jadbabaie_2012} and adaptive social learning \citep{bordignon_2021} in Secs. \ref{sec:main_results_nbsl} and \ref{sec:main_results_asl}, respectively. The non-Bayesian framework \citep{jadbabaie_2012} models network dynamics through a (log)-linear autoregressive process, stated in future expression \eqref{eq:dif_combine_step}, and incorporates in-network interactions as well as possibly hidden out-of-network factors. The adaptive framework \citep{bordignon_2021} additionally introduces temporal weightings that emphasize more recent interactions, as stated in \eqref{eq:LLF_evolution_asl}. Beyond general causal influence expressions, we also analyze useful special cases to illustrate how the causal effects depend on the model parameters. 
    \item Analyzing the influence for every pair \((m,k)\) results in a network causality matrix, which offers various options to rank agents for their overall influence. In Sec.~\ref{sec:causal_ranking}, we propose a particular algorithm, {\sf \small CausalRank}, to rank agents based on their overall influence on the network, while accounting for the fact that more impactful agents should be weighted more. This algorithm leverages the eigenvector centrality of the bipartite causal relations matrix introduced in Sec.~\ref{sec:causal_effect_def}. 
    \item Additionally, we introduce a graph causality learning (GCL) algorithm in Sec.~\ref{sec:causal_discovery}, designed to estimate causal influences using observational data comprised of social interactions and the graph topology underlying the agents.
    \item In Sec.~\ref{sec:computer_sims}, we illustrate our results with synthetic data. In Sec.~\ref{sec:twitter_app}, we apply these findings to real social media data, demonstrating their practical usefulness.
\end{itemize}

\subsection{Challenges for estimating influence over social networks}

\begin{figure}[t]
\centering
\begin{minipage}{0.3\textwidth}
  \includegraphics[width=\linewidth]{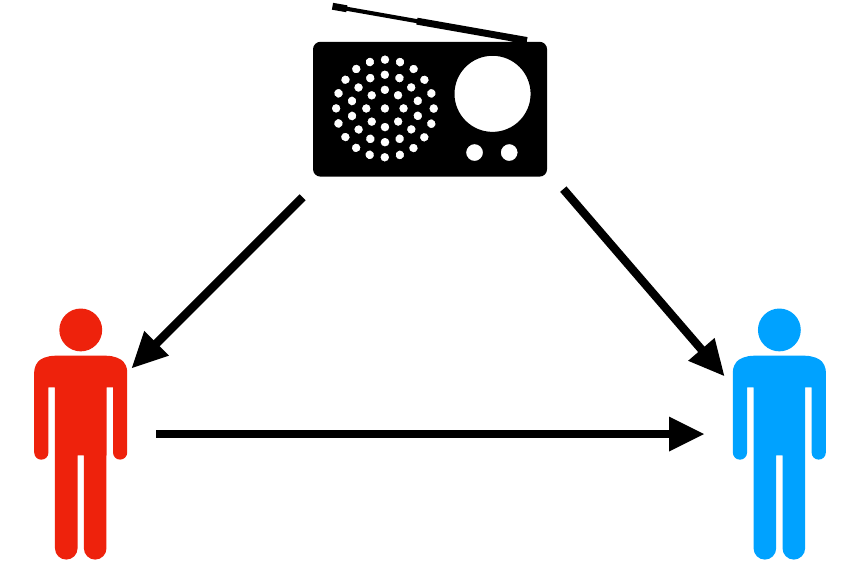}
  %\label{fig:contributions_confounding}
\end{minipage}%
\hfill
\begin{minipage}{0.25\textwidth}
  \includegraphics[width=\linewidth]{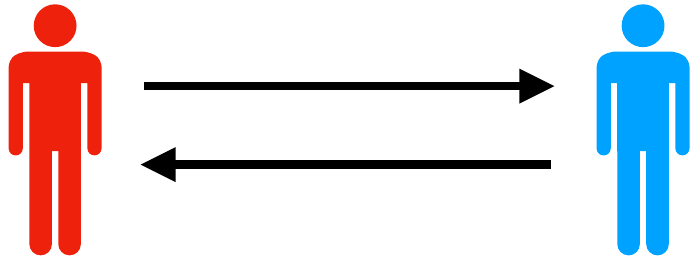}
      %\label{fig:contributions_time_series}
\end{minipage}%
\hfill
\begin{minipage}{0.3\textwidth}
  \includegraphics[width=\linewidth]{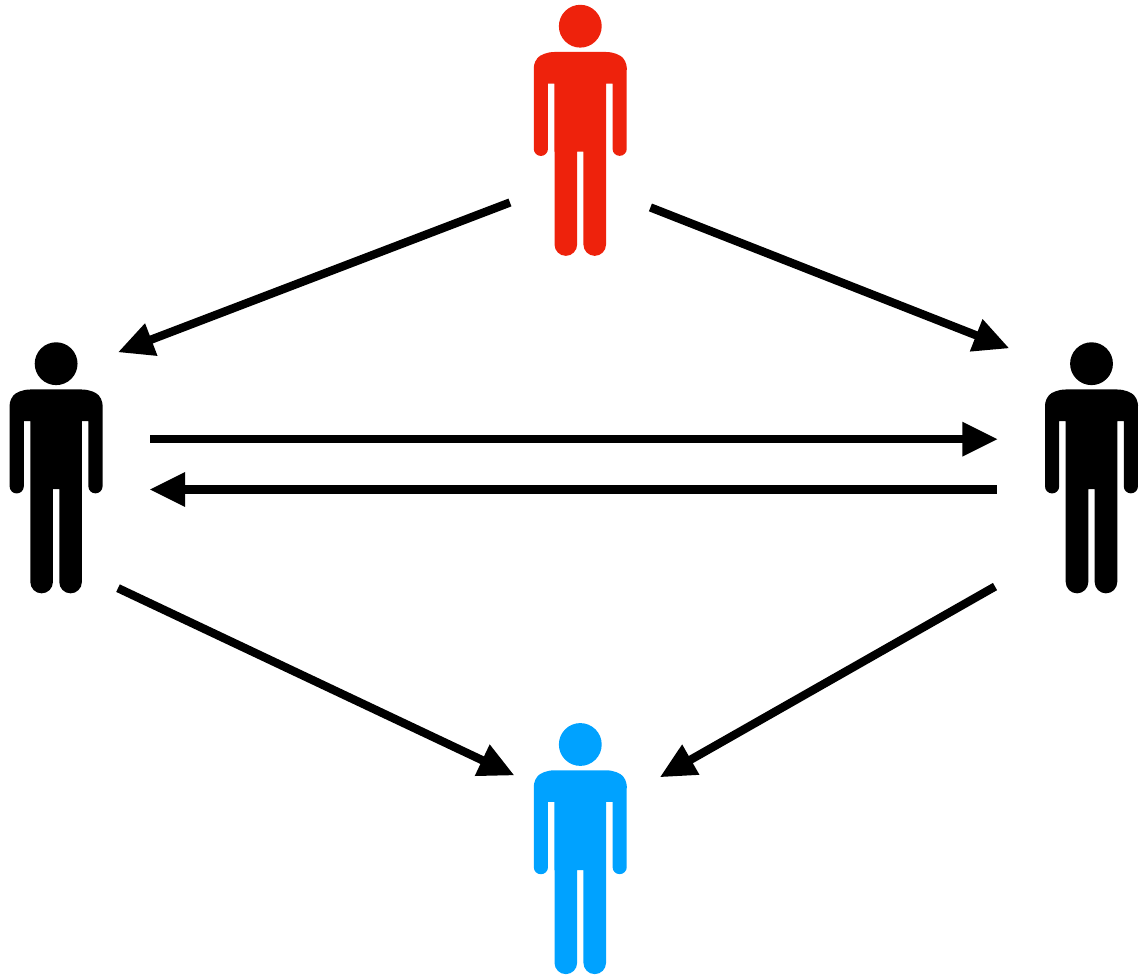}
      %\label{fig:contributions_mixing}
\end{minipage}\caption{Some inherent challenges in estimating the causal influence of one agent on another, represented by the red and blue agents, respectively. \textit{(Left):} There can be confounding factors that influence both parties and induce non-causal correlation. \textit{(Middle):} Relationships within social networks are often bidirectional that are not instantaneous, but rather spread out over time. \textit{(Right):} Information transmitted from the source is processed and potentially changed along the way.}\label{fig:challenges_inf}
\end{figure}

Influence estimation over social networks faces several challenges. \\

\noindent \textbf{Confounding factors.} Understanding influence requires disentangling correlation from causation. For example, over a social network, it is often observed that individuals who are connected tend to have similar (correlated) opinions. However, this does not necessarily imply a causal relationship and there can also exist confounding factors. For instance, individuals may obtain information from similar external sources (say, the same TV channels), or they may be connected to others who share similar preferences (a.k.a. homophily) --- see left panel in Fig.~\ref{fig:challenges_inf}. Similar issues arise in distributed decision-making systems, such as networks of wireless sensors or robots. Devices that communicate with each other are often in spatial proximity, leading to correlated observations. Therefore, accounting for these confounding factors is crucial for discovering {\em true causal} relationships. \\

\noindent \textbf{Temporal dynamics.} Social influence is not a one-time occurrence but rather a continuous process that unfolds over time. Therefore, we adopt a time-series approach to capture this dynamic nature. Unlike directed acyclic graph based models in the literature, we accommodate both cyclic networks and bidirectional links to capture feedback mechanisms --- see middle panel in Fig.~\ref{fig:challenges_inf}. Doing so is necessary in order to discover the propagation of influence over both space and time. \\

\noindent \textbf{Mixing and diffusion of information.} When examining the influence of an agent $m$ on another agent $k$, it is essential to acknowledge that information leaving agent $m$ can undergo alterations and become intertwined with the opinions of other agents in the network before reaching agent $k$ --- see right panel in Fig.~\ref{fig:challenges_inf}. This phenomenon introduces additional complexity to the study of influence propagation over graphs. \\

In this work, the social network models we consider allow us to treat these challenges together in some detail. Specifically, the expressions we derive within a rigorous causal theoretical framework quantify how the instantaneous and direct social effects diffuse over time and network. In doing so, as we seek to understand the total and overall causal effects, we take into account potential hidden confounding factors.

\section{Related Work}

\paragraph{Influence quantification.} \hspace{-0.9em} Most existing works in the literature quantify influence by relying on network topology-based measures (such as degree centrality) \citep{dablander2019node} or by examining task-specific importance factors \citep{kempe2003maximizing, shah2011, banerjee2013diffusion,valentina2023discovering}. However, these measures often lack a causal interpretation and they may fail to eliminate non-causal correlation-inducing factors. Another line of work use randomized experiments in order to identify causal relations, which, in principle, can discard latent confounders \citep{aral2012identifying,bond2012}. Unfortunately, conducting controlled experiments is impractical for many real-world scenarios. In this work, we study social influence from a structural causal model framework \citep{wright1934method,pearl2009causality}, which will enable us to circumvent these issues with the help of useful models for interactions over social graphs. 

\paragraph{Causal influence.} \hspace{-0.9em} While a few previous works have treated social influence as a causal effect, they have some limitations. For instance, \citep{sridhar2022} only considered a single time step of interaction and did not address the time series formulation, which is critical in capturing information diffusing over time. Likewise, \citep{soni2019} utilized Granger causality \citep{granger1969investigating}, which is a predictive tool that only depicts the natural behavior of the system and is not robust against latent confounders \citep{eichler2010granger}. In contrast, our approach is based on the causality principle that manipulating the causes, while keeping other variables constant, changes the effect \citep{woodward2004book,pearl2009causality}. This idea was first applied to general time series in \citep{eichler2010granger}. Here, we apply it to a time series of variables over a graph consisting of interconnected agents. 

\paragraph{Networked time series.} \hspace{-0.9em} Some works such as \citep{lee2023finding,smith2018,smith2021automatic} relied on the Rubin's potential outcomes framework \citep{rubin1974estimating}. Specifically, \citep{smith2018,smith2021automatic} analyzed the causal impact of each agent by \emph{directly} modeling the response of other agents under a binary treatment (i.e., intervention) on a particular agent. Our work is complementary but different in the sense that we apply the procedure of hypothetical interventions on graphical models \citep{pearl2009causality} to the canonical models of social learning networks \citep{jadbabaie_2012,nedic_2017,bordignon_2021}. While doing so, we derive the expressions for the causal impact factors by analyzing the step-by-step dynamics of the social networks, and arrive at results that clarify how causal influences propagate over a graph. 

\paragraph{Social learning.} \hspace{-0.9em} The literature on mathematical models for social learning is rich --- see the surveys \citep{krishnamurthy2022dynamics, mossel2017opinion} and the papers \citep{jadbabaie_2012,bordignon_2021,nedic_2017}. In this study, we focus on the DeGroot type of networked consensus modeling \citep{degroot1974, golub2010naive}, which has been widely applied in social network research. In particular, we examine the non-Bayesian extensions studied by \citep{jadbabaie_2012, molavi2018theory}, which allow agents to interact with the environment and to exchange information within the network. These models and their extension to adaptive agents \citep{bordignon_2021} are briefly reviewed in Sec.~\ref{sec:social_learning_model}. 
Another related line of work in social learning considers stubborn agents \citep{acemoglu2010spread, ghaderi2014, vial2021local, hunter2022optimizing}. 
This literature is primarily concerned with forward modeling. It considers questions like how agents with fixed opinions (e.g., misinformation spreading agents) affect collective behavior, including whether consensus is disrupted \citep{acemoglu2010spread}, what equilibrium or disagreement patterns emerge \citep{ghaderi2014}, how quickly beliefs converge  \citep{vial2021local}, and how outcomes depend on the agents’ locations over the network \citep{hunter2022optimizing}. 
Our objective is different in that we study hypothetical interventions in which an agent is assigned an arbitrary belief in order to identify that agent’s causal influence on the rest of the network. This has fundamentally different goals and leads to distinct formulas and conclusions.

\paragraph{Network interference.} \hspace{-0.9em} At the intersection of causal inference and networks, there also exists a body of work focused on network interference \citep{toulis2013estimation, sussman2017elements, agarwal2022network}. These studies examine causal inference when interventions on certain agents can affect others within the network, and aim to design randomized experiments that take this into account. However, they typically assume that only immediate neighbors of an agent can impact an individual's response. Although our problem setting differs, our analytical contributions may still prove useful for this area of research, since we discover how influence and interference propagate throughout the network. For example, our findings allow us to quantify how an agent's influence diminishes with increasing distance from other agents in the graph. 

\paragraph{Notation.} Random variables are denoted in bold, e.g., \(\bm{x}\). A sequence of random variables \(\{\bm{x}_i\}\) over index \(i\) converging to a random variable \(\bm{x}\) in distribution is denoted by \(\bm{x}_i \idc \bm{x}\). Almost sure convergence is denoted by \(\bm{x}_i \asc \bm{x}\). The \(k\)-th entry of a vector \(v\) is denoted by \(v_k\) or \([v]_k\). We denote the all-ones vector of dimension \(K\times 1\) by \(\mathds{1}_{K}\). For distributions \(p_1\) and \(p_2\), \(\dkl(p_1||p_2)\) denotes their KL divergence.

\section{Social Learning}\label{sec:social_learning_model}

\subsection{Problem setting}

We consider a network of \( K \) agents that are trying to infer the hypothesis that best explains their observations about the world. More formally, agents are trying to learn the true state of nature \( \theta^\circ \in \Theta \) from a finite set of \( H \) hypotheses, \( \Theta = \{\theta_1,\theta_2,\dots,\theta_H \}\). At each time instant \( i \), each agent \( k \) receives a personal and partially informative observation \( \bxi_{k,i}\), which encapsulates all out-of-network information available to \(k\) at time \(i\) and is distributed according to some marginal likelihood function \( L_k (\xi | \theta^\circ)\) dependent on the true state.
The true state of nature can be viewed as an underlying environmental or contextual variable that governs the data-generating process of the observations.
In engineering settings, it may also capture a distributional regime or shift. 
As a recent example, the state of nature may encode whether there is an ``AI bubble''.
People then gather evidence from their local environments, for instance, news reportings, company revenues and spending, user adoption, or their own experience with AI tools, and update their beliefs through interactions with their peers.

The observations \( \bxi_{k,i}\) are assumed i.i.d.\ over time, and the likelihood function  \( L_k (\xi | \theta^\circ)\) is assumed to be time-invariant. From a behavioral perspective of social networks, the observations often represent external sources of information. Fixed distribution assumptions correspond to these sources providing relatively stable information over time. For instance, a TV channel may consistently support Candidate A over Candidate B for an election, justifying the assumption of time-invariant likelihood functions. Therefore, studying causal influences under the fixed distribution model is useful in many situations. As the analysis in the paper shows, even in this case, the arguments are not straightforward and there are many useful insights that can be gained from such studies. We acknowledge that in some applications the distributions can drift with time. We leave the study of this scenario to future work. 

However, the observations are not necessarily independent across the agents. Since spatial independence is not required, our setting does not exclude possible latent confounders in the observation model. Each agent \(k\) knows the likelihood model \( L_k (\xi | \theta)\) for every possible hypothesis \( \theta \in \Theta \). If the distribution \( L_k (\xi | \theta)\) for a hypothesis \(\theta \neq \theta^\circ\) is sufficiently different from \( L_k (\xi | \theta^\circ)\), then it is easier for agent \( k \) to distinguish $\theta^\circ$ from $\theta$. This ``distinguishing power'' can be quantified using the KL divergence between the likelihood models, namely,
\begin{equation}
    d_k (\theta) \triangleq \dkl \Big (L_k(\cdot|\theta^{\circ}) || L_k(\cdot|\theta) \Big).
\end{equation}
The larger this quantity is, the more \emph{informative} agent \(k\)'s observations are for distinguishing a wrong hypothesis $\theta$ from the true hypothesis $\theta^\circ$. 
It can also be interpreted as the degree to which an agent is exposed to the true state of nature.
In order to avoid pathological cases, we assume that \(d_k (\theta) < \infty\) for each agent \(k\) and hypothesis \(\theta\). This condition makes sure that the likelihood functions for different hypotheses share the same support; and no observation alone is sufficient to refute any hypothesis.
\begin{mydefinition}[Global identifiability]\label{def:global_identifiability}
For each wrong hypothesis \(\theta \neq \theta^\circ\), if there exists at least one clear-sighted agent $k^\star$ that can distinguish $\theta$ from the true hypothesis $\theta^\circ$, i.e., \(d_{k^\star} (\theta) >0\), then the true state of the environment is said to be globally identifiable. \qed
\end{mydefinition}
Notice that global identifiability does not require \emph{local} identifiability, which is the ability of each individual agent to distinguish $\theta^\circ$ from any other hypothesis without cooperation with the other agents.

Based on the local observations and on interactions with other agents, each agent \( k \) forms an opinion (i.e., a belief vector), denoted by \( \bmu_{k,i} \), which is a probability mass function defined over the set of hypotheses $\Theta$. Here, the entry \( \bmu_{k,i} (\theta) \) quantifies the confidence \( k \) has about \( \theta \in \Theta \) being the true hypothesis \( \theta^\circ\) at time $i$. We assume that all agents have positive initial beliefs and they do not reject any hypothesis at the start of the learning process, i.e., \(\bmu_{k,0}(\theta)>0\) for each agent \(k\) and \(\theta \in \Theta\). Agents exchange their beliefs with each other according to the communication topology described next. 

\newpage
\subsection{Network topology}

\begin{wrapfigure}[17]{l}{6cm}\vspace{-1em}
\includegraphics[width=6cm]{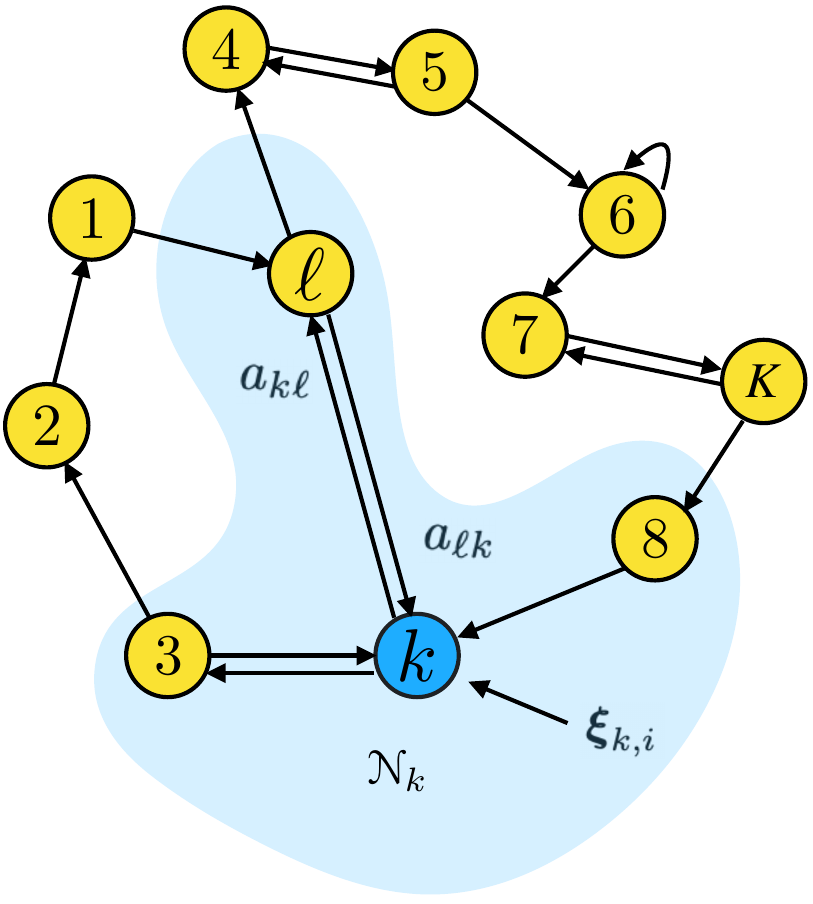}
\vspace{-2.1em}\caption{An illustration of a strongly connected network topology.}\label{fig:network_topology}
\end{wrapfigure} 

The network of agents is assumed to be strongly-connected \citep{sayed_2014}---see Fig.~\ref{fig:network_topology}. This means that there exists a path linking any agent $\ell$ to any other agent $k$, which starts at $\ell$ and ends at $k$. Moreover, there should exist at least one agent $k^\circ$ with a self-loop, i.e., an agent that does not discard its own observations. The entry \( a_{\ell k }\) of the combination matrix \( A=[a_{\ell k}]\) denotes the confidence score that agent \( k \) assigns to the information received from agent \( \ell \). This score is positive if, and only if, agent \( \ell \) is in the neighborhood of agent \( k \). Otherwise it is equal to zero. We denote the neighborhood of agent $k$ with \(\mathcal{N}_k\), i.e., $a_{\ell k} > 0$ if, and only if, $\ell \in \mathcal{N}_k$. The graph underlying the network is directed and hence the combination matrix \( A \) is not necessarily symmetric, i.e., in general \( a_{\ell k} \neq a_{k \ell} \). Nevertheless, the confidence scores that an agent assigns to its neighbors should add up to one. This means that the matrix \( A \) is left-stochastic. Furthermore, since the underlying graph is strongly-connected, $A$ is irreducible and aperiodic. Therefore, in view of the Perron-Frobenius theorem \citep{pillai2005perron}, it satisfies
\begin{align}\label{eq:perron_def}
   A^{\T}\mathds{1}_{K}=\mathds{1}_{K}, \quad \mathds{1}_{K}^{\T}v=1, \quad Av=v,
\end{align}
where \( v \) is called the Perron vector of \( A \) whose entries are all positive and add up to one. Here, the $k$-th entry of the vector $v$ measures the relative centrality of agent $k$ in the network.

\subsection{Non-Bayesian social learning (NBSL)}

In this section, we present the non-Bayesian social learning (NBSL) strategy from \citep{jadbabaie_2012,nedic_2017,lalitha_2018}. In this strategy, based on the observation \( \bxi_{k,i}\) at time instant \( i \), each agent first updates its belief in a \emph{locally Bayesian} fashion to obtain its intermediate belief:
\begin{align}\label{eq:local_bayesian}
    \bpsi_{k,i} (\theta) \propto L_k (\bxi_{k,i} | \theta ) \bmu_{k,i-1}(\theta).
\end{align}
The proportionality sign $\propto$ means that the entries of the resulting vector are normalized to add up to one, as befits a true probability mass function. The motivation for \eqref{eq:local_bayesian} is at least two-fold. From a behavioral point of view, Bayes's rule is used to model human reasoning under uncertainty in neuroscience \citep{friston2005theory} and the social sciences \citep{oaksford2007bayesian,easley2010networks}. From a system design perspective, Bayes's rule is known to be an optimal information processing rule \citep{zellner1988}. In the next step, the intermediate beliefs are shared with other agents, which may average them in a geometric manner to form the updated belief using the confidence scores they assign to their neighbors as follows:
\begin{equation}\label{eq:dif_combine_step}
    \bmu_{k,i}(\theta) \propto \prod_{\ell \in \mathcal{N}_k}  \big ( \bpsi_{\ell,i} (\theta) \big )^{a_{\ell k}}.   
\end{equation}
The combination step \eqref{eq:dif_combine_step} is a non-Bayesian way of combining beliefs and is inspired by the fact that humans are boundedly rational \citep{conlisk1996bounded}. In the above implementation, the agents are combining their neighbors' \emph{instantaneous} opinions, as opposed to behaving in a \emph{fully} Bayesian manner \citep{acemoglu_2011}, which would require global information (such as the graph topology and access to all observations). This requirement makes the fully Bayesian solution NP-hard in general \citep{hkazla2021bayesian}. Although there are variations based on the arithmetic averaging \citep{jadbabaie_2012}, in this work we consider the geometric averaging form described above.

An important quantity for the analysis of the strategy \eqref{eq:local_bayesian}--\eqref{eq:dif_combine_step} is the log-belief ratio vector $\bS_{i}(\theta) \triangleq [\bS_{1,i}(\theta), \dots, \bS_{K,i}(\theta)]^{\T}$ with individual entries defined as
\begin{align}\label{eq:lambda_definition}
\bS_{k,i}(\theta) \triangleq \log\frac{\bmu_{k,i}(\theta^{\circ})}{\bmu_{k,i}(\theta)}.
\end{align}
For an agent $k$, if this quantity is positive for each $\theta \neq \theta^\circ$, then its belief vector is maximized at the true hypothesis and the agent can guess the correct hypothesis. Assuming no additional noise in the belief update and combination steps \eqref{eq:local_bayesian}--\eqref{eq:dif_combine_step}, these equations imply that the vector $\bS_{i}(\theta)$ evolves according to the following linear recursion:
\begin{align}\label{eq:LLF_evolution}
   \bS_{i}(\theta) =  A^{\T}\big (\bS_{i-1}(\theta) + \bL_{i}(\theta) \big ),
\end{align}
where $\bL_{i}(\theta) \triangleq [\bL_{1,i}(\theta), \dots, \bL_{K,i}(\theta)]^{\T}$ is the vector of log-likelihood ratios (LLRs) at time instant \( i \):
\begin{equation}
\bL_{k,i}(\theta) \triangleq \log \frac{L_k(\bxi_{k,i} | \theta^\circ)}{L_k(\bxi_{k,i} | \theta)}.
\end{equation}
\begin{theorem}[Truth learning \citep{nedic_2017,lalitha_2018}]\label{theorem:truth_learning_nbsl}
If agents employ the strategy \eqref{eq:local_bayesian}--\eqref{eq:dif_combine_step}, the asymptotic decay rate of a wrong hypothesis is equal for all agents and given by
\begin{equation}\label{eq:rho_def}
      \frac 1 i \bS_{k,i}(\theta) \asc  \sum_{\ell=1}^K v_\ell d_{\ell} (\theta).
\end{equation}
Moreover, if global identifiability (Def.~\ref{def:global_identifiability}) holds, then all agents learn the truth with full confidence, i.e.,
\begin{equation}
    \bmu_{k,i} (\theta^\circ) \asc 1.
\end{equation}\qed
\end{theorem}

\subsection{Adaptive social learning (ASL)}
The traditional non-Bayesian social learning (NBSL) strategy \eqref{eq:local_bayesian}--\eqref{eq:dif_combine_step} described in the previous section has the drawback that agents do not prioritize new observations against their old observations. In addition to falling short in modelling human behavior, this strategy can be disadvantageous for engineering applications that require adaptation under non-stationary environments. To tackle this issue, the work \citep{bordignon_2021} proposed changing the adaptation step \eqref{eq:local_bayesian} into\footnote{In fact, \citep{bordignon_2021} only considers the special cases of \( \beta = \delta\) and \(\beta = 1 \). However, their results can be adapted to general \(\beta > 0 \) straightforwardly.}
\begin{equation}\label{eq:asl_local_update}
        \bpsi_{k,i} (\theta) \propto L_k^\beta (\bxi_{k,i} | \theta ) \bmu_{k,i-1}^{1-\delta}(\theta),
\end{equation}
where $0 < \delta < 1$ and $\beta > 0 $ are design parameters. In particular, large values of $\delta$ or $\beta$ place more focus on new observations, whereas small values give importance to past beliefs. The modified adaptation step \eqref{eq:asl_local_update} alters the log-belief ratio recursion \eqref{eq:LLF_evolution} to
\begin{align}\label{eq:LLF_evolution_asl}
   \bS_{i}(\theta) =  A^{\T} \big ((1-\delta)\bS_{i-1}(\theta) + \beta \bL_{i}(\theta) \big ).
\end{align}
In contrast to the NBSL case from the previous section where the beliefs converge to the truth almost surely (Theorem~\ref{theorem:truth_learning_nbsl}), in the \emph{adaptive social learning} (ASL) strategy defined by steps \eqref{eq:asl_local_update} and \eqref{eq:dif_combine_step}, the beliefs will have everlasting random fluctuations that are necessary for keeping adaptation alive. The next result states that these random fluctuations have a regular behavior in the limit.
\begin{theorem}[Convergence in distribution \citep{bordignon_2021}]\label{th:asl_conv_dist}
    Under the ASL strategy \eqref{eq:asl_local_update} and \eqref{eq:dif_combine_step}, the log-belief ratios converge in distribution to the following absolutely convergent series:
    \begin{equation}\label{eq:asl_conv_dist}
         \bS_{i}(\theta) \idc \beta \sum_{j=1}^\infty (1-\delta)^{j-1} (A^{\T})^j \bL_{j}(\theta).\vspace{-0.5em}
    \end{equation}\qed
\end{theorem}
In the sequel, we need the following result to analyze the average causal effect between agents.
\begin{mycorollary}[Expected log-belief ratio in ASL]\label{cor:expected_logbelief}Theorem~\ref{th:asl_conv_dist} implies that the log-belief ratios converge in the mean, i.e.,
\begin{equation}\label{eq:expected_logbelief_asl_idle}
     \lim_{i \to \infty} \e [ \bS_{i}(\theta) ] = \frac{\beta}{1-\delta} \Big ((I-(1-\delta)A^{\T})^{-1} - I \Big ) d(\theta)
\end{equation}
where \( d(\theta) \triangleq [d_1(\theta), d_2(\theta), \dots , d_K(\theta)]^{\T}\) is the vector of network KL divergences.
\end{mycorollary}
\begin{proof}
    Since the series on the RHS of \eqref{eq:asl_conv_dist} is uniformly integrable, the expectation on \( \bS_{i}(\theta) \) converges to the expectation of the RHS of \eqref{eq:asl_conv_dist}. The result then follows from the fact that \( \e [\bL_{j}(\theta)] = d(\theta) \) for any time $j$, and from the closed-form expression for the series of absolutely convergent matrices.
\end{proof}\vspace{-2em}

\section{Causal Effects in Social Learning}\label{sec:causal_effect_def}

\begin{wrapfigure}[11]{l}{7cm}
\includegraphics[width=7cm]{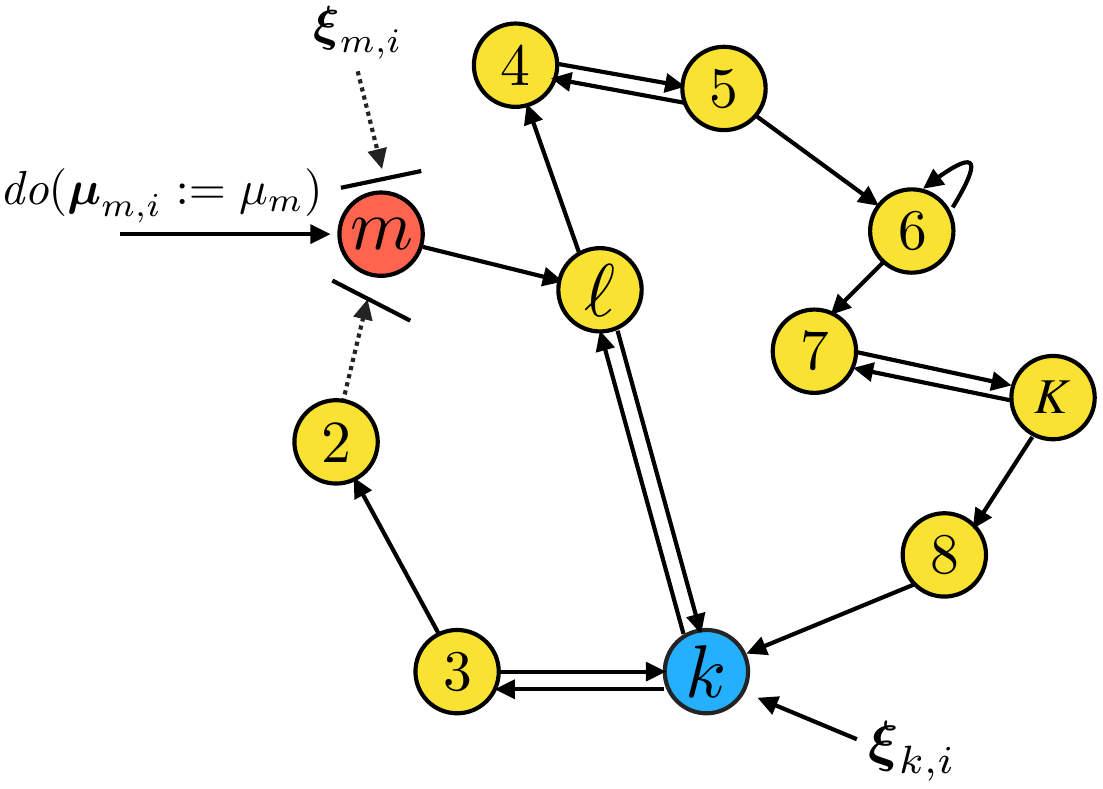}
\caption{An illustration of an atomic and persistent intervention on agent \(m=1\) from Fig.~\ref{fig:network_topology}.}\label{fig:network_intervention}
\end{wrapfigure} 

An intuitive and widely used assertion in defining causality is that manipulation of the causes should result in changes in the effect \citep{woodward2004book}. Based on this principle, interventions on a system, real or hypothetical, have been the primary tool for testing whether a variable causes another \cite[Chap. 1]{pearl2009causality}. In this work, in order to measure the causal influence strength between agents, we rely on the most basic intervention known as \emph{atomic} and \emph{persistent} intervention \cite[Chap. 3]{pearl2009causality}. 
Specifically, in order to measure the social effect of an agent \( m \) on other agents, we analyze the belief evolution of these other agents if the belief of agent \( m \) is fixed to some particular constant belief vector, say, \( \bmu_{m,i}  := \mu_m \) for all time instants --- see Fig.~\ref{fig:network_intervention} for a visual depiction\footnote{We use an atomic intervention since it isolates the causal contribution of a single agent by fixing only that agent's belief and preserving the remaining interaction dynamics of the network. We also use it in a persistent manner since we are interested in the total long term effect rather than a transient one.}. In canonical causality notation, this intervention is denoted by \( \textit{do}(\bmu_{m,i}  := \mu_m ) \) \citep{pearl2009causality}. Since we consider only this intervention in this work and there is no room for ambiguity, we will use the notation that the post-intervention counterparts of the variables in Sec.~\ref{sec:social_learning_model} are topped with the symbol `\( \sim\)'. For example, the log-belief ratio definition from \eqref{eq:lambda_definition} transforms into the following, under the intervention \( \textit{do}(\bmu_{m,i}  := \mu_m ) \):
    \begin{equation}
        \wbS_{k,i}(\theta) \triangleq \log\frac{\wbmu_{k,i}(\theta^{\circ})}{\wbmu_{k,i}(\theta)}.
    \end{equation}
\paragraph{Causal influence strength.} Intuitively, the amount of change in the effect following an intervention on the cause is expected to be related to the causal strength. Therefore, the difference between the post and pre-intervention distributions, or between appropriate functions of these distributions such as expectations, can be used to quantify the causal effect \cite[Chap. 3]{eichler2012causal,pearl2009causality}. In this work, we employ the following definition in order to measure the causal influence of agent \(m \) on agent \(k\) (its connection to established definitions of causality in the literature will be discussed in the sequel):
\begin{equation}\label{eq:cmkgeneral_def}
       \boxed{ \wcg \triangleq \mu_{k,\infty}(\theta^{\circ}) - \widetilde{\mu}_{k,\infty}(\theta^{\circ})}
    \end{equation}
This formula measures the alteration of agent $k$ as a consequence of an intervention on agent $m$. Specifically, it quantifies the magnitude of change of the expected asymptotic belief of agent $k$ on the true state $\theta^\circ$. In general, its value depends on the belief $\mu_m$ of the intervention $\textit{do}(\bmu_{m,i}  := \mu_m )$. Here, we use the following expression for the belief vector, which is explained in the sequel:
\begin{equation}\label{eq:expected_log_belief_trans}
    \mu_{k,\infty}(\theta^{\circ}) \triangleq \dfrac{1}{1 + \!\sum_{\theta \in \Theta \setminus \{ \theta^\circ\}}   \exp \{ - \lambda_{k,\infty}(\theta) \}}.
\end{equation}
This expression is defined in terms of the expected asymptotic log-belief ratio:
\begin{equation}\label{eq:expected_logbelief_ratio}
        \lambda_{k,\infty}(\theta) \triangleq \lim_{i \to \infty} \e [ \bS_{k,i}(\theta) ]
    \end{equation}
The variables topped with the symbol `\( \sim\)' for the intervention case are defined similarly (the existence of the limit for both NBSL and ASL under interventions will be discussed in the sequel). The transformation \eqref{eq:expected_log_belief_trans} is motivated by noting from \eqref{eq:lambda_definition} that:
\begin{equation}
    \exp \big \{\!-\!\bS_{k,i}(\theta)\big\} = \frac{\bmu_{k,i}(\theta)}{\bmu_{k,i}(\theta^{\circ})},
    \end{equation}
which implies
    \begin{equation}
    1 + \!\!\!\! \sum_{\theta \in \Theta \setminus \{ \theta^\circ\}} \!\!\! \exp \big \{\!-\!\bS_{k,i}(\theta)\big\} = \dfrac{1}{\bmu_{k,i}(\theta^{\circ})} \sum_{\theta \in \Theta} \bmu_{k,i}(\theta) 
    = \dfrac{1}{\bmu_{k,i}(\theta^{\circ})},
\end{equation}
which, in turn, yields
\begin{equation}\label{eq:transformation_as}
    \bmu_{k,i}(\theta^{\circ}) = \dfrac{1}{1 + \!\sum_{\theta \in \Theta \setminus \{ \theta^\circ\}}   \exp \{ - \bS_{k,i}(\theta) \}}.
\end{equation}
Here, if we replace log-belief ratio $\bS_{k,i}(\theta)$ with the expected asymptotic log-belief ratio $\lambda_{k,\infty}(\theta)$, we arrive at the definition \eqref{eq:expected_log_belief_trans} for $\mu_{k,\infty}(\theta^{\circ})$. Note that defining $\mu_{k,\infty}(\theta^\circ)$ in terms of the expected log-belief ratios, as opposed to, say, expected beliefs (i.e., \( \lim_{i \to \infty} \e [ \bmu_{k,i}(\theta^\circ) ] \)), will enable us to obtain closed-form expressions for causality in terms of the informativeness of the agents, represented by the entries of \(d(\theta)\), in Sec.~\ref{sec:main_results}. Next, we treat NBSL and ASL separately.

\begin{itemize}
    \item \textbf{Non-Bayesian social learning}. In the idle case (i.e., no intervention) of NBSL, it is known from Theorem~\ref{theorem:truth_learning_nbsl} that under global identifiability, \(\bmu_{k,i}(\theta^\circ) \asc \ 1\) (i.e., \( \bS_{k,i}(\theta) \asc +\infty \) for each \(\theta \neq \theta^\circ\)). Hence, for the NBSL case, the average (i.e., expected) causal influence \eqref{eq:cmkgeneral_def} is given by
    \begin{equation}\label{eq:cmknb_def}
         \wcnb = 1- \widetilde{\mu}_{k,\infty} (\theta^\circ).
    \end{equation}
    This immediately implies that $\wcnb \in [0,1]$, and it gets larger as the post-intervention belief diverges from the truth in expectation. 
    \item \textbf{Adaptive social learning}. In a similar fashion, the causal effect strength for the ASL case is given by the general expression \eqref{eq:cmkgeneral_def}
    \begin{equation}\label{eq:cmkasl_def}
        \wcasl = \mu_{k,\infty}(\theta^\circ)-\widetilde{\mu}_{k,\infty}(\theta^\circ),
    \end{equation}
    where the pre-intervention asymptotic belief \(\mu_{k,\infty}(\theta^\circ)\) can be found by inserting \(\lambda_{k,\infty}(\theta)\) established in Corollary~\ref{cor:expected_logbelief} to \eqref{eq:expected_log_belief_trans}. 
\end{itemize}

\begin{myremark}[Controllability] Notably, the influence $C_{m \to k}$ of agent $m$ on agent $k$ can also be interpreted as the controllability or manipulability of agent $k$ by agent $m$. \qed
\end{myremark}

\section{Theoretical Results}\label{sec:main_results}
In this section, we derive closed-form expressions for \(\widetilde{\lambda}_{k,\infty} \triangleq [\widetilde{\lambda}_{k,\infty}(\theta_1), \dots, \widetilde{\lambda}_{k,\infty}(\theta_H) ]^{\T}\) in terms of the network topology and the informativeness of agents to obtain the causal strength measures \( \wcnb\) and \(\wcasl \). For ease of notation and without loss of generality, we set \( m = 1 \). One can obtain \( \wcnb\) by setting \(\beta = 1\) and \(\delta \to 0\) in \(\wcasl \) due to \eqref{eq:local_bayesian} and \eqref{eq:asl_local_update}. Nevertheless, we first present the analysis for NBSL since it is easier to derive and provides useful insights for ASL. Subsequently, we provide the results for ASL with proofs deferred to the appendix. 
\subsection{Non-Bayesian Social Learning}\label{sec:main_results_nbsl}
 The intervention \( \textit{do}(\bmu_{1,i} := \mu_1 ) \) ceases (or obstructs the use of) all incoming information at agent $1$ from the neighbors \(\mathcal{N}_1\) and the use of the streaming observations \(\bxi_{1,i}\) from the environment. Consequently, we can model this effect by redefining the combination matrix and the LLR vector counterparts under the intervention:
\begin{align}\label{eq:effective_definitions}
\widetilde{A} \triangleq [\widetilde{a}_{\ell k}], \quad \widetilde{a}_{\ell k} \triangleq
    \begin{dcases}
    1, & \quad \ell=k=1 \\ 0,& \quad \ell\neq k=1\\a_{\ell k},& \quad \ell\neq 1,k\neq 1
    \end{dcases}, \quad \wbL_{i}(\theta) \triangleq [0 ,\bL_{2,i}(\theta), \dots, \bL_{K,i}(\theta)]^{\T}.
\end{align}
Observe that the effective combination matrix \( \widetilde{A} \) can be obtained from $A$ as follows:
\begin{align}\label{eq:effective_definitions_block}
    A= \left [ \begin{array}{c|ccc}
   a_{11}  &  & r^{\T} &  \\
   \hline 
   a_{21} &  &  &  \\
   \vdots &  & R & \\
   a_{K1} &  &  & \\
\end{array} \right ] \Longrightarrow \widetilde{A} = \left [ \begin{array}{c|ccc}
   1  & & r^{\T} & \\
   \hline 
   \\
   0  & & R & \\
   \\
\end{array} \right ],
\end{align}
for a \((K-1)\times 1 \) dimensional vector $r$ and a \((K-1)\times(K-1)\) dimensional matrix $R$: \begin{equation}
    r \triangleq [a_{12} \quad a_{13} \quad \dots \quad a_{1K}]^{\T},  \qquad    R \triangleq \left [ \begin{array}{ccc}
   a_{22} & \dots  &  a_{2K}  \\
   \vdots & \ddots & \vdots   \\
   a_{K2} & \dots & a_{KK}  \\
\end{array} \right ] .
\end{equation}
The matrix structure of $\widetilde{A}$ belongs to the class of \textit{reducible} combination matrices \citep{resnick2013adventures,smith2014} that arise in the analysis of \emph{weakly} connected networks \citep{molavi2013reaching,mossel2015strategic,ying2016information}. As opposed to the strongly connected networks where information can flow thoroughly, in weakly connected networks information can flow only in one direction between certain parts of the network. In the current context, this corresponds to the fact that information continues to flow from agent $m$ in the form of its belief vector fixed at $\mu_m$, but no information is flowing into it in the sense that agent $m$ ignores all signals arriving from its neighbors and does not use them to update its local belief. However, in contrast to these prior works that analyze opinion dynamics under weakly connected networks, we are interested in the effect of the intervention on the original \emph{strongly} connected network. This alters the LLR $\bL_i(\theta)$ as well --- see \eqref{eq:effective_definitions}. Similar to the original case in \eqref{eq:LLF_evolution}, we proceed by studying the log-belief ratio evolution that results from using $\widetilde{A}$:
\begin{equation}\label{eq:llf_evolution_intervention}
       \wbS_{i}(\theta) =  \widetilde{A}^{\T}(\wbS_{i-1}(\theta) + \wbL_{i}(\theta)).
\end{equation}
Recursive application of \eqref{eq:llf_evolution_intervention} across \(i\) iterations to the log-belief ratio \(\wbS_{i}(\theta)\) yields
\begin{equation}\label{eq:llr_slln}
     \wbS_{i}(\theta) =   \sum_{j = 1}^{i}  (\widetilde{A}^{i-j+1})^{\T}\wbL_{j}(\theta) +   (\widetilde{A}^i)^{\T}\wbS_{0}(\theta).
\end{equation}
 To study \eqref{eq:llr_slln}, we need to evaluate the powers of the effective combination matrix:
\begin{align}\label{eq:matrix_powers}
      \widetilde{A}^i = \left [ \begin{array}{c|ccc}
   1  & & {r_i^\prime}^{\T} & \\
   \hline 
   \\
   0  & & R^i & \\
   \\
\end{array} \right ], \qquad   \widetilde{A}^\infty \stackrel{(a)}{=} \left [ \begin{array}{c|ccc}
   1  & 1& \dots & 1 \\
   \hline 
   \\
   0  & & 0 & \\
   \\
\end{array} \right ]
\end{align}
where \((a)\) follows from the fact that the spectral radius of the matrix \(R\) is strictly smaller than 1, i.e., \( R \) is a stable matrix \cite[Lemma 1]{ying2016information}. For each time $i$ and $0<j\leq i$, observe that
\begin{subequations}
    \begin{equation}
     (\widetilde{A}^{i-j+1})^{\T}\wbL_{j}(\theta) \stackrel{\eqref{eq:effective_definitions},\eqref{eq:matrix_powers}}{=}  \left [ \begin{array}{c|ccc}
   1  & & 0 & \\
   \hline 
   \\
   {r_{i-j+1}^\prime}  & & {R^{i-j+1}}^{\T} & \\
   \\
\end{array} \right ] \begin{bmatrix}
0 \\
\bL_{2,j}(\theta) \\
\vdots \\
\bL_{K,j}(\theta)
\end{bmatrix}  
\end{equation}
\begin{equation}\label{eq:verification_exp_1}
   \Longrightarrow [(\widetilde{A}^{i-j+1})^{\T}\wbL_{j}(\theta)]_1  \:\: = \:\: 0
\end{equation}
\end{subequations}
where $r_i^\prime$ denotes the accumulated upper-right block in the power $\widetilde{A}^i$, and
\begin{subequations}
\begin{equation}
    (\widetilde{A}^i)^{\T}\wbS_{0}(\theta) \stackrel{\eqref{eq:effective_definitions},\eqref{eq:matrix_powers}}{=} \left [ \begin{array}{c|ccc}
   1  & & 0 & \\
   \hline 
   \\
    {r_i^\prime}  & & {R^i}^{\T} & \\
   \\
\end{array} \right ] \begin{bmatrix}
\log\dfrac{\mu_{1}(\theta^{\circ})}{\mu_{1}(\theta)}. \\
\log\dfrac{\bmu_{2,0}(\theta^{\circ})}{\bmu_{2,0}(\theta)}. \\
\vdots \\
\log\dfrac{\bmu_{K,0}(\theta^{\circ})}{\bmu_{K,0}(\theta)}. 
\end{bmatrix}  
\end{equation}\begin{equation}\label{eq:verification_exp_2}
    \Longrightarrow [(\widetilde{A}^i)^{\T}\wbS_{0}(\theta)]_1  \:\: = \:\:  \log\dfrac{\mu_{1}(\theta^{\circ})}{\mu_{1}(\theta)}.
\end{equation}  
\end{subequations}
Inserting these into \eqref{eq:llr_slln} verifies that the intervention is fixed as intended for all time instants, since
\begin{subequations}
    \begin{equation}
    [\wbS_{i}(\theta)]_1 \stackrel{\eqref{eq:llr_slln}}{=}   \sum_{j = 1}^{i}  [(\widetilde{A}^{i-j+1})^{\T}\wbL_{j}(\theta)]_1 +   [(\widetilde{A}^i)^{\T}\wbS_{0}(\theta)]_1. 
    \end{equation}
    \begin{equation}\label{eq:fixed_intervention_nbsl}
     \wbS_{1,i}(\theta) \stackrel{\eqref{eq:verification_exp_1},\eqref{eq:verification_exp_2}}{=} \log\frac{\mu_{1}(\theta^{\circ})}{\mu_{1}(\theta)}, \quad  \wbmu_{1,i}(\theta) = \mu_1(\theta).
    \end{equation}
\end{subequations}
Moreover, if we take the expectation of both sides of \eqref{eq:llr_slln}, we get
\begin{align}
    \e [\wbS_{i}(\theta)] &=  \sum_{j = 1}^{i}  (\widetilde{A}^{i-j+1})^{\T} \e[\wbL_{j}(\theta)] +   (\widetilde{A}^i)^{\T}\e[\wbS_{0}(\theta)] \notag \\
    &= \sum_{j = 1}^{i}  (\widetilde{A}^{i-j+1})^{\T} \widetilde{d}(\theta) +   (\widetilde{A}^i)^{\T}\e[\wbS_{0}(\theta)].
\end{align}
where we use the definition $\widetilde{d}(\theta) \triangleq [0, d_2(\theta), \dots , d_K(\theta)]^{\T}$. Hence, in the limit (the existence is guaranteed by the finiteness of LLRs and positive initial beliefs), it holds that
\begin{align}\label{eq:expl_limit_exp}
    \lim_{i \to \infty} \e [\wbS_{i}(\theta)] &= \lim_{i \to \infty} \sum_{j = 1}^{i}  (\widetilde{A}^{i-j+1})^{\T} \widetilde{d}(\theta) +   (\widetilde{A}^{\infty})^{\T}\e[\wbS_{0}(\theta)] \notag \\
    &= \sum_{j = 1}^{\infty}  (\widetilde{A}^{j})^{\T} \widetilde{d}(\theta) +   (\widetilde{A}^{\infty})^{\T}\e[\wbS_{0}(\theta)].
\end{align}
If we incorporate \eqref{eq:matrix_powers} into \eqref{eq:expl_limit_exp}, this implies for the log-belief ratios of all agents except agent \(m =1\) that
\begin{align}\label{eq:general_lambda_temp}
    \wS_{-m,\infty}(\theta) \triangleq [\wS_{2,\infty}(\theta), \dots, \wS_{K,\infty}(\theta)]^{\T} = \sum_{j = 1}^{\infty}  ({R}^{j})^{\T} d_{-m}(\theta) + \Big (\log\frac{\mu_{m}(\theta^{\circ})}{\mu_{m}(\theta)} \Big) \mathds{1}_{K-1}
\end{align}
where \( d_{-m}(\theta) \) is the \((K-1)\times1 \) dimensional vector of local KL divergences of the remaining agents, i.e., \( d_{-m}(\theta)~\triangleq~\text{col} \{ d_{\ell}(\theta) \}_{\ell=2}^{K} \). Since \(R\) is a stable matrix \cite[Lemma 1]{ying2016information}, Eq.~\eqref{eq:general_lambda_temp} can alternatively be written as
\begin{align}\label{eq:general_lambda}
    \boxed{\wS_{-m,\infty}(\theta) =\big ((I-R^{\T})^{-1}-I \big ) d_{-m}(\theta) +  \Big (\log\frac{\mu_{m}(\theta^{\circ})}{\mu_{m}(\theta)} \Big) \mathds{1}_{K-1}}
\end{align}
The causal influence of agent \( m=1 \) on agent \( k \) can now be calculated by inserting \( \wS_{k,\infty}(\theta) \) into \eqref{eq:expected_log_belief_trans} to find $\widetilde{\mu}_{k,\infty}(\theta)$, which is the input for \eqref{eq:cmknb_def} that yields $\wcnb$. More specifically, if we incorporate \eqref{eq:general_lambda} into \eqref{eq:expected_log_belief_trans}, we get
\begin{equation}
    \widetilde{\mu}_{k,\infty}(\theta^{\circ}) = \dfrac{1}{1 + \!\sum\limits_{\theta \in \Theta \setminus \{ \theta^\circ\}} \dfrac{\mu_{m}(\theta)}{\mu_{m}(\theta^\circ)}  \exp \Big \{ - \Big [\big ((I-R^{\T})^{-1}-I \big ) d_{-m}(\theta) \Big]_k \Big \}}
\end{equation}
which, by \eqref{eq:cmknb_def}, implies
\begin{equation}\label{eq:general_lambda_transf}
         \wcnb = 1-\dfrac{1}{1 + \!\sum\limits_{\theta \in \Theta \setminus \{ \theta^\circ\}} \dfrac{\mu_{m}(\theta)}{\mu_{m}(\theta^\circ)}  \exp \Big \{ - \Big [\big ((I-R^{\T})^{-1}-I \big ) d_{-m}(\theta) \Big]_k \Big \}}.
\end{equation}
Equation \eqref{eq:general_lambda} is a general result which shows that \(\wcnb\) is a function of \((i)\) the combination weights (via $R$), and \((ii)\) the individual informativeness of each agent (via $d_{-m}(\theta)$). 
\begin{myremark}[Generalization to sub-networks]
Note that expressions \eqref{eq:general_lambda}-\eqref{eq:general_lambda_transf} can generalize to the influence of a sub-network with multiple agents rather than an individual agent $m$. In this case, upon intervening on all agents within the sub-network, the effective combination matrix becomes
\begin{align}
    \widetilde{A} = \left [ \begin{array}{c|ccc}
   I  & & r^{\T} & \\
   \hline 
   \\
   0  & & R & \\
   \\
\end{array} \right ]
\end{align}
where the first entry in \eqref{eq:effective_definitions_block} is replaced by an identity matrix. Namely, if the sub-network under consideration has size $s$, the identity, $r$ and $R$ matrices would be of dimensions $s \times s$, $(K-s) \times s$, and $(K-s) \times (K-s)$, respectively. Similarly, \( d_{-m}(\theta) \) can be replaced with local KL divergences of the agents that do not belong to the treated sub-network. \qed
\end{myremark}
\begin{myremark}[Finite-time spread of influence] Expressions \eqref{eq:general_lambda}-\eqref{eq:general_lambda_transf} reveal how the total \textit{overall} effects depend on direct \textit{instantaneous} effects in temporal networked interactions. For the spread of direct effects in a finite time instant $n$, we can modify \eqref{eq:general_lambda_temp} and \eqref{eq:general_lambda} as
\begin{align}
\wS_{-m,n}(\theta) \triangleq [\wS_{2,n}(\theta), \dots, \wS_{K,n}(\theta)]^{\T} = \sum_{j = 1}^{n}  ({R}^{j})^{\T} d_{-m}(\theta) + \Big (\log\frac{\mu_{m}(\theta^{\circ})}{\mu_{m}(\theta)} \Big) \mathds{1}_{K-1}
\end{align}
due to \eqref{eq:llr_slln}, and consequently,
\begin{align}
    \wS_{-m,n}(\theta) = \left((I-R^{\T})^{-1} (I-R^{n+1 \T}) - I \right) d_{-m}(\theta) + \Big (\log\frac{\mu_{m}(\theta^{\circ})}{\mu_{m}(\theta)} \Big) \mathds{1}_{K-1}. 
\end{align}\qed
\end{myremark}
For ease of the presentation, we continue to describe the total causal effects of a single agent, even though our results can be extended to sub-network influence and finite time analysis in a straightforward manner as discussed in the above remarks. 

In \eqref{eq:general_lambda_transf}, \(\wcnb\) represents a dose-response curve, assuming different values for different intervention strengths (i.e., dose) \(\mu_m\). This is a typical situation in the context of continuous-valued interventions. In some applications, however, it proves beneficial to encapsulate the causal effect value with a single number. For this purpose, we may set $\mu_m$ to be uniform across all hypotheses, i.e., $\mu_m(\theta) = 1/H, \forall \theta \in \Theta$. This method of summarizing the causal effect is denoted as follows:
\begin{equation}\label{eq:wcnbd_definition}
    \wcnbd \triangleq \wcnb \Big |_{\mu_m(\theta) = 1/H} 
\end{equation}
In Appendix A, we show that this choice effectively parallels the process of determining the average causal \emph{derivative} effect \cite[Chap. 6]{peters2017elements}, which is a method commonly used in the literature for summarizing the causal effect. Basically, it quantifies the extent of change in agent $k$ in response to an infinitesimal variation in the intervention strength $\mu_m$.

In the next section, we study two special network topologies that help illustrate the dependencies of the causal effect more explicitly.

\subsubsection{Special cases}

\paragraph{Fully-connected and federated architectures.} In this example, we consider a fully-connected network (see Fig.~\ref{fig:fully_connected}) with a rank-one combination matrix and Perron vector \( v\), i.e.,
\begin{align}\label{eq:fully_connected_a}
    A = \left [ \begin{array}{cccc}
   v_{1}  & v_{1} &   \cdots  & v_{1} \\
   v_{2}  & v_{2}   &  \cdots  & v_{2} \\
   \vdots  & \vdots   &  \ddots  & \vdots \\
   v_{K} & v_{K}   &  \cdots  & v_{K} \\
\end{array} \right ] = v\mathds{1}_{K}^{\T}.
\end{align}
Note that in terms of performance, this system is equivalent to a federated architecture in which \( (i) \) agents send their beliefs to a fusion center after local adaptation, \((ii)\) the center averages the received beliefs in a weighted manner, and \((iii)\) then broadcasts the combined belief to all agents --- see Fig.~\ref{fig:federated}.
\begin{figure}
\centering
\begin{minipage}{.5\textwidth}
  \centering
  \includegraphics[width=.9\linewidth]{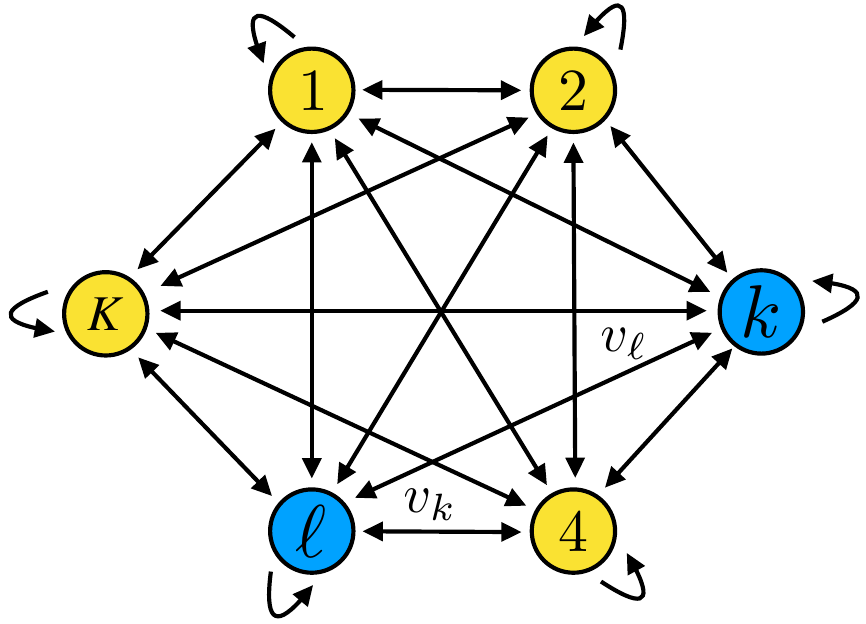}
  \caption{A fully connected network topology.}
  \label{fig:fully_connected}
\end{minipage}%
\hspace{0.2em}
\begin{minipage}{.47\textwidth}
  \centering
  \includegraphics[width=.97\linewidth]{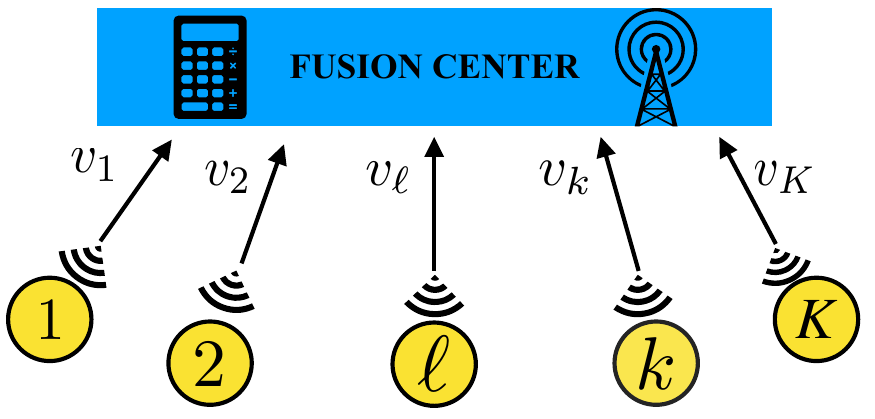}
  \caption{A federated architecture. The server broadcasts the weighted average of agents' beliefs back to them at each iteration. In terms of performance, this system is equivalent to the fully connected architecture of Fig.~\ref{fig:fully_connected}.}
  \label{fig:federated}
\end{minipage}
\end{figure}
Under intervention \( \textit{do}(\bmu_{1,i} := \mu_1 ) \), we have
\begin{align}\label{eq:fully_connected_intervention_a}
    \widetilde{A} = \left [ \begin{array}{c|ccc}
   1  & v_{1} &   \cdots  & v_{1} \\
   \hline 
   0  & v_{2}   &  \cdots  & v_{2} \\
   \vdots  & \vdots   &  \ddots  & \vdots \\
   0 & v_{K}  &  \cdots  & v_{K} \\
\end{array} \right ] \Longrightarrow   R = \left [ \begin{array}{ccc}
      v_{2}   &  \cdots  & v_{2} \\
      \vdots   &  \ddots  & \vdots \\
     v_{K}  &  \cdots  & v_{K} 
\end{array} \right ] = v_{-m}\mathds{1}_{K-1}^{\T}
\end{align}
where \( v_{-m}~\triangleq~\text{col} \{ v_{\ell}\}_{\ell=2}^{K}  \) is a \( (K-1) \times 1\) dimensional vector consisting of all Perron entries except for agent \(m=1\). Observe that
\begin{equation}
    R^2=v_{-m}\mathds{1}_{K-1}^{\T}v_{-m}\mathds{1}_{K-1}^{\T} \stackrel{(a)}{=}(1-v_1)v_{-m}\mathds{1}_{K-1}^{\T},
\end{equation}
where \( (a) \) follows from the fact that \(\mathds{1}_{K}^{\T}v=1\) (Eq. \eqref{eq:perron_def}). Repeating the same arguments, it holds that
\begin{equation}\label{eq:R_i_expression_fully}
    R^i = (1-v_1)^{i-1}v_{-m}\mathds{1}_{K-1}^{\T}.
\end{equation}
Therefore, 
\begin{equation}
    (I-R^{\T})^{-1}-I = \sum_{i=1}^{\infty} (R^{\T})^i = \mathds{1}_{K-1}v_{-m}^{\T} \sum_{i=1}^{\infty} (1-v_1)^{i-1} = \frac{1}{v_1} \mathds{1}_{K-1}v_{-m}^{\T}.
\end{equation}
Inserting this into \eqref{eq:general_lambda}, we arrive at the following expression for each agent \(k \neq m\):
\begin{equation}\label{eq:perron_sc_lambda}
    \boxed{\wS_{k,\infty}(\theta) = \frac{1}{v_m} \sum_{\ell=2}^K v_{\ell} d_{\ell} (\theta) +  \log \frac{\mu_{m}(\theta^{\circ})}{\mu_{m}(\theta)}}
\end{equation}
Combining \eqref{eq:expected_log_belief_trans} and \eqref{eq:cmknb_def} with \eqref{eq:perron_sc_lambda} yields the causal effect:
\begin{align}\label{eq:perron_sc_cmknb}
    \wcnb \stackrel{\eqref{eq:expected_log_belief_trans}, \eqref{eq:cmknb_def}}{=} &1-\dfrac{1}{1 + \!\sum_{\theta \in \Theta \setminus \{ \theta^\circ\}}   \exp \{ - \widetilde{\lambda}_{k,\infty}(\theta) \}} \notag \\
     \stackrel{\eqref{eq:perron_sc_lambda}}{=} \:\:\: & 1-\dfrac{1}{1 + \!\sum\limits_{\theta \in \Theta \setminus \{ \theta^\circ\}} \dfrac{\mu_{m}(\theta)}{\mu_{m}(\theta^\circ)}  \exp \Big \{ - \dfrac{1}{v_m} \sum\limits_{\ell \neq m} v_{\ell} d_{\ell} (\theta) \Big \}}.
\end{align}
The effect of agent \(m\) on all other agents is the same, which is expected due to the symmetric nature of this special example. Furthermore, observe that the causal effect \( \wcnb \) decreases with increasing \( \wS_{k,\infty}(\theta)\). On that account, from \eqref{eq:perron_sc_lambda} and \eqref{eq:perron_sc_cmknb} it can be seen that:
\begin{itemize}
    \item Increasing the network centrality of agent \(m=1\) (i.e., increasing $v_m$) decreases \(\wS_{k,\infty}(\theta)\), and in turn increases the causal effect \( \wcnb\). Therefore, an agent has more effect on other agents if it has a higher network centrality. In particular, if
    \begin{equation}
        v_m \to 0 \Longrightarrow \wS_{k,\infty}(\theta) \to +\infty \Longrightarrow \wcnb \to 0,
    \end{equation}
    which means that an agent with negligible network centrality has no causal effect on other agents.
    \item Increasing network centrality and informativeness of the other agents $\ell \neq m$ (i.e., increasing $v_{\ell}$ and $d_{\ell}(\theta)$) increases \( \wS_{k,\infty}(\theta)\), and in turn decreases the causal effect \( \wcnb \). In particular, if the most informative agents are equipped with the highest network centrality, then $\sum_{\ell=2}^K v_{\ell} d_{\ell} (\theta)$ is large and it is harder for agent \(m=1\) to control other agents.
    \item If the fixed belief on the true hypothesis \( \mu_m (\theta^\circ)\) decreases, then \( \wS_{k,\infty}(\theta)\) decreases and the causal effect \( \wcnb \) increases. This suggests that the further from the truth the information an agent supplies, the more effect that agent will have on other agents. In other words, agents supplying misinformation have more effect on the rest of the network. Specifically, observe that if the rest of the agents have a low informativeness average, i.e., if
    \begin{equation}
        \sum_{\ell=2}^K v_{\ell} d_{\ell} (\theta) \approx 0  \Longrightarrow \wcnb \approx 1-\mu_m(\theta^\circ).
    \end{equation}
    Therefore, the causal effect is proportional to the difference from the truth. It is maximized (i.e., $\wcnb = 1$) when the fixed belief assigns 0 to the true hypothesis.
\end{itemize}

\begin{wrapfigure}{l}{6.5cm}
\includegraphics[width=6.5cm]{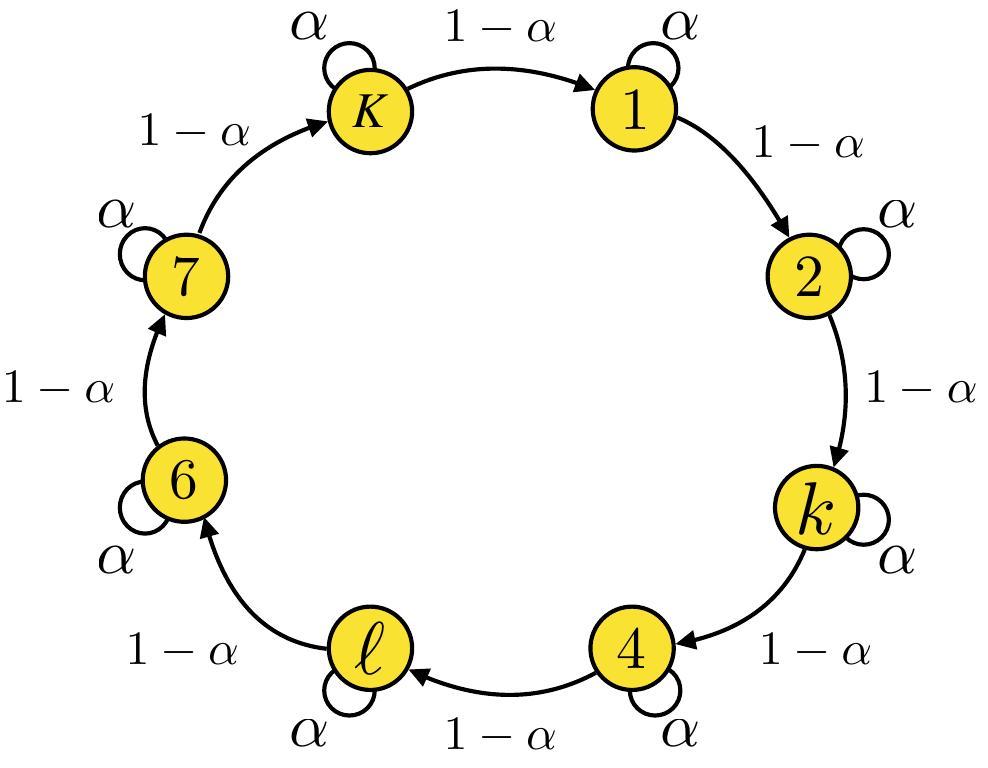}
\vspace{-1.5em}\caption{A unidirectional ring.}\label{fig:ring}
\end{wrapfigure} 
\paragraph{Ring architecture.} In this example, we consider a unidirectional ring network where each agent has a self-confidence of \(\alpha\), and assigns a confidence of \(1-\alpha\) to the preceding agent in the ring --- see Fig.~\ref{fig:ring}. Agents are indexed such that agent \(k+1\) receives (or uses) information from agent \(k\) only. The combination matrix has the form:

\begin{align}\label{eq:ring_A_nbsl}
    A = \left [ \begin{array}{ccccc}
   \alpha  & 1-\alpha & 0 &  \cdots  & 0 \\
   0  & \alpha & 1-\alpha   &  \cdots  & 0 \\
   \vdots  & 0 & \alpha   &  \cdots  & \vdots \\
   0  & \vdots & \vdots  &  \ddots  & 1-\alpha \\
    1-\alpha & 0 & 0  &  \cdots  & \alpha \\
\end{array} \right ].
\end{align}
Under intervention \( \textit{do}(\bmu_{1,i} := \mu_1 ) \), we have
\begin{align}\label{eq:ring_R_nbsl}
    \widetilde{A} = \left [ \begin{array}{c|cccc}
   1  & 1-\alpha & 0 &  \cdots  & 0 \\
   \hline 
   0  & \alpha & 1-\alpha   &  \cdots  & 0 \\
   0  & 0 & \alpha   &  \cdots  & \vdots \\
   \vdots  & \vdots & \vdots  &  \ddots  & 1-\alpha \\
   0 & 0 & 0  &  \cdots  & \alpha \\
\end{array} \right ] \Longrightarrow R = \left [ \begin{array}{cccc}
    \alpha & 1-\alpha   &  \cdots  & 0 \\
    0 & \alpha   &  \cdots  & \vdots \\
    \vdots & \vdots  &  \ddots  & 1-\alpha \\
    0 & 0  &  \cdots  & \alpha \\
\end{array} \right ].
\end{align}
As a result,
\begin{align}
    (I-R^{\T})= (1-\alpha)\left [ \begin{array}{cccc}
    1 & 0   &  \cdots  & 0 \\
    -1 & 1   &  \cdots  & 0 \\
    \vdots & \vdots  &  \ddots  & \vdots \\
    0 & \cdots  &  -1  & 1 \\
\end{array} \right ] \Longrightarrow (I-R^{\T})^{-1}= \frac{1}{1-\alpha} \left [ \begin{array}{cccc}
    1 & 0   &  \cdots  & 0 \\
    1 & 1   &  \cdots  & 0 \\
    \vdots & \vdots  &  \ddots  & \vdots \\
    1 & 1  &  \cdots  & 1 \\
\end{array} \right ],
\end{align}
which implies
\begin{align}
    ((I-R^{\T})^{-1}-I)d_{-m}(\theta)= \frac{1}{1-\alpha} \Big [\alpha d_2(\theta), d_2(\theta)+\alpha d_3(\theta), \dots, \sum_{\ell=2}^{K-1} d_{\ell}(\theta)+\alpha d_{K}(\theta) \Big ]^{\T}.
\end{align}
Consequently, for agent $k$,
\begin{align}\label{eq:ring_sc_nb_lambda}
    \boxed{\wS_{k,\infty}(\theta) = \frac{1}{1-\alpha}\sum_{\ell=2}^{k-1} d_{\ell} (\theta) + \frac{\alpha}{1-\alpha} d_{k} (\theta)  + \log\frac{\mu_{m}(\theta^{\circ})}{\mu_{m}(\theta)}}
\end{align}
and, in addition, by definitions \eqref{eq:expected_log_belief_trans} and \eqref{eq:cmknb_def},
\begin{equation}\label{eq:ring_sc_nb_cmknb}
    \wcnb  = 1-\dfrac{1}{1 + \!\sum\limits_{\theta \in \Theta \setminus \{ \theta^\circ\}} \dfrac{\mu_{m}(\theta)}{\mu_{m}(\theta^\circ)}  \exp \Big \{ - \dfrac{1}{1-\alpha}\sum\limits_{\ell=2}^{k-1} d_{\ell} (\theta) - \dfrac{\alpha}{1-\alpha} d_{k} (\theta)  \Big \}}.
\end{equation}
As stated before, the causal effect $\wcnb$ decreases with increasing $\wS_{k,\infty}(\theta)$. Therefore, the following remarks for \eqref{eq:ring_sc_nb_lambda} and \eqref{eq:ring_sc_nb_cmknb} are in place:
\begin{itemize}
    \item  Since the KL divergence \(d_{\ell}(\theta)\) is non-negative, $\wS_{k,\infty}(\theta)$ is monotonically increasing along the path \(k=2~\to~\dots~\to~k=K\). Therefore, the causal effect of agent \(m=1\) is monotonically decreasing along the same path: the closer agent $m$ is to an agent, the higher its effect on that agent. This is intuitive because the effect that agent \(m\) has on agent \(k+1\) is transferred via agent \( k \) in the ring structure. The difference between the causal effects of agent \(m\) on agents $k$ and $k+1$ is proportional to the increase in $\wS_{k,\infty}(\theta)$, that is, 
    \begin{equation}
        \wS_{k+1,\infty}(\theta)-\wS_{k,\infty}(\theta) = d_k(\theta) + \frac{\alpha}{1-\alpha} d_{k+1} (\theta).
    \end{equation}
    This means that informative agents with high KL divergence on the path between agent \(m\) and agent \(k\) reduce the causal effect \(\wcnb\). In other words, the sphere of influence of an agent \(m\) is bigger if there are no other informative agents in the vicinity.
    \item For the immediate follower of agent \(m=1\), it follows that
    \begin{equation}
            \wS_{2,\infty}(\theta) = \frac{\alpha}{1-\alpha} d_{2} (\theta)  + \log\frac{\mu_{1}(\theta^{\circ})}{\mu_{1}(\theta)}.
    \end{equation}
    If agent \(2\) is not sufficiently informative itself, i.e., \(d_2(\theta)\) is small, then \(\wS_{2,\infty}(\theta)\) gets smaller and \( C_{1 \to 2}^{\textup{NB}}\) gets higher. In other words, an agent is more controllable if it is not knowledgeable.
    \item The limiting average \(\wS_{k,\infty}(\theta)\) increases with increasing \( \alpha\). Therefore, if agents are more self-confident, the causal strength is smaller and agents are less controllable by other agents.
\end{itemize}

\subsection{Adaptive Social Learning}\label{sec:main_results_asl}
Similar to the modification in the NBSL case, the log-belief recursion \eqref{eq:LLF_evolution_asl} in ASL is modified as follows under intervention \( \textit{do}(\bmu_{1,i} := \mu_1 ) \):
\begin{equation}\label{eq:llf_evolution_intervention_asl}
       \wbS_{i}(\theta) =  \widetilde{A}^{\T}((1-\delta)\wbS_{i-1}(\theta) +  \beta \wbL_{i}(\theta)).
\end{equation}
The effective combination matrix \(\widetilde{A}\) continues to be given by \eqref{eq:effective_definitions}. However, the effective LLR is now given by
\begin{equation}\label{eq:llr_asl_intervention}
    \wbL_{i}(\theta) \triangleq \Big [\frac{\delta}{\beta}\log\frac{\mu_{1}(\theta^{\circ})}{\mu_{1}(\theta)} ,\bL_{2,i}(\theta), \dots, \bL_{K,i}(\theta)\Big ]^{\T},
\end{equation}
where the first entry is different than the NBSL case. This is to compensate for the presence of the parameters \( \delta\) and \(\beta\). Observe from \eqref{eq:llf_evolution_intervention_asl} that  \(\widetilde{A}\) from \eqref{eq:effective_definitions} and \(\wbL_{i}(\theta)\) from \eqref{eq:llr_asl_intervention} verify \( \wbS_{1,i}(\theta) = \log\frac{\mu_{1}(\theta^{\circ})}{\mu_{1}(\theta)} \) for all time instants, i.e., the intervention is fixed, by similar arguments to \eqref{eq:fixed_intervention_nbsl}. In Appendix~\ref{appendix:asl_general}, we derive the following expression for the limiting log-belief ratio expectations for the rest of the network:
\begin{equation}\label{eq:lambda_m_infty_asl}
   \boxed{ \widetilde{\lambda}_{-m,\infty}(\theta)\!\!=\!\Big (\! \log\frac{\mu_{m}(\theta^{\circ})}{\mu_{m}(\theta)} \Big )(I-(1-\delta)R^{\T})^{-1} r + \frac{\beta}{1-\delta} \Big ((I-(1-\delta)R^{\T})^{-1}\!-\!I\! \Big)  d_{-m} (\theta)}
\end{equation}
The causal effect \(\wcasl\) can be calculated by inserting post-intervention expression \eqref{eq:lambda_m_infty_asl} and pre-intervention expression \eqref{eq:expected_logbelief_asl_idle} into the definitions \eqref{eq:expected_log_belief_trans} and \eqref{eq:cmkasl_def}. Notice from \eqref{eq:lambda_m_infty_asl} that similar to the NBSL case, the causal effect depends on the informativeness of agents, the network topology, and the strength of intervention via \(d_{-m}(\theta), R,\) and \(\mu_m\), respectively. In fact, if we set \(\delta = 0 \) and \(\beta = 1\), Eq. \eqref{eq:lambda_m_infty_asl} reduces to the NBSL expression \eqref{eq:general_lambda} as expected. This is because the left-stochastic nature of $A$ implies that
\begin{equation}
       r + R^{\T} \mathds{1}_{K-1} = \mathds{1}_{K-1} \Longleftrightarrow (I-R^{\T})^{-1} r = \mathds{1}_{K-1}.
\end{equation}
In addition, the causal effects in ASL are affected by the importance weighting parameters \(\delta\) and \(\beta\), as well as by the vector \(r\) that represents the confidence weights other agents assign to agent \(m\). In \eqref{eq:general_lambda}, the entries of \(r\) implicitly influence the causal effect via \(R\): the column-wise summation of the entries of \(r\) and \(R\) results in 1 at all columns due to the left-stochastic nature of \(A\). In comparison, in ASL, both \(r\) and \(R\) impact \(\wcasl\) explicitly. If we take a closer look at the terms in \eqref{eq:lambda_m_infty_asl}, we can see that:
\begin{itemize}
    \item The vector that scales the intervened log-belief ratio \( \log\frac{\mu_{m}(\theta^{\circ})}{\mu_{m}(\theta)}\) can be expanded as 
    \begin{equation}
     (I-(1-\delta)R^{\T})^{-1} r = r+(1-\delta)R^{\T}r+ (1-\delta)^2R^{2\T}r+ (1-\delta)^3R^{3\T}r +\dots  
    \end{equation}
    On the RHS of this equation, the first $r$ represents the scaling of the information transferred from agent \(m\) to the rest of the network \emph{directly}. Namely, for an agent \(k\), the scaling of the direct information is \(a_{mk}\) if \(m\) is an immediate neighbor (\(m \in \mathcal{N}_k\)); 0 if it is not (\(m \notin \mathcal{N}_k\)). On the other hand, the second term $(1-\delta)R^{\T}r$ describes the scaling of the information transferred from agent \(m\) to the rest of the network, which is then mixed with the other agents \(\forall k \neq m\) (via \(R^{\T}\)) and ``forgotten'' (i.e., lose its recency) for one time instant by a factor of \((1-\delta\)). The consecutive terms over time follow from the same scaling argument.
    \item In a similar manner, we can express the matrix that scales the vector of individual informativeness in the rest of the network \(d_{-m}(\theta)\) as: 
    \begin{equation}
    \frac{1}{1-\delta} \Big ((I-(1-\delta)R^{\T})^{-1}\!-\!I \Big ) = R^{\T}+ (1-\delta)R^{2\T}+(1-\delta)^2R^{3\T}+\dots 
    \end{equation}
    Since there is no outgoing link from the rest of the network to agent \(m=1\) under the intervention, the terms in this expression only depend on the combination matrix \(R\). When new information arrives to the remaining agents, it is first mixed among these agents (corresponding to the first term \(R^{\T}\) on RHS), and then in the next iteration, it is mixed again but also gets forgotten due to the time discount factor \(\delta\) (corresponding to the second term \((1-\delta)R^{2\T}\) on RHS), and so on.
    \item   Remember from \eqref{eq:asl_local_update} that \(\beta>0\) scales the likelihood of observations, reflecting the weight agents place on their own observations originating from out-of-network sources. As a result, notice that in \eqref{eq:lambda_m_infty_asl}, \(\beta\) scales the individual informativeness \(d_{\ell}(\theta), \forall \ell \neq m\). In other words, it amplifies the effect of self observations on the state of nature.
\end{itemize}
Next, we analyze the special cases introduced in the NBSL case under ASL framework.

\subsubsection{Special cases}

\paragraph{Fully-connected and federated architectures.} 
In Appendix~\ref{appendix:asl_fully_connected}, we prove that the additional \(\delta\) and \(\beta\) parameters introduced for the ASL change the NBSL expression \eqref{eq:perron_sc_lambda} to 
\begin{equation}\label{eq:asl_lambda_fully_connected}
        \boxed{\wS_{k,\infty}(\theta) = \frac{1}{1-(1-\delta)(1-v_m)} \Bigg (\beta \sum_{\ell=2}^K v_{\ell} d_{\ell} (\theta) +  v_m \log \frac{\mu_{m}(\theta^{\circ})}{\mu_{m}(\theta)} \Bigg )}
\end{equation}
Notice that as \(\delta \to 0 \) and \(\beta \to 1\), \eqref{eq:asl_lambda_fully_connected} recovers \eqref{eq:perron_sc_lambda} as expected, and the following remarks from the NBSL case continue to hold here: \((i)\) the influence of an agent \(m=1\) is identical for each agent \(k \neq 1\) due to symmetry in the network topology, \((ii)\) increasing the network centrality of agent \(m\) increases its causal influence, and \((iii)\) increasing the network centrality and informativeness of the rest of the agents \(\ell \neq 1\) decreases the causal effect of \(m=1\). Moreover, since causal effect $\wcasl$ is a monotonic decreasing function of $\wS_{k,\infty}(\theta)$ by \eqref{eq:cmkasl_def}, Eq.~\eqref{eq:asl_lambda_fully_connected} also implies the following conclusions:

 \begin{itemize}
\item As stated after \eqref{eq:lambda_m_infty_asl}, \(\beta\) scales the informativeness of agents. Accordingly, \eqref{eq:asl_lambda_fully_connected} reveals that the causal effect is decreasing with increasing \(\beta\). This is justifiable because as agents have greater reliance on their own observations about the state of nature, they are less influenced by other agents in the network. It is worth mentioning that in \eqref{eq:asl_lambda_fully_connected}, the intervened log-belief ratio \(\log \frac{\mu_{m}(\theta^{\circ})}{\mu_{m}(\theta)}\) behaves as ``pseudo-informativeness''. It is scaled with the Perron entry \(v_m\) of agent \(m=1\), similar to how the rest of the agents' informativeness is scaled with their own Perron entries. The difference is that the other agents' informativeness are based on their log-likelihood ratios averaged with respect to the true distribution, whereas the intervened log-belief ratio can be arbitrary, possibly supplying misinformation.
 
 \item In the special case of the rest of the agents having no informativeness (i.e., \(d_{\ell}(\theta)=0, \: \forall \ell \neq m\)), the limiting mean log-belief ratio vector in \eqref{eq:asl_lambda_fully_connected} turns into
 \begin{equation}\label{eq:asl_fully_connected_zero_inform}
      \wS_{k,\infty}(\theta) = \frac{1}{1+\delta \Big (\dfrac{1}{v_m} -1 \Big )}  \log \frac{\mu_{m}(\theta^{\circ})}{\mu_{m}(\theta)}. 
 \end{equation}
This is in contast to the NBSL case, where, \(\wS_{k,\infty}(\theta) =\log \frac{\mu_{m}(\theta^{\circ})}{\mu_{m}(\theta)}\). In other words, in steady-state of NBSL, the average beliefs of all agents become equal to the intervened fixed belief \(\mu_m\), implying full controllability. In ASL, however, the controllability is reduced by a factor of \((1+\delta(1/v_m-1)) \geq 1\) as shown in \eqref{eq:asl_fully_connected_zero_inform}. In particular, increasing the forgetting factor \(\delta\) decreases controllability, especially when the network centrality of ``controlling'' agent \(m=1\) is small. However, if agent \(m\) is highly central, i.e., \(v_m \to 1\), then the forgetting factor \(\delta\) has negligible effect on controllability.

\item Considering the general case \eqref{eq:asl_lambda_fully_connected}, note that unlike \(\beta\) which only affects non-intervened observations, \(\delta\) affects both intervened beliefs and non-intervened observations. Thus, to fully understand the impact of \(\delta\) on the overall causal effect, we must consider the exact values of the relevant parameters. 

\end{itemize}

\paragraph{Ring architecture.}
In Appendix~\ref{appendix:asl_ring}, we prove that the additional \(\delta\) and \(\beta\) parameters introduced for the ASL change the NBSL expression \eqref{eq:ring_sc_nb_lambda} to 
\begin{empheq}[box=\fbox]{align}\label{eq:asl_lambda_ring}
     \wS_{k,\infty}(\theta) = \Big (\!\log\frac{\mu_{m}(\theta^{\circ})}{\mu_{m}(\theta)} \Big ) \dfrac{(1-\delta)^{k-2}(1-\alpha)^{k-1}}{(1-(1-\delta)\alpha)^{k-1}} \! &+ \dfrac{\beta\alpha}{1-(1-\delta)\alpha}  d_k (\theta) \notag \\& + \frac{\beta}{1-\delta} \sum_{\ell=2}^{k-1} \dfrac{(1-\delta)^{k-\ell}(1-\alpha)^{k-\ell}}{(1-(1-\delta)\alpha)^{k-\ell+1}} d_{\ell} (\theta)
\end{empheq}
from which we can make the following observations:
\begin{itemize}
    \item As \(\delta \to 0\) and \(\beta \to 1\), the NBSL expression \eqref{eq:ring_sc_nb_lambda} for ring architectures is recovered. Similar to the earlier expressions, the causal effect decreases with increasing informativeness of agents along the path between \(m=1\) and \(k\) \((\ell=2, \dots, k-1)\) (the last term on RHS), and also decreases with increasing informativeness of agent \(k\). Informativeness is scaled by \(\beta\) as before.

    \item Recall that in NBSL, if the rest of the agents have no informativeness, it holds that \(\wS_{k,\infty}(\theta) =\log \frac{\mu_{m}(\theta^{\circ})}{\mu_{m}(\theta)}\). In other words, agent \(m=1\) can fully control other agents' beliefs. Instead, in ASL, if \(d_{\ell}(\theta)=0, \: \forall \ell \neq 1 \), it holds that
    \begin{align}
         \wS_{k,\infty}(\theta) &=  \dfrac{(1-\delta)^{k-2}(1-\alpha)^{k-1}}{(1-(1-\delta)\alpha)^{k-1}} \Big ( \log\frac{\mu_{m}(\theta^{\circ})}{\mu_{m}(\theta)} \Big ) \notag \\
         &= \dfrac{1}{1-\delta} \left (1- \dfrac{1}{\dfrac{1-\alpha}{\delta}+\alpha} \right )^{k-1}  \Big ( \log\frac{\mu_{m}(\theta^{\circ})}{\mu_{m}(\theta)} \Big ).
    \end{align}
     Observe that as the agent index \(k\) increases, the controllability decays at each hop by a factor of 
     \begin{equation}
         \dfrac{\wS_{k+1,\infty}(\theta)}{\wS_{k,\infty}(\theta)} =    1- \dfrac{1}{\dfrac{1-\alpha}{\delta}+\alpha} \quad \in [0,1].
     \end{equation}
    The decrease is higher when \(\delta\) is higher because the information from agent \(m=1\) gets ``partially forgotten'' at each hop as \(k\) (i.e., the distance to agent 1) increases. However, in general, the informativeness of agents along the path is not 0, and they have a shadowing effect on agent $m$'s influence, as argued before. The forgetting factor \(\delta\) decreases this shadowing effect as well, particularly for agents far from agent \(k\). 
\end{itemize}

\section{Causal Ranking of Agents}\label{sec:causal_ranking}

In the previous sections, we examined the bipartite influence between agents, that is, how much an agent \(m\) affects another agent \(k\) in the network. By calculating this influence for any pair of agents \((m,k)\), we can construct a \(K \times K\) influence matrix \(C\) with entries \([C]_{mk}=C_{m\to k}\). One is often interested in the overall influence of agent \(m\) on the network rather than its effect on individual agents. To that end, in this section, we describe a procedure to use \(C\) for ranking and quantifying the agents' cumulative effect over the network. 

Since \(C\) is constructed from intervened belief dependent entries \(\wcg\), an ordering based on \(C\) would be valid for a particular intervention. For an intervention dose independent ranking of agents, one can consider the matrix \(\overbar{C}\), which is formed with dose independent causal effects $\overbar{C}_{m\to k}$:
\begin{equation}\label{eq:general_cmkb_definition}
    \overbar{C}_{m\to k} \triangleq C_{m \to k} \Big |_{\mu_m(\theta) = 1/H}
\end{equation}
where we extend the definition \eqref{eq:wcnbd_definition} for the NBSL case to the general case. For simplicity of the presentation, in the sequel, we focus on the NBSL case, even though our arguments keep holding for a general $\overbar{C}$ as well as an ordering based on intervention dose dependent matrix \(C\), too.
First, note that the causal effect for the NBSL case is given by
\begin{align}\label{eq:cmknb_derivative_causalrank}
    \wcnb &\stackrel{\eqref{eq:expected_log_belief_trans}, \eqref{eq:cmknb_def}}{=} 1-\dfrac{1}{1 + \!\sum_{\theta \in \Theta \setminus \{ \theta^\circ\}}   \exp \{ - \widetilde{\lambda}_{k,\infty}(\theta) \}} \notag \\   
    &  \: \: \:  \stackrel{\eqref{eq:general_lambda}}{=} 1-\dfrac{1}{1 + \!\sum\limits_{\theta \in \Theta \setminus \{ \theta^\circ\}} \dfrac{\mu_{m}(\theta)}{\mu_{m}(\theta^\circ)}  \exp \Big \{ - \Big [\big ((I-R^{\T})^{-1}-I \big ) d_{-m}(\theta) \Big]_k \Big \}}.
\end{align}
Here, setting \(\mu_m(\theta) = 1/ H\) for any \(\theta \in \Theta\) based on \eqref{eq:general_cmkb_definition} yields
\begin{equation}\label{eq:causalrank_cmknbd}
    \wcnbd = 1-\dfrac{1}{1 + \!\sum\limits_{\theta \in \Theta \setminus \{ \theta^\circ\}}   \exp \Big \{ - \Big [\big ((I-R^{\T})^{-1}-I \big ) d_{-m}(\theta) \Big]_k \Big \}}.
\end{equation}
Since all KL divergences are assumed to be finite ($d_k (\theta) < \infty$) and the strongly connected graph assumption in Sec.~\ref{sec:social_learning_model} implies $\rho (R) < 1$ \cite[Lemma 1]{ying2016information},
\begin{equation}
     \Big \|((I-R^{\T})^{-1}-I \big ) d_{-m}(\theta) \Big \|_{\infty} < \infty.
\end{equation}
Incorporating this into \eqref{eq:causalrank_cmknbd} implies that $\wcnbd > 0, \: \forall m \neq k$. Furthermore, regarding the diagonal elements of \(\overbar{C}\), it holds by definition that an intervention on agent \(m\) implies
\begin{equation}
    \wS_{m,\infty} (\theta) = \log\frac{\mu_{m}(\theta^{\circ})}{\mu_{m}(\theta)}
\end{equation}
As a result, if we set $\mu_m(\theta) = 1 / H , \forall \theta \in \Theta$
\begin{align}
    \overbar{C}_{m \to m}^{\textup{NB}} &\stackrel{\eqref{eq:expected_log_belief_trans}, \eqref{eq:cmknb_def}}{=} 1-\dfrac{1}{1 + \!\sum_{\theta \in \Theta \setminus \{ \theta^\circ\}}   \exp \{ - \widetilde{\lambda}_{m,\infty}(\theta) \}} \notag \\ & \quad= \quad 1-\dfrac{1}{1 + \!\sum_{\theta \in \Theta \setminus \{ \theta^\circ\}}   \exp \{ 0 \}} \notag \\ & \quad =  \quad 1 - \dfrac{1}{H} \notag \\
    & \quad  >  \quad 0
\end{align}

Consequently, all entries of \(\overbar{C}\) are positive, which implies that \(\overbar{C}\) is a primitive matrix. Therefore, according to Perron's theorem \citep{pillai2005perron,easley2010networks}, \(\overbar{C}\) has a unique, real and positive eigenvalue \(\rho \) that dominates all other eigenvalues in magnitude. Moreover, the eigenvector \(q\) corresponding to \(\rho\) is unique up to a scaling and all its entries are positive, i.e.,
\begin{equation}\label{eq:causal_perron_def}
   \overbar{C} q = \rho q, \quad q_k > 0, \quad \forall k=1,\dots, K.
\end{equation}
The entry \(q_k\) is a measure of agent \(k\)'s overall influence over the network. The agents can be ranked with respect to these entries. We name the resulting algorithm {\sf \small CausalRank} which is summarized in Algorithm~\ref{alg:causal_rank}. Importantly, the vector \(q\) --- which is the output of Alg.~\ref{alg:causal_rank} --- differs from the network centrality eigenvector \(v\) in general. While \(v\) is determined solely by the combination matrix \(A\) (see \eqref{eq:perron_def}), as shown in previous sections, causal influences and hence \(q\) depend on the informativeness of agents as well.
\begin{algorithm}[]
\begin{algorithmic}[1] 
\State \textbf{Input}: a network of \(K\) agents with indices \( \{ 1,2,\dots, K\}\), combination matrix \(A\), set \(\Theta\) of \(H\) hypotheses, informativeness vector \(d(\theta)\)
\State \textbf{Initialize}: \(K \times K\)-dimensional influence matrix \(\overbar{C}\)
\FOR{each agent $m = 1,2,\dots, K$}
\FOR{each agent $k = 1,2,\dots, K$}
\IF{$m=k$}
\State set \([\overbar{C}]_{mm} := 1- \dfrac{1}{H}\)
\ELSE
\State compute the causal effect \(\overbar{C}_{m \to k}\) with \eqref{eq:causalrank_cmknbd}
\State set \([\overbar{C}]_{mk} := \overbar{C}_{m \to k}\)
\ENDIF
\ENDFOR
\ENDFOR
\State find the largest eigenvalue \(\rho\) of \(\overbar{C}\)
\State \textbf{Output}: the eigenvector \(q\) satisfying \( \overbar{C} q = \rho q\)
\caption{{\sf \small CausalRank} Algorithm}\label{alg:causal_rank}
\end{algorithmic}
\end{algorithm}
More specifically, \eqref{eq:causal_perron_def} computes a \emph{causal} eigenvector centrality that attributes higher importance to exerting influence on agents who are themselves influential. A possible alternative approach (which we call \emph{average influence ranking} (AIR)) can treat all agents with equal regard in the averaging process by assigning the following ranking score to each agent $m$:
\begin{equation}
    \text{AIR}(m) = \frac{1}{K-1} \sum_{k \neq m} \overbar{C}_{m \to k}
\end{equation}
In contrast, rather than employing a simple averaging, {\sf \small CausalRank} seeks the equilibrium vector by assigning significant weights to those agents that have a higher influence on other influential agents. This concept bears resemblance to other methodologies based on eigenvector centrality, such as the PageRank algorithm \citep{brin1998}. While ranking websites, PageRank gives preferential treatment to links from more central websites.

It is also worth mentioning that {\sf \small CausalRank} is distinct from the causal ordering methods for directed acyclic graphical models \citep{peters2017elements} since we are dealing with cyclic graphs with bidirectional links due to our time-series setting. Furthermore, {\sf \small CausalRank} is not only useful for ranking, but also provides information on the strength of agents' overall influence on others. 

\section{Causal Discovery from Observational Data}\label{sec:causal_discovery}

In Sec.~\ref{sec:main_results}, we derived the closed-form expressions \eqref{eq:general_lambda} and \eqref{eq:lambda_m_infty_asl} for the steady-state equilibrium of the network under interventions, which necessitate knowledge of the combination matrix $A$ and the informativeness of agents $d(\theta)$. In practice, these parameters might not be readily available. The work \citep{valentina2023discovering} introduced the Graph Social Learning (GSL) algorithm, which can be used to recover $A$ and $d(\theta)$ using a sequence of publicly shared intermediate beliefs (a.k.a. \emph{actions}) $\{\bpsi_{k,i}\}$ in the \emph{observational} setting of the ASL algorithm. Using observational data only can be especially useful in social network contexts where conducting experiments is not feasible. Nonetheless, \citep{valentina2023discovering} acknowledge that the algorithm may not perform well in real-world scenarios, mainly due to the limitations of the social learning model in accurately describing the real world. However, in many applications, some information about the underlying combination matrix $A$ may already be available. For instance, in Twitter, the publicly available adjacency matrix can provide information about which user follows which other users. 

Taking these aspects into account, in this section, we propose an algorithm that utilizes the adjacency matrix and a temporal sequence of publicly shared intermediate beliefs $\{\bpsi_{k,i}\}$ to estimate bipartite causal effects for both NBSL and ASL algorithms. Specifically, we leverage the graph of user connections to estimate combination matrix weights by using existing methods in the literature (e.g., averaging rule). Then, using the estimated combination matrix, we estimate the informativeness of agents by using belief update recursions. By inserting these to the closed-form expressions \eqref{eq:general_lambda} and \eqref{eq:lambda_m_infty_asl}, we estimate the causal effects. To that end, observe that the intermediate log-belief ratios evolve based on a linear recursion due to \eqref{eq:asl_local_update}:
\begin{equation}\label{eq:gsp_recursion_linear}
    \bLambda_i = (1-\delta)A^{\T}\bLambda_{i-1} + \beta \bLX_i
\end{equation}
where we are now defining the following \(K \times H\) matrices over all agents \(k \in \{1, \dots, K\}\) and hypotheses \(\theta_j \in \Theta\):
\begin{equation}
    [\bLambda_i]_{kj} \triangleq \log\frac{\bpsi_{k,i}(\widehat{\btheta}^\circ)}{\bpsi_{k,i}(\theta_j)},  \quad  [\bLX_i]_{kj} \triangleq \log \frac{L_k(\bxi_{k,i} | \widehat{\btheta}^\circ )}{L_k(\bxi_{k,i} | \theta_j)} .
\end{equation}
Here, \(\widehat{\btheta}^\circ\) is an estimate for the latent state of nature \(\theta^\circ\) computed as follows after some time $M$:
\begin{equation}\label{eq:theta_circ_estimation}
   \widehat{\btheta}^\circ \triangleq \argmax_{\theta \in \Theta} \sum_{k=1}^K \bpsi_{k,M} (\theta).
\end{equation}
The rationale behind \eqref{eq:theta_circ_estimation} is that under proper assumptions we know from Theorems~\ref{theorem:truth_learning_nbsl} and \ref{th:asl_conv_dist} that agents learn the true hypothesis with more confidence as \(M\) grows. Our goal is to infer the true combination matrix \(A\) and informativeness vector \(d(\theta)\) for each hypothesis from a sequence of \(M+1\) matrices \(\{\bLambda_M, \bLambda_{M-1},\dots, \bLambda_0 \}\) and the adjacency matrix of the agents.

We can estimate $A$ by using existing procedures in the literature for forming combination matrices from adjacency matrices, e.g., by using the averaging or relative degree rules. For instance, the averaging rule assigns the same weight to all neighbors of an agent, i.e.,
\begin{equation}\label{eq:averaging_rule}
    [\:\widehat{A}\:]_{\ell k} = \begin{dcases} 
      \dfrac{1}{|\mathcal{N}_k|}, & \text{if there is a link from \(\ell\) to \(k\) (i.e., $\ell \in \mathcal{N}_k$)} \\
      0, & \text{otherwise}
   \end{dcases}
\end{equation}
After forming the combination matrix estimate $\widehat{A}$, we can insert it into \eqref{eq:gsp_recursion_linear} and average over available $M$ samples to estimate the average log-likelihood ratios that correspond to the informativeness of agents using
\begin{equation}\label{eq:llr_estimation_M}
    \widehat{\bm{D}} =  \frac{1}{\beta M} \sum_{i=1}^M \! \Big (\bLambda_i - (1-\delta) \widehat{A}^{\T} \bLambda_{i-1} \Big) .
\end{equation}
Then, one can replace $A$ and $d_k (\theta_j)$ with $\widehat{A}$ and $[\widehat{\bm{D}}]_{k j}$ in Sec.~\ref{sec:main_results} to obtain the causal effect estimate $\widehat{\bm{C}}_{m \to k}$. The complete procedure is summarized in Alg.~\ref{alg:social_causal_learning}. Essentially, it combines our causality results in Sec.~\ref{sec:main_results} with a straightforward adjustment to the GSL algorithm from \citep{valentina2023discovering}.

 \begin{algorithm}[]
 \caption{Graph Causality Learning (GCL)}
 \begin{algorithmic}[1]
    \State \textbf{Input}: a sequence of shared beliefs $ \{\bpsi_{k,i}\}$ for $M+1$ time instants, the graph topology of agents
    \State \textbf{Parameters}: for NBSL $\delta = 0, \beta = 1$, for ASL $\delta \in (0,1), \beta > 0$ 
    \State set true state of nature estimate:
\(
   \widehat{\btheta}^\circ = \argmax\limits_{\theta \in \Theta} \sum\limits_{k=1}^K \bpsi_{k,M} (\theta)
\)
    \State form left-stochastic combination matrix estimate $\widehat{A}$ from input adjacency matrix, e.g., by \eqref{eq:averaging_rule}
    \FOR{$ i = 0,1,\dots,M $} 
    \State  for each agent \(k\) and hypothesis \(\theta_j\), set the entry:\vspace{-1em}
    \begin{equation}
         [\bLambda_i]_{kj} = \log\frac{\bpsi_{k,i}(\widehat{\btheta}^\circ)}{\bpsi_{k,i}(\theta_j)}
    \end{equation}\vspace{-1.5em}
        \ENDFOR
    \State estimate the informativeness:\vspace{-1em}
    \begin{equation}
        \widehat{\bm{D}}=  \frac{1}{\beta M} \sum_{i=1}^M  \Big (\bLambda_i - (1-\delta) \widehat{A}^{\T} \bLambda_{i-1} \Big) 
    \end{equation}\vspace{-1em}
    \State for any given agent pair $(m,k)$, compute approximate causal effect \( \widehat{\bm{C}}_{m \to k} \)  by replacing $A$ and $d_k (\theta_j)$ with $\widehat{A}$ and $[\widehat{\bm{D}}]_{k j}$ in the original expression \eqref{eq:cmkgeneral_def} for $\wcg$ (which also  requires using \eqref{eq:general_lambda} for NBSL and \eqref{eq:lambda_m_infty_asl} for ASL)
  \State \textbf{Output}: \( \widehat{\bm{C}}_{m \to k} \) 
 \end{algorithmic} \label{alg:social_causal_learning}
 \end{algorithm}

The graph causality learning (GCL) algorithm (Alg.~\ref{alg:social_causal_learning}) only requires a sequence of shared intermediate beliefs (actions) and the knowledge of adjacency matrix. This enhances its practicality and makes it advantageous in terms of privacy for scenarios where only limited information is publicly accessible. For example, in a network of Twitter users, shared beliefs (opinions) in the form of tweets (posts) and the knowledge of who follows whom can usually be accessed by all users, while the external exposure to information (e.g., from mass media channels distinct from Twitter) may not be available. Therefore, the GCL algorithm can be useful for analyzing social media content while respecting privacy. In the next result, we provide a performance bound on the GCL algorithm.

\begin{theorem}[Causal influence estimation]\label{th:inform_learning}
    For sufficiently small combination matrix estimation errors and $\delta$ values, the error in causal influence estimation decreases with increasing number of samples $M$ in expectation, namely, 
    \begin{equation}
    \e \: \Big |  \wcg - \widehat{\bm{C}}_{m \to k}  \Big | =  O(1/\sqrt{M})
\end{equation}
for both NBSL and ASL under any intervention strength $\mu_m$.

\end{theorem}
\begin{proof}
    See Appendix~\ref{appendix:inform_learning}.
\end{proof}

We show the practical usefulness of Alg.~\ref{alg:social_causal_learning} by means of a real-world application to social media data in Sec.~\ref{sec:twitter_app}. Furthermore, a detailed analysis of the time complexity of the algorithms discussed in this paper is provided in Appendix~\ref{appendix:complexity}.

\section{Computer Simulations}\label{sec:computer_sims}

For our numerical simulations, we first study a network of $K=11$ agents, interconnected with the strongly connected graph topology in Fig.~\ref{fig:simulations_network}.
\begin{figure}[t]
    \begin{minipage}[b]{0.49\textwidth}
        \centering
        \includegraphics[width=\textwidth]{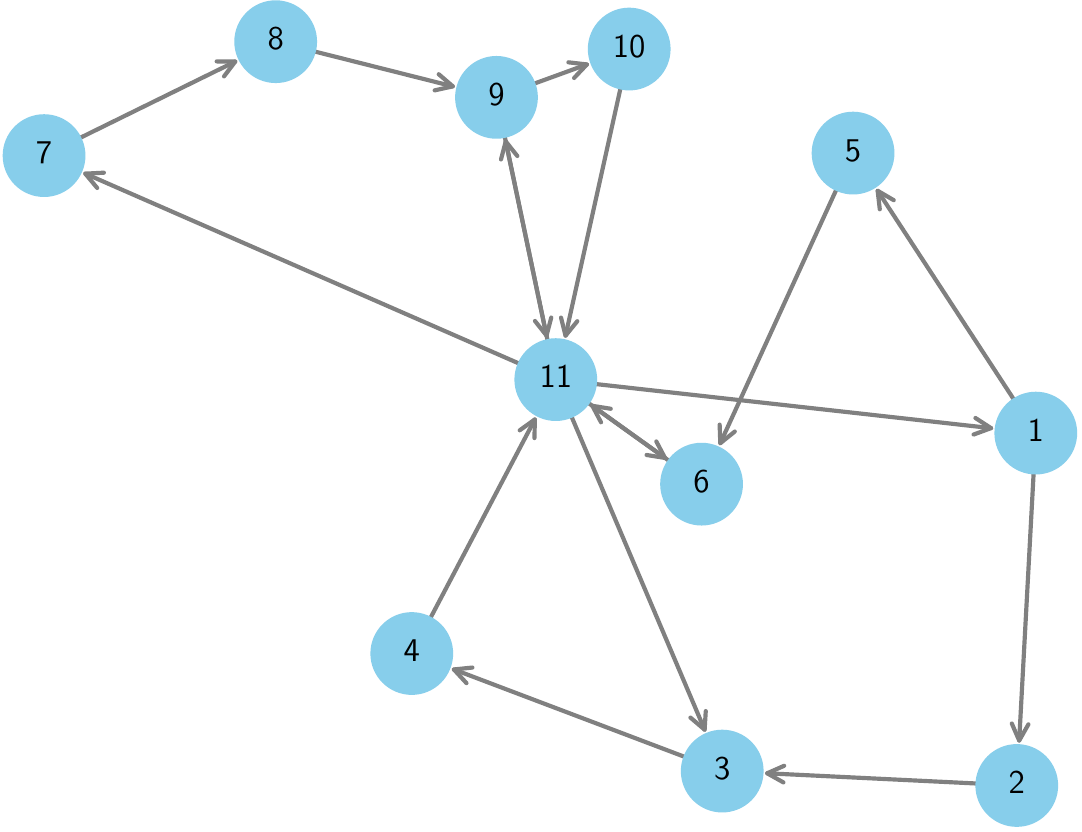}\vspace{0.7cm}
        \caption{The strongly connected network architecture used in Sec.~\ref{sec:computer_sims}. Each agent also has a self-loop, which is omitted for visual simplicity.}
        \label{fig:simulations_network}
    \end{minipage}
    \hfill
    \begin{minipage}[b]{0.47\textwidth}
        \centering
        \begin{tabular}{|p{1.8cm}|p{1.8cm}|p{1.8cm}|}
            \hline
            Agent & $\nu_k$ & $d_k (\theta^\prime)$ \\
            \hline
            1 & 0.8 & 0.32 \\
            2 & 0.6 & 0.18 \\
            3 & 0.2 & 0.02 \\
            4 & 0.6 & 0.18 \\
            5 & 0 & 0 \\
            6 & 0 & 0 \\
            7 & 0.4 & 0.08 \\
            8 & 0.4 & 0.08 \\
            9 & 0.2 & 0.02 \\
            10 & 0.6 & 0.18 \\
            11 & 0.8 & 0.32 \\
            \hline
        \end{tabular}
        \captionof{table}{Mean $\nu_k$ of the observations under alternative hypothesis and the corresponding informativeness levels $d_k (\theta^\prime)$ for each agent $k$.}
        \label{tab:my_label}
    \end{minipage}
\end{figure}
The agents observe data drawn from a Gaussian distribution and aim to distinguish the true state $\theta^\circ$ from $H=2$ possible hypotheses. Under the true state, each agent $k$ observes data that follows a Gaussian distribution with zero mean and unit variance, expressed as:
\begin{equation}
L_k (\xi | \theta^\circ) = \frac{1}{\sqrt{2\pi}} \exp \Big \{-\dfrac{1}{2}\xi^2 \Big \}.
\end{equation}
Under the alternative hypothesis $\theta^\prime \neq \theta^\circ$, we assume that the data still has unit variance for all agents, but the mean vector $\nu_k$ changes as shown in Table~\ref{tab:my_label}. Therefore, the informativeness of each agent, which is equal to the KL divergence between $L_k (\xi | \theta^\circ)$ and $L_k (\xi | \theta^\prime)$, is given in Table~\ref{tab:my_label} and is calculated as follows:
\begin{align}
d_k (\theta^\prime)  = \dkl (L_k (\xi |\theta^\circ) || L_k (\xi | \theta^\prime) )  = \frac{1}{2} \nu_k^2.
\end{align}
Notably, agents 5 and 6 have no informativeness, that is, they are not able to learn the truth without cooperating with the other agents. Initially, we assume that the agents observe spatially independent data. In other words, the covariance matrix is an identity matrix.

\begin{figure}
  \centering
  \includegraphics[width=.9\linewidth]{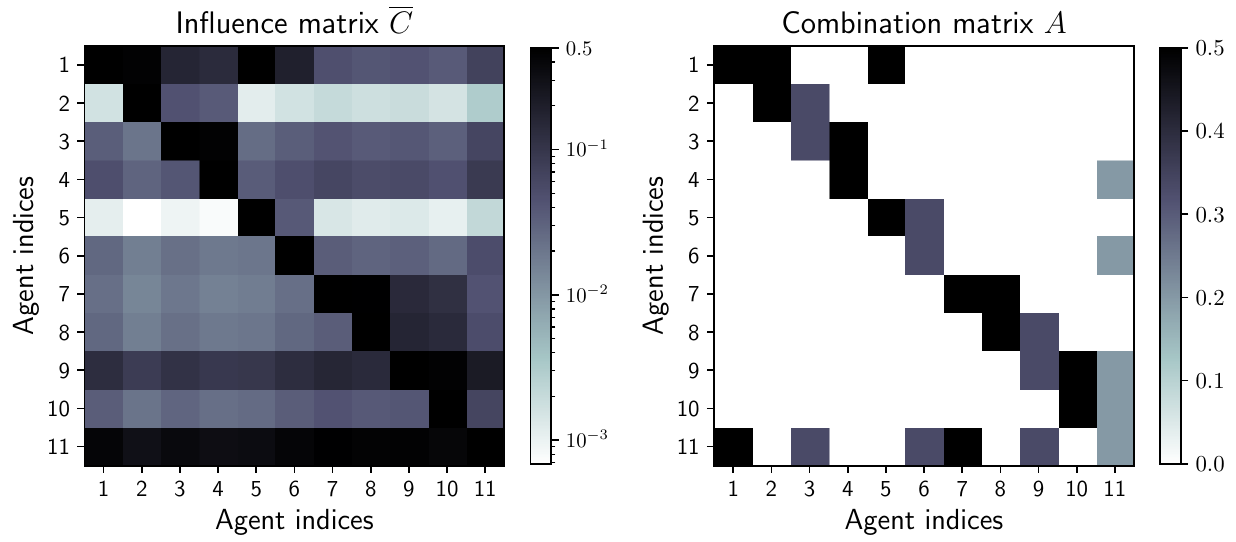}
  \caption{(\emph{Left}): Bipartite causal influence matrix. (\emph{Right}): Combination matrix corresponding to the network topology in Fig.~\ref{fig:simulations_network} formed with averaging rule.}
  \label{fig:inf_comb}
\end{figure}

We start with the NBSL case ($\delta = 0, \beta = 1$). The right panel in Fig.~\ref{fig:inf_comb} shows the combination matrix that is derived from the averaging rule applied to the graph topology in Fig.~\ref{fig:simulations_network}. Notice that the averaging rule generates a matrix whose entries are constant column-wise. The left panel in Fig.~\ref{fig:inf_comb} shows the matrix of bipartite causal effects where the entry in $m$-th row $k$-th column represents $\wcnbd$ (see \eqref{eq:causalrank_cmknbd} for the explicit formula).

Upon comparing the two heat maps in Fig.~\ref{fig:inf_comb}, it becomes apparent that the combination matrix entries do not reveal the causal relationships directly. For example, despite the absence of a direct connection in the combination matrix (as indicated by 0 entries), agent 11 exerts significant influence on agents 2 and 8. This phenomenon highlights the importance of taking the ripple effects over a network into account. Furthermore, the influence of agent 1 on agent 5 is notably high. Given the zero informativeness of agent 5, this finding aligns with our expectations, as low-informativeness agents are easier to control (remember the discussion in Sec.~\ref{sec:main_results_nbsl}). Agent 5 being a low-informativeness agent also facilitates the propagation of influence from agent 1 to agent 6 via agent 5. Intriguingly, despite the absence of a direct connection between agents 1 and 6, this indirect influence is more substantial than the influence of agent 5 on agent 6. This shows that mixing of information over a network necessitates an understanding of causal influence beyond local interactions.

\begin{figure}
  \centering
  \includegraphics[width=.65\linewidth]{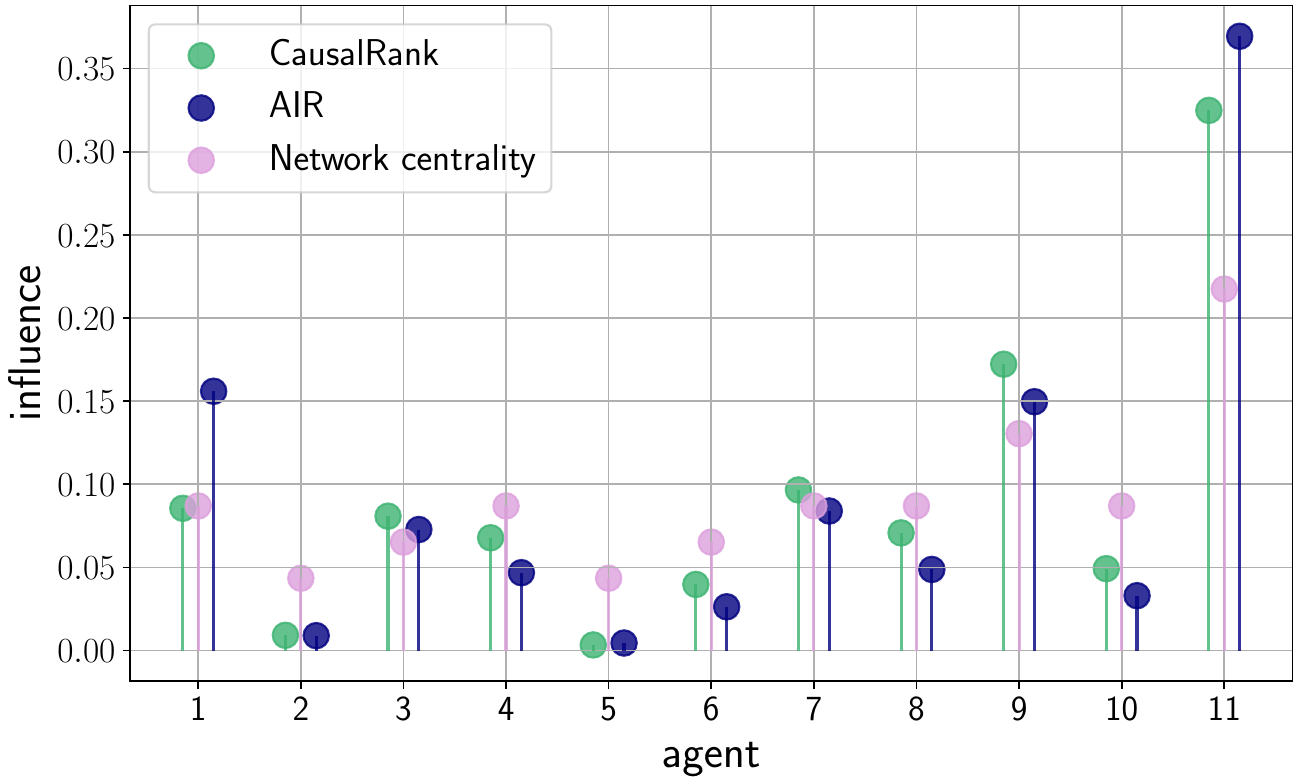}
  \caption{Ranking of agents based on their causal influence (for both {\sf \small CausalRank} and AIR) and their network-topology based eigenvector centrality. All ranking scores are normalized to sum up to one.}
  \label{fig:numerical_ranking}
\end{figure}

Next, in Fig.~\ref{fig:numerical_ranking}, by using the matrices in Fig.~\ref{fig:inf_comb}, we compare the overall influences of agents using three methods: {\sf \small CausalRank}, AIR, and network eigenvector centrality. Notably, the {\sf \small CausalRank} and AIR metrics yield similar results as they both use the bipartite causal relations matrix for causal ranking. For instance, agents 2 and 5 possess relatively low rankings in both of these metrics. The network eigenvector centrality, on the other hand, only relies on the combination matrix, and often deviates from these two metrics. Specifically, it assigns relatively higher scores to agents 2 and 5, and a comparatively lower score to agent 11. Moreover, an interesting distinction between AIR and {\sf \small CausalRank} becomes apparent when considering the case of agent 9. We can see from the causal influence matrix in Fig.~\ref{fig:inf_comb} that agent 9 has a substantial impact on agent 11 --- the most influential agent (see Fig.~\ref{fig:numerical_ranking}). Consequently, agent 9's {\sf \small CausalRank} score surpasses its AIR score. This can be attributed to {\sf \small CausalRank}'s consideration of the significance of influencing agent 11. Unlike AIR, which assigns uniform weights, {\sf \small CausalRank} assigns a higher weight to influences on more influential agents.

\begin{figure}
  \centering
  \includegraphics[width=.6\linewidth]{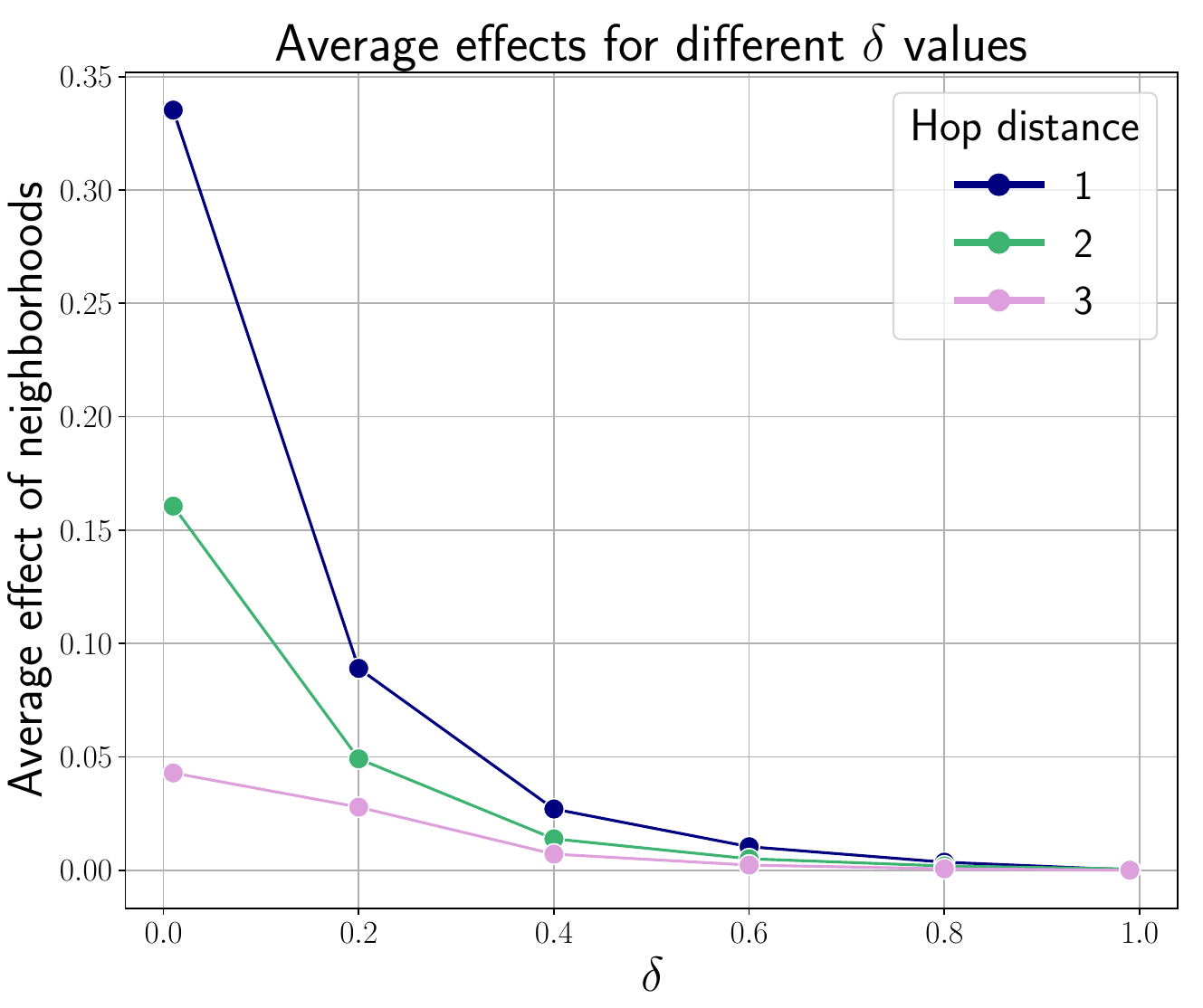}
  \caption{Average influence on agent $k=4$ from its neighborhood of 1,2 and 3 hop distances with respect to changing values of forgetting factor $\delta$.}
  \label{fig:different_delta}
\end{figure}

\begin{figure}
  \centering
  \includegraphics[width=.85\linewidth]{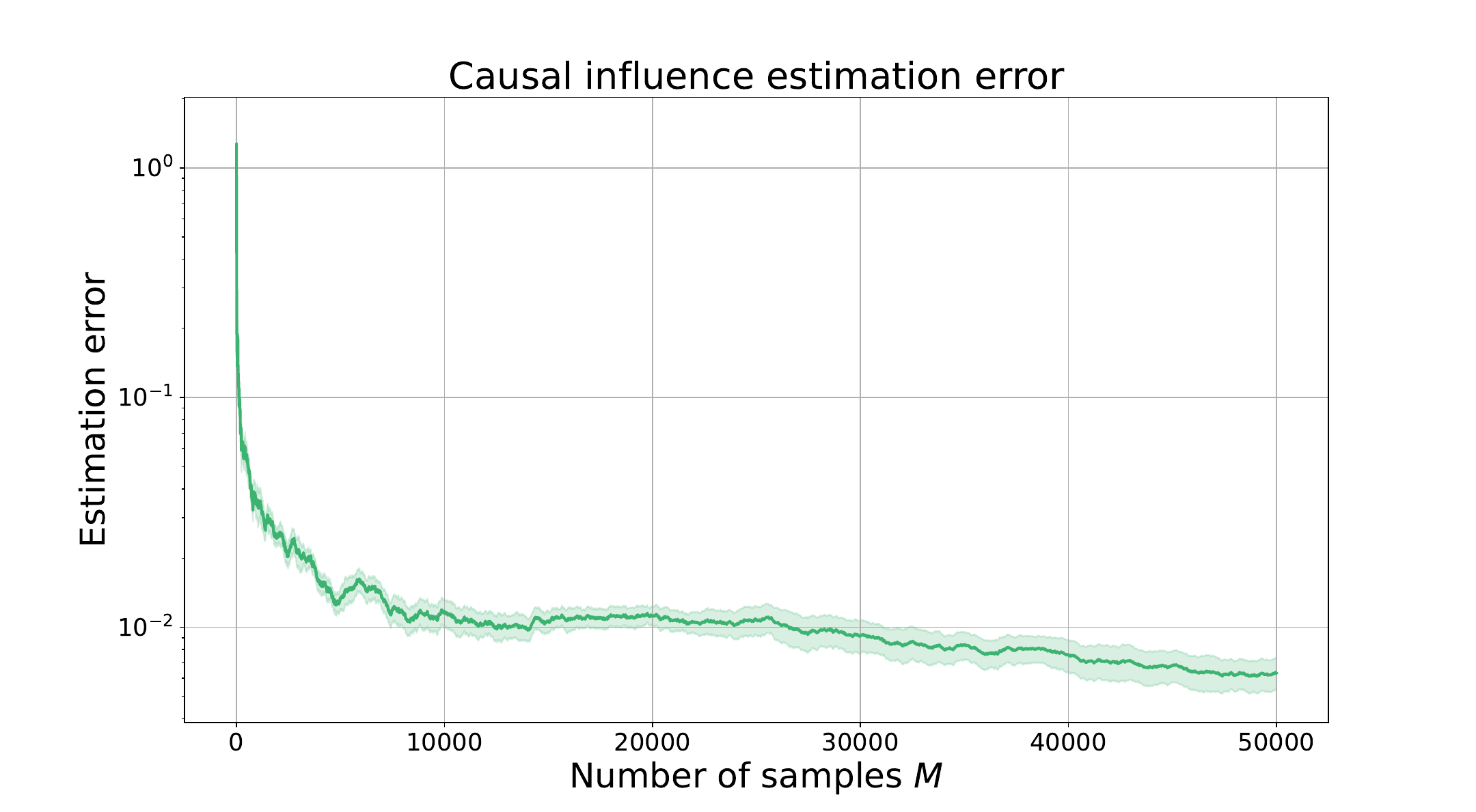}
  \caption{Causal influence estimation error with respect to increasing number of time samples $M$. Averaged over $10$ independent experiments.}
  \label{fig:causal_num_error_log}
\end{figure}

\begin{figure}
  \centering
  \includegraphics[width=.8\linewidth]{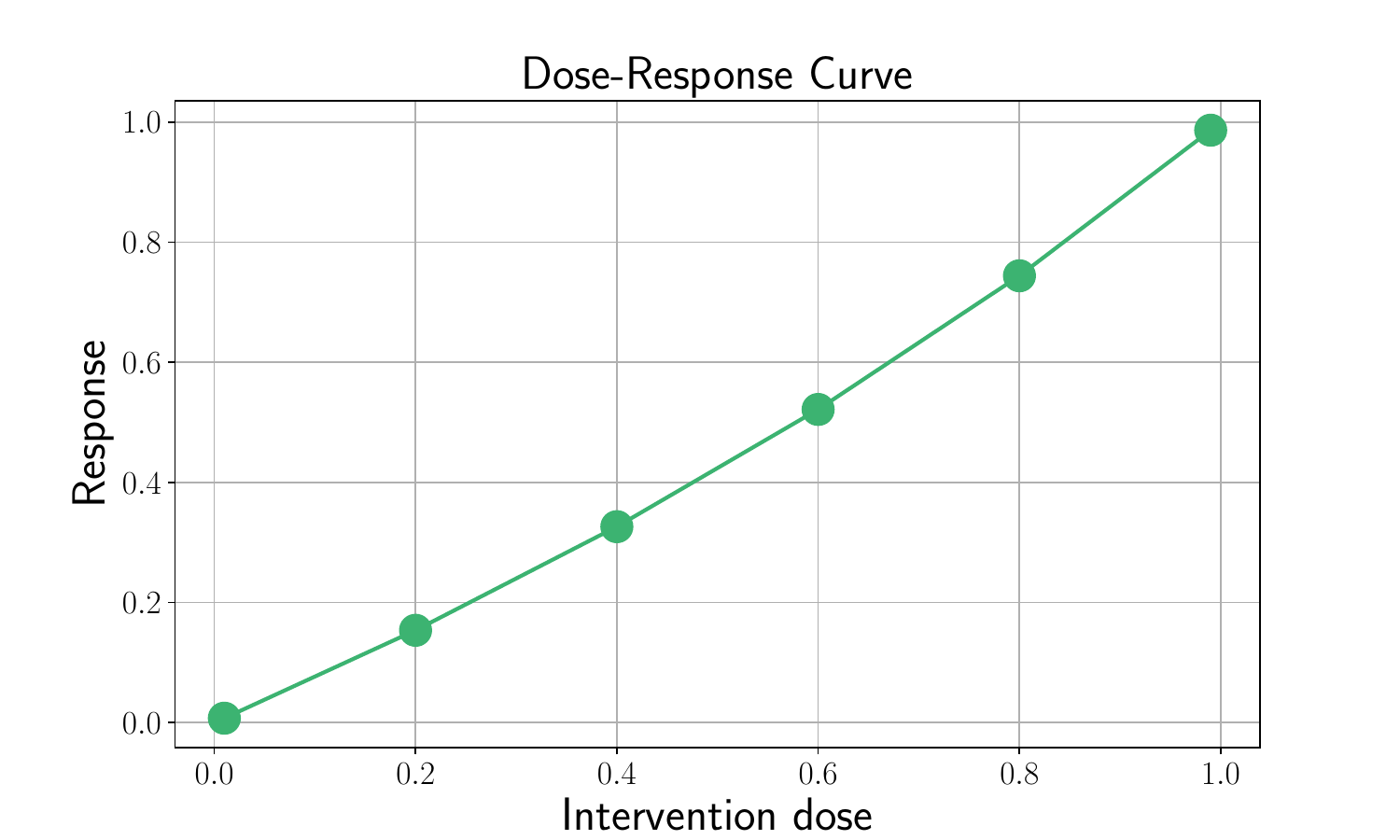}
  \caption{$\wcnb$ value from \eqref{eq:general_lambda_transf} for $m=11$ and $k=6$ with respect to increasing intervention strength $\mu_m(\theta)$ for $\theta \neq \theta^\circ$.}
  \label{fig:dose_response}
\end{figure}

To gain insights into the influence of the forgetting factor $\delta$ in the ASL case, we focus our attention on agent $k=4$. In Fig.~\ref{fig:different_delta}, we present the average influence exerted by agent 4's neighbors that are 1, 2, and 3 hops away. It is clear from Fig.~\ref{fig:different_delta} that the influence of distant agents diminishes with increasing $\delta$. This is because increasing $\delta$ increases the significance of recent observations, and since information from distant agents loses its recency by the time it arrives at agent 4, this implies assigning less importance to information from those distant agents.

In the simulations conducted so far, we have considered dose-independent causal effects $\overbar{C}_{m\to k}$. Figure~\ref{fig:dose_response} demonstrates the dependency of $\wcg$ on the intervened belief $\mu_m (\theta)$ for $\theta \neq \theta^\circ$, as described in equation~\eqref{eq:general_lambda_transf}. We use $m=11$ and $k=6$ for this plot, which illustrates the theoretical result presented in equation~\eqref{eq:general_lambda_transf}.

\begin{figure}
  \centering
  \includegraphics[width=.65\linewidth]{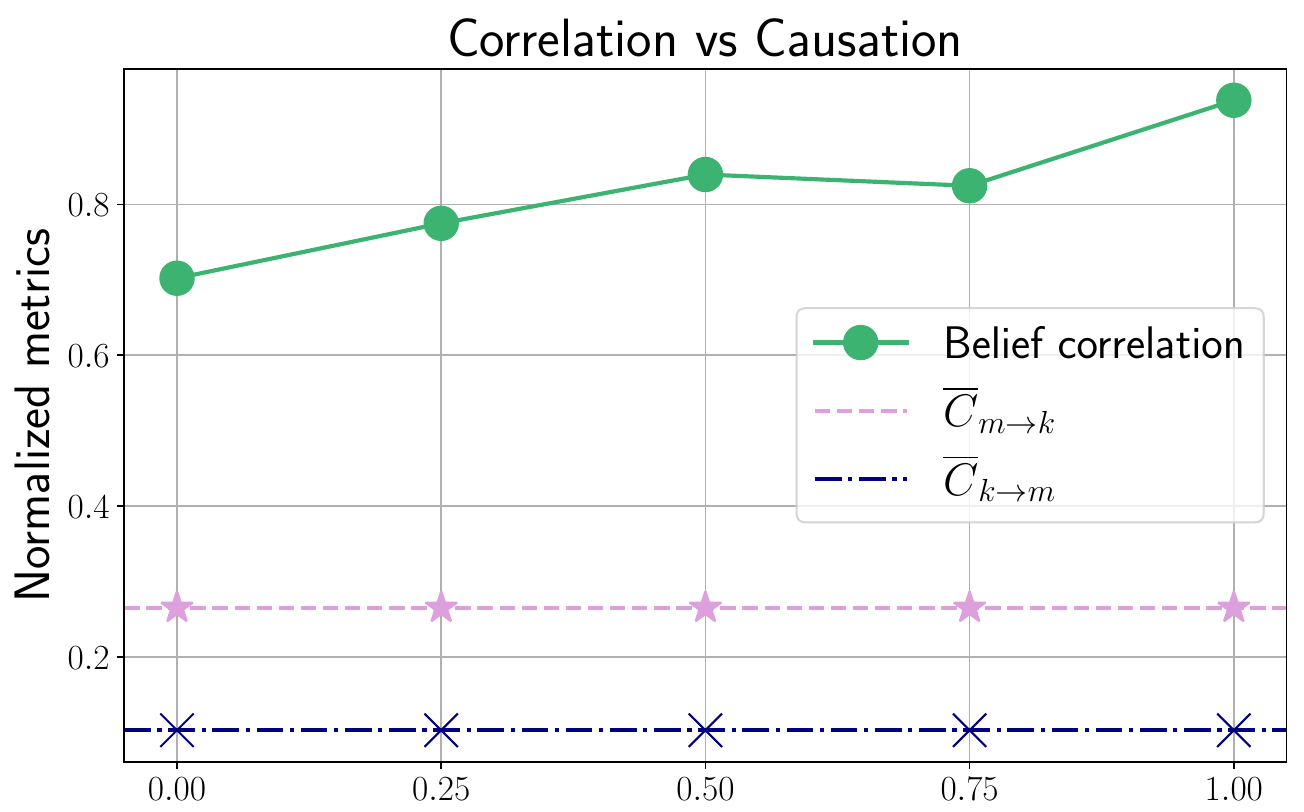}
  \caption{Correlation and causation between agents 6 and 11 with respect to varying dependence between the streaming observations they receive.}
  \label{fig:correlation_vs_causation_num}
\end{figure}

Then, we fix $\delta = 0.1$, and use the GCL algorithm (Alg.~\ref{alg:social_causal_learning}) in order to estimate the causal effects using observational data (shared beliefs) as described in Sec.~\ref{sec:causal_discovery}. The norm disagreement of the causal influence matrix formed with estimates and the true causal influences, averaged over 10 Monte Carlo simulations, is given in Fig.~\ref{fig:causal_num_error_log}. Observe that the error is decreasing as the number of samples $M$ increases, which supports Theorem~\ref{th:inform_learning}. 

Next, we illustrate the distinction between causality and correlation by again considering agents $m=11$ and $k=6$. The joint distribution of their data is changed by introducing varying levels of correlation to the observations that these agents are receiving. Fig.~\ref{fig:correlation_vs_causation_num} shows that as the correlation in data increases, the correlation of the asymptotic beliefs of these agents also changes. However, the causal effects (both the effect of agent 6 on 11 and that of agent 11 on 6) remain constant. This shows that external observations can act as a correlation inducing confounding factor. Yet, our method maintains consistent results, which shows its robustness against non-causal factors.

\begin{figure}
  \centering
  \includegraphics[width=.68\linewidth]{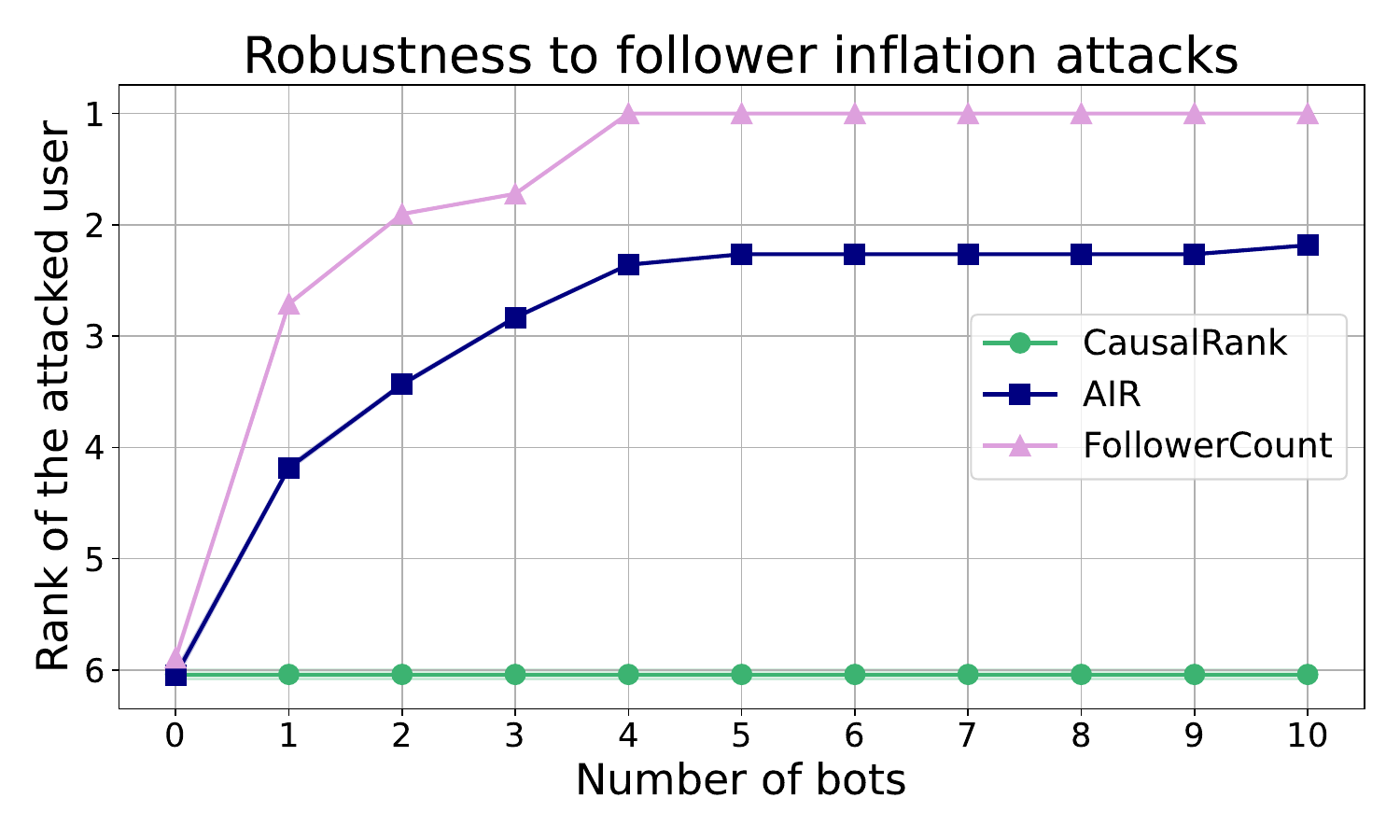}
    \caption{Influence rank of an attacked agent whose followers are artificially inflated with bot accounts. The attacked agents are randomly selected, and the experiment is repeated $5000$ times.}
  \label{fig:bot_attack_raw_rank}
\end{figure}

In order to demonstrate the practical usefulness of {\sf \small CausalRank}, we compare the robustness of the considered ranking methods against follower-inflation attacks. Such attacks are practically relevant in social networks because visible popularity metrics can be monetized through sponsorships, advertising and other reputation-based opportunities, which creates incentives to purchase fake indicators of social media influence.
We use the same network topology and informativeness levels as in Fig.~\ref{fig:simulations_network} and Table~\ref{tab:my_label}, but augment the network by introducing additional agents (bot accounts) that act as artificial followers of a randomly selected (attacked) agent.
These bots are non-informative and are designed only to increase the apparent popularity of the attacked agent without contributing meaningful influence over the other agents.
Fig.~\ref{fig:bot_attack_raw_rank} shows the average rank of the attacked agent as the number of bots increases. 
The reported values are averaged over 5000 independent repetitions.
As expected, the ranking based on the number of followers is highly sensitive to this type of manipulation, while AIR is also noticeably affected since it rewards influence over a larger number of agents irrespective of their importance.
In contrast, the average rank produced by {\sf \small CausalRank} remains close to the theoretical expectation of $6$ for $K=11$ honest agents.
We provide additional experiments in Appendix~\ref{appendix:additional_experiments} with networks of different architectures and sizes to show that this effect continues to hold across those settings.

\begin{figure}
  \centering
  \includegraphics[width=.68\linewidth]{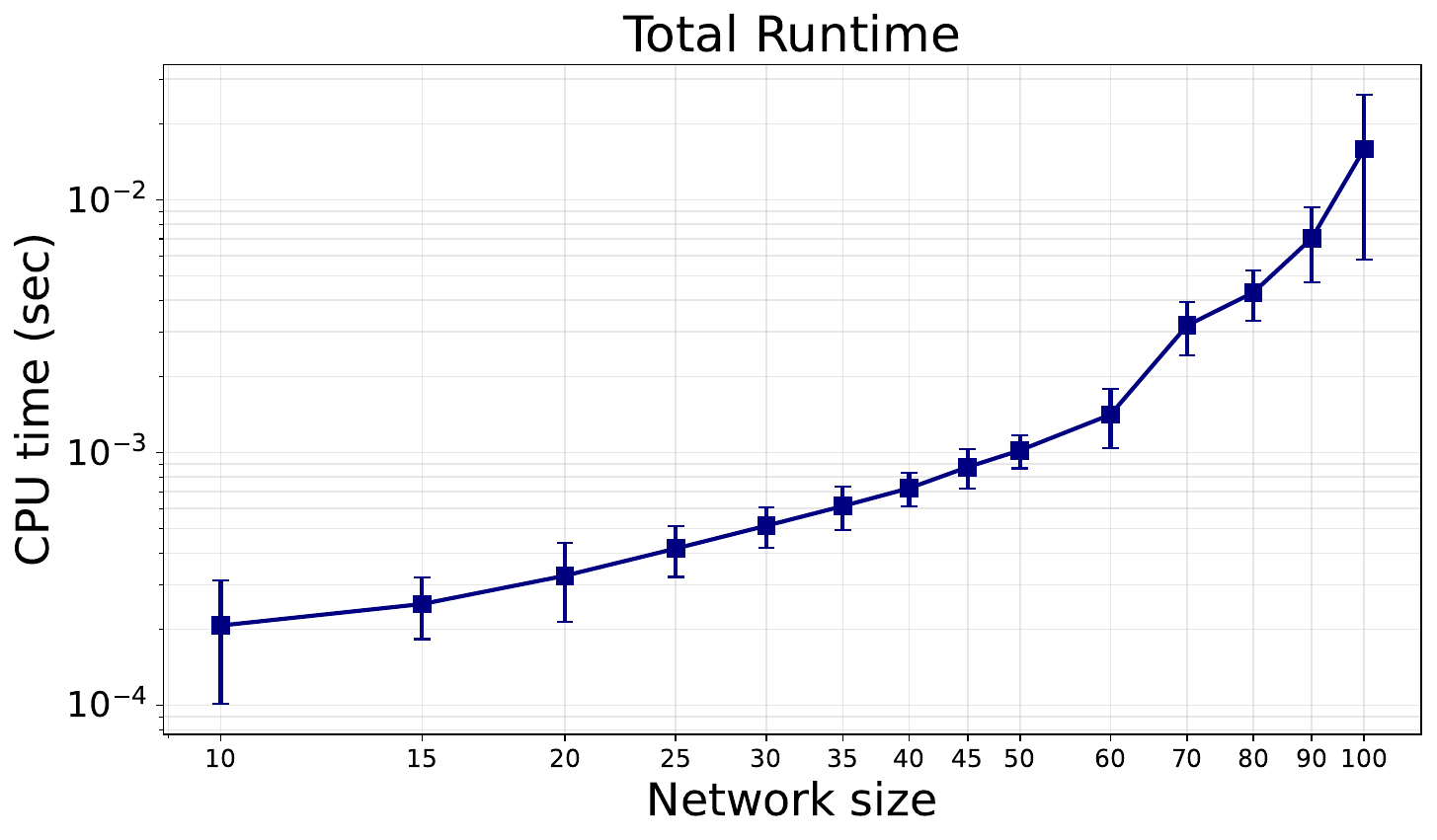}
    \caption{CPU-time to compute influence matrix $C$ and apply {\sf \small CausalRank} for increasing network sizes.}
  \label{fig:total_runtime_different_sizes}
\end{figure}

Finally, we benchmark the computational time required to calculate the causal influence matrix $C$ and {\sf \small CausalRank} as a function of network size $K$ for NBSL. For each value of $K$ between $10$ and $100$, we generate $20$ independent experiments, each on a distinct strongly-connected Erd\H{o}s--Renyi graph. In Fig.~\ref{fig:total_runtime_different_sizes}, we report the total CPU time required for this computation in seconds, averaged over the independent experiments. As shown in Appendix~\ref{appendix:complexity}, the influence matrix calculation can be accelerated by reusing common matrix operations across different agents in the network. We provide additional plots illustrating this effect and also benchmark the CPU time for ASL in Appendix~\ref{appendix:additional_experiments}.

\section{Application to Social Media Data}\label{sec:twitter_app}

In this section, we use the GCL algorithm (Alg.~\ref{alg:social_causal_learning}) on real-world data to assess the influence of Twitter users. Our approach distinguishes itself from prior works \citep{twitter_quercia2011mood,twitter_influence_measure,smith2021automatic,valentina2023discovering}, which typically rely on some descriptive statistics to measure influence in Twitter. More specifically, in our approach,
\begin{itemize}
    \item All input requirements are publicly available, i.e., publicly shared posts (tweets) by users and the information of who follows whom. This offers a significant advantage in terms of privacy, as we do not require any private feature about users.
    \item Going beyond providing a mere ranking of influential users, we also quantify the bipartite causal relations.
    \item We leverage natural language processing tools to extract meaningful information from the content of users' posts to form belief inputs, rather than relying on traditional simpler metrics such as the posting frequency. In contrast to the binary treatment approach, adopting the continuous treatment metric in \eqref{eq:cmkgeneral_def} for our causality definition enables us to allow for the use of continuous variables such as opinions and sentiment scores.
\end{itemize}
 Note that we utilize Alg.~\ref{alg:social_causal_learning} for the NBSL model ($\delta = 0, \beta = 1$) as-is, without employing additional techniques to enhance its accuracy for real-world modeling. Our intention is to demonstrate the practical usefulness of our algorithm rather than striving to develop the most advanced practical algorithm available. 

\begin{wrapfigure}[20]{L}{7cm}
\includegraphics[width=7cm]{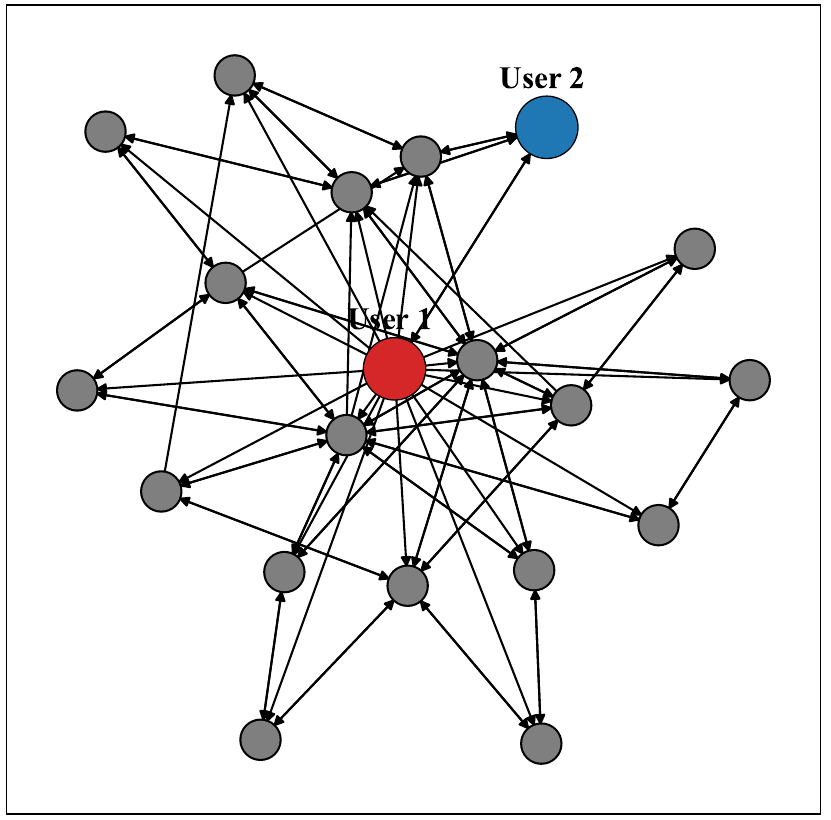}
\vspace{-2em}\caption{Sampled Twitter sub-network with $K=20$ users. An arrow from a user $k$ to $\ell$ means that user $\ell$ is following user $k$. The connectivity density of the graph is approximately $0.24$.}\label{fig:twitter_network}
\end{wrapfigure} 
\noindent \textbf{Network structure.} Performance evaluation of the influence estimation algorithms in real-world social networks is challenging due to the absence of ground truth regarding influential users. There is also no ground truth reference for the confidence scores assigned by users to one another (i.e., combination weights) or for the information the users are obtaining from out-of-network resources (i.e., informativeness levels). Therefore, we utilize the framework from \citep{valentina2023discovering}, namely, a sub-network consisting of $K=20$ Twitter users, as illustrated in Fig.~\ref{fig:twitter_network}. Notably, this sub-network incorporates Elon Musk (User 1), a public figure with 140 million followers across Twitter, who is reportedly influential on cryptocurrency prices \citep{lennart_bitcoin_musk}. All users within the sub-network actively share posts related to cryptocurrencies and bitcoin-related topics. Furthermore, we highlight another user (User 2), who has 1,167 followers across Twitter, is notable for being followed by User 1, as depicted in Fig.~\ref{fig:twitter_network}. Importantly, the sub-network exhibits a strong connectivity among its members.

\noindent \textbf{Opinion processing.} The Twitter API is leveraged in order to collect the posts (tweets) of users between 01.01.2017 and 01.05.2022 relevant to crypto-currency discussions, using query keywords such as ``coin'', ``bitcoin'', or ``crypto-currency''. To quantify the contextual information of these posts to form the input beliefs, sentiment analysis based on neural language models \citep{loureiro2022timelms} is utilized. We refer to Fig.~\ref{fig:tweets_table} for some illustrative examples. The sentiment scores obtained through natural language processing ranges from 0 to 1, signifying the degree of positive attitude towards Bitcoin. These scores correspond to the beliefs of the agents on the hypothesis of ``Bitcoin is good/useful''. We consider two hypotheses, i.e., $H=2$, where the counter-hypothesis is ``Bitcoin is bad/harmful''. 

We then integrated these beliefs obtained from users' tweets, along with the sub-network topology of who follows whom, into Alg.~\ref{alg:social_causal_learning}. Specifically, we employed NBSL modeling, i.e., $\delta=0$ and $\beta=1$. We binned the sequence of data to days, that is, each $i$ in our algorithms corresponds to one day. Combination weights were estimated using an averaging rule on the sub-network topology. 

\begin{figure}
  \centering
  \includegraphics[width=.55\linewidth]{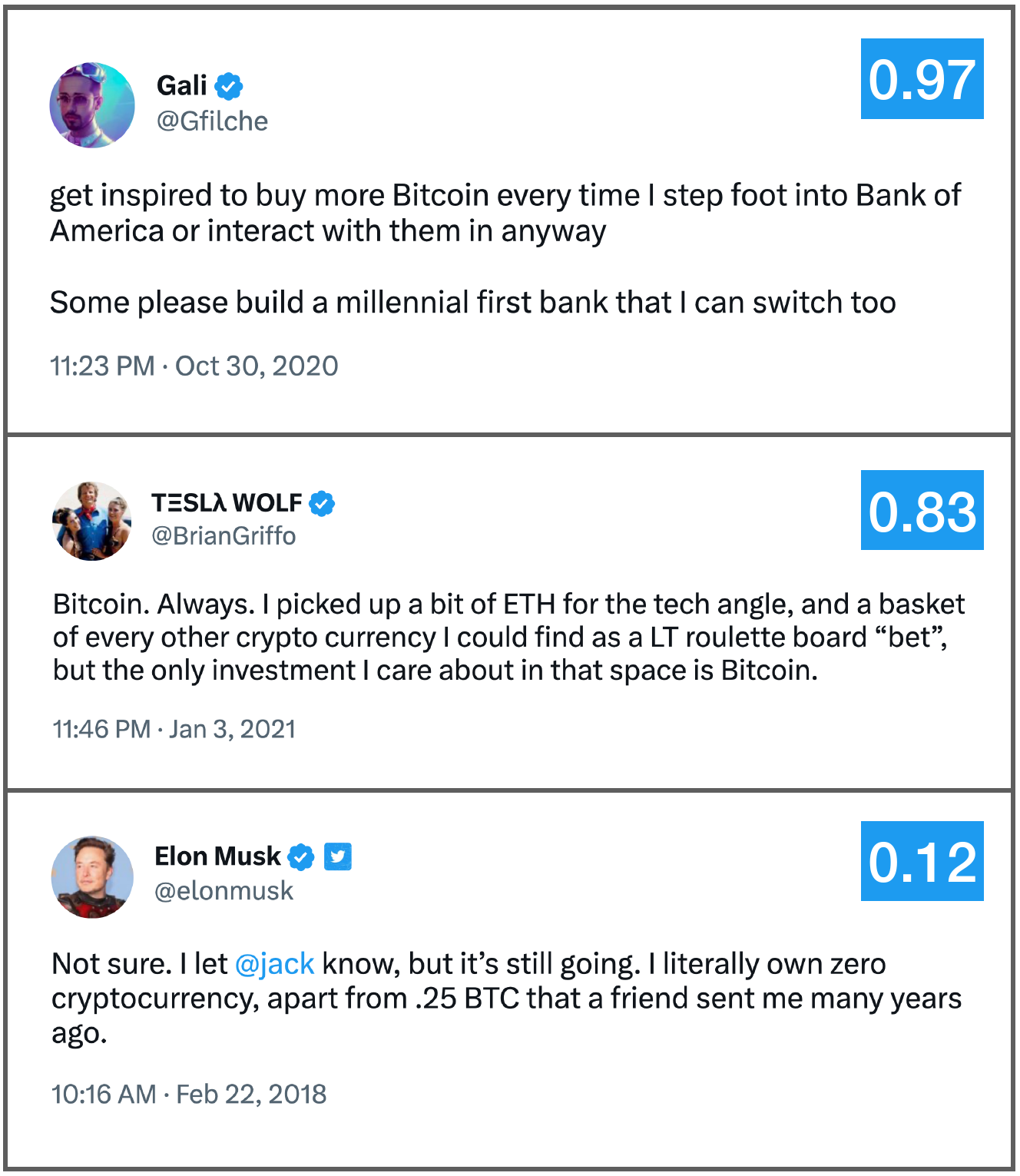}
  \caption{Sample tweets from the users in the sub-network. The number on the upper RHS quantifies the positive attitude towards Bitcoin.}
  \label{fig:tweets_table}
\end{figure}

\begin{figure}
  \centering
  \includegraphics[width=.95\linewidth]{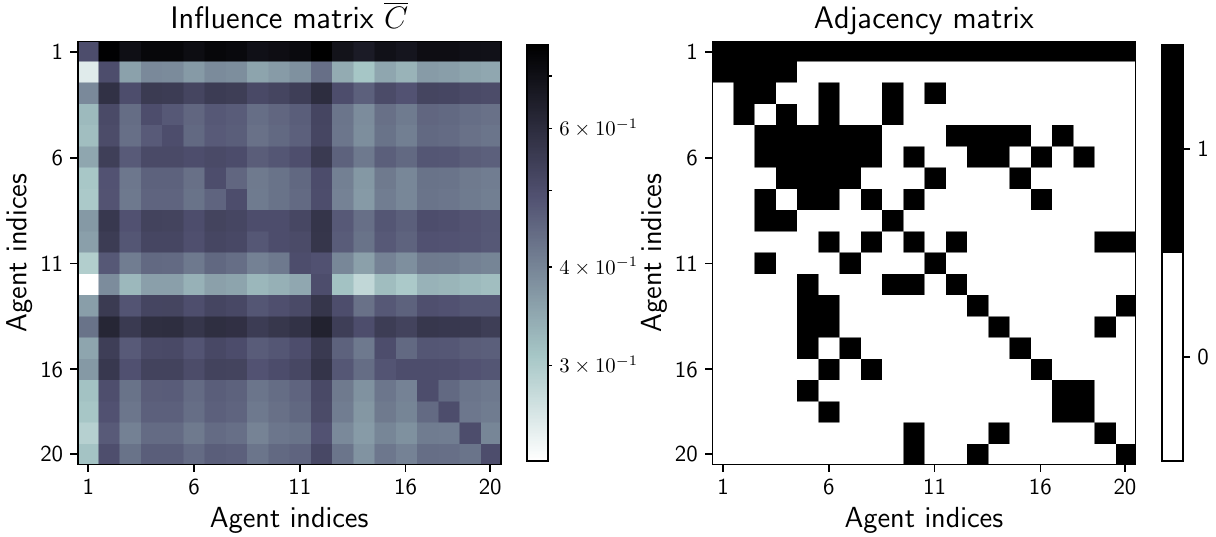}
  \caption{(\emph{Left}): Bipartite causal influence matrix. (\emph{Right}): Adjacency matrix corresponding to the sub-network topology in Fig.~\ref{fig:twitter_network}.}
  \label{fig:inf_adj}
\end{figure}

\noindent \textbf{Bipartite causality.} Inserting the observational input into Alg.~\ref{alg:social_causal_learning}, the resulting average causal derivative effect matrix is shown in Fig.~\ref{fig:inf_adj} in the form of a heat map. To facilitate comparison, we also include the adjacency matrix, which describes the connections between users. In these plots, the indices 1 and 2 correspond to User 1 and User 2, respectively.

Upon observing the heat map, it is evident that User 1 holds significant influence over all other users, as indicated by the high values in the 1th row, which aligns with our expectations. However, notice that the adjacency matrix does not precisely mirror the causal relationships. For instance, User 2 is followed by User 1, yet their influence on User 1, as depicted in the heat map, is relatively low. On the other hand, User 14 exerts a substantial impact on User 2, despite not being directly followed by User 2. This fact may arise from the fact that User 14 holds one of the highest influences on User 1 among all the users in this particular sub-network. These observations highlight the fact that the nature of influence dynamics within real-world social networks cannot simply be explained with direct follower relations.

\begin{figure}
  \centering
  \includegraphics[width=.7\linewidth]{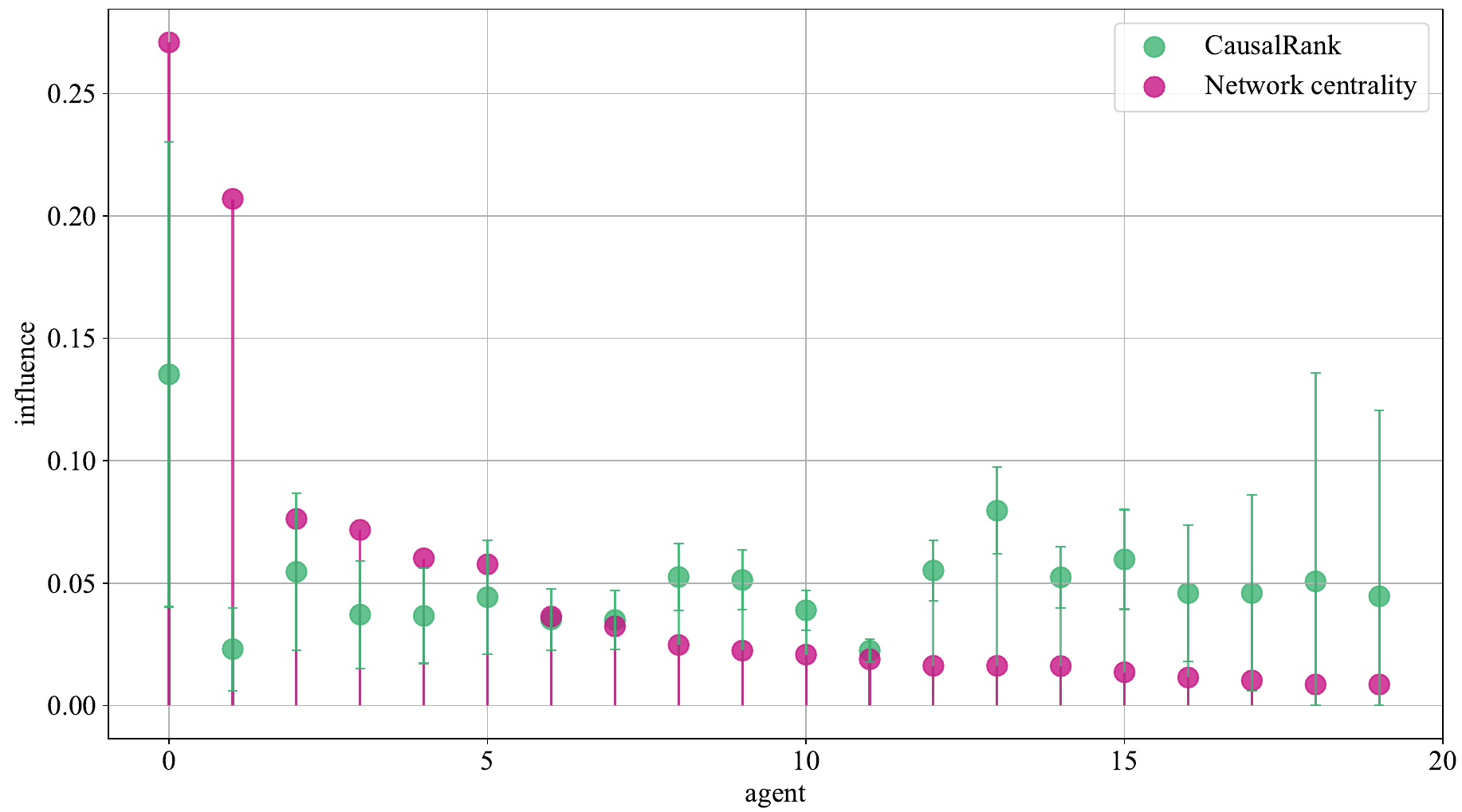}
  \caption{Ranking of agents based on their causal influence and their corresponding network eigenvector centrality. Both ranking scores are summing up to one.}
  \label{fig:ranking}
\end{figure}

\noindent \textbf{Causal impact ranking.} Once the influence matrix is determined, we apply the {\sf \small CausalRank} algorithm (Alg.~\ref{alg:causal_rank}) to rank the agents based on their overall influence within the sub-network. The resulting plot is depicted in Fig.~\ref{fig:ranking}. Notably, User 1 emerges as the most influential agent, aligning with our initial expectations.

However, an intriguing observation can be made regarding User 2. Despite having a high eigenvector centrality, their causal impact score appears relatively small. This phenomenon arises because the causal effect is not solely determined by network centrality but also takes into account the informativeness of the agents. For instance, if a user primarily retweets (reposts) what their neighbors are tweeting, such users tend to possess low informativeness, decreasing their causal impact score. Thus, even though User 2 may have a high centrality within the considered sub-network primarily due to being followed by User 1, their causal influence on their neighbors is low and does not propagate to other users, leading to a relatively small causal impact.

\section{Conclusion}

In this study, we analyzed causal influences among agents that are connected over a network and whose interactions occur over time. Using social learning models, we derived expressions for the causal relationships between pairs of agents. These expressions offer key insights into the diffusion of influences across a social network. We also proposed the {\sf \small CausalRank} algorithm for ranking the overall influence of agents, which allows discovering highly influential actors within a network. Furthermore, to enhance the practical usage of our results, we proposed the graph causality learning algorithm (GCL) that learns the necessary model parameters from raw observational data in order to estimate the causal effects. We demonstrated how GCL can be applied in practice through an application to real social media data.

The social learning models we considered in this work are useful for both modeling opinion formation over social networks as well as for designing multi-agent decision-making systems. Therefore, potential applications range from the analysis of human social networks, such as those on social media platforms, to cooperative decision-making processes of socially intelligent machines like networks of drones or sensors. In addition to these, our results can be useful for applications that involve time-series networked interactions, since they provide insights on the diffusion of influence across graphs.

A key challenge in our experimental evaluation was the lack of definite ground truth for influential agents in real-world social networks.
To address this issue, we conducted our social media experiment on data in which at least the most influential agent was readily identifiable.
We also attempted to design controlled experiments using LLMs, whose capabilities have advanced rapidly since the initial submission of this work.
However, our exploratory results indicated that their inherent biases remain too strong to support fully reliable evaluations of the algorithms at this stage.
We therefore leave this direction for future work.

\section*{Acknowledgments}

The authors thank the anonymous reviewers and the action editor for their insightful and constructive feedback.

\appendix

\section{Connection of \eqref{eq:wcnbd_definition} to Average Causal Derivative Effect }\label{appendix:derivative_effect}

In this appendix, we demonstrate how definition \eqref{eq:wcnbd_definition} for the causal effect summary $\wcnbd$ can also be interpreted as the increment in causal effect $\wcnb$ after an infinitesimal change in the intervention strength $\mu_m(\theta)$. In the literature, this is referred to as the average causal \emph{derivative} effect, as it computes the derivative in causal effect with respect to the intervention strength \cite[Chap. 6]{peters2017elements}. However, a simple derivative of \(\wcnb\) with respect to \(\mu_m\) fails to produce a dose-independent summary, given that \(\wcnb\) is not a linear function of \(\mu_m\), which can be seen from \eqref{eq:cmknb_derivative_causalrank}. Nonetheless, we introduce the function
\begin{equation}\label{eq:fcmknb_def}
    f(\wcnb) \triangleq \dfrac{\wcnb}{1-\wcnb} \stackrel{\eqref{eq:cmknb_derivative_causalrank}}{=}  \!\sum\limits_{\theta \in \Theta \setminus \{ \theta^\circ\}} \dfrac{\mu_{m}(\theta)}{\mu_{m}(\theta^\circ)}  \exp \Big \{ - \Big [\big ((I-R^{\T})^{-1}-I \big ) d_{-m}(\theta) \Big]_k \Big \},
\end{equation}
and notice that $f(\wcnb) \in [0,\infty)$ is a linear function of the intervened belief ratio vector $s_m$, defined as
    \begin{equation}\label{eq:s_m_definition}
        s_m(\theta) \triangleq \frac{\mu_{m}(\theta)}{\mu_{m}(\theta^\circ)}, \quad  s_{m} \triangleq [s_{m}(\theta_1), \dots, s_{m}(\theta_H) ]^{\T}.
    \end{equation}
Here, the vector $s_m$ quantifies the amount of relative misinformation the intervention $\mu_m$ produces. Therefore, the gradient 
\begin{equation}\label{eq:derivative_nbsl_def}
         \nabla_{s_m} f(\wcnb) = \!\!\!\sum_{\theta \in \Theta \setminus \{\theta^\circ\}} \! \frac{\partial f(\wcnb)}{\partial s_m(\theta)} \stackrel{\eqref{eq:fcmknb_def}}{=} \!\!\!\sum\limits_{\theta \in \Theta \setminus \{ \theta^\circ\}} \!\!\! \exp \Big \{ - \Big [\big ((I-R^{\T})^{-1}-I \big ) d_{-m}(\theta) \Big]_k \Big \}
    \end{equation}
is independent of the intervention dose $\mu_m$, and satisfies
\begin{equation}
    \nabla_{s_m} f(\wcnb) = f(\wcnbd).
\end{equation}
In other words, to find the gradient of $f(\wcnb)$ with respect to $s_m$, setting 
\begin{equation}
    \dfrac{\mu_{m}(\theta)}{\mu_{m}(\theta^\circ)} = 1  \Longrightarrow \mu_{m}(\theta) = \frac{1}{H}, \: \forall \theta \in \Theta
\end{equation}
in \eqref{eq:fcmknb_def} as it was done in \eqref{eq:wcnbd_definition} for finding $\wcnbd$, is sufficient. The reason we are interested in $\nabla_{s_m} f(\wcnb)$ is the following: Notice from \eqref{eq:fcmknb_def} that $f(\wcnb)$ is a monotonic increasing function of $\wcnb$. Also, $s_m (\theta)$ is clearly a monotonic function of $\mu_m (\theta)$ --- see \eqref{eq:s_m_definition}. Therefore, $\nabla_{s_m} f(\wcnb)$ can be considered as some proxy for the derivative of the causal effect $\wcnb$ with respect to $\mu_m (\theta)$. Consequently, we conclude that setting $\mu_m(\theta)$ as a uniform belief in $\wcnb$ effectively parallels finding the average causal derivative effect.
It is also important to note here that the average causal derivative effect is the analogue of the binary treatment causal effect. For binary interventions, the causal effect is given by the difference between the outcomes under the two possible interventions. When the intervention varies continuously, the corresponding quantity is the derivative with respect to intervention strength, which captures the infinitesimal change in the causal effect.

\section{Proof of \eqref{eq:lambda_m_infty_asl}}\label{appendix:asl_general}
Theorem~\ref{th:asl_conv_dist}, which establishes convergence in distribution for the pre-intervention ASL case does not require the strongly connected graph assumption. In fact, it holds as long as the observations that the agents receive are i.i.d.\ over time. This condition is satisfied under the post-intervention case as well. Hence, the log-beliefs under intervention converge in distribution, i.e., 
\begin{equation}\label{eq:asl_conv_dist_inter}
         \wbS_{i}(\theta) \idc \beta \sum_{j=1}^\infty (1-\delta)^{j-1} (\widetilde{A}^{\T})^j \wbL_{j}(\theta).
\end{equation}
Then, following the same arguments from Corollary~\ref{cor:expected_logbelief}, the limiting expectation becomes
\begin{equation}\label{eq:limiting_exp_asl_int}
    \widetilde{\lambda}_{\infty}(\theta) = \frac{\beta}{1-\delta} \Big ((I-(1-\delta)\widetilde{A}^{\T})^{-1} - I \Big ) \widetilde{d}(\theta)
\end{equation}
where we define the vector of expected LLRs and the vector of limiting log-belief ratio expectations from across the network:
\begin{equation}
     \widetilde{d}(\theta) \triangleq \Big [\frac \delta \beta \log\frac{\mu_{1}(\theta^{\circ})}{\mu_{1}(\theta)}, d_2(\theta), \dots , d_K(\theta) \Big ]^{\T}, \:  \widetilde{\lambda}_{\infty}(\theta) \triangleq \Big [\log\frac{\mu_{1}(\theta^{\circ})}{\mu_{1}(\theta)}, \lambda_{2,\infty}(\theta), \dots, \lambda_{K,\infty}(\theta) \Big ]^{\T}.
\end{equation}
Using the block matrix form of \(\widetilde{A}\) from \eqref{eq:effective_definitions_block}, we have
\begin{equation}
    I-(1-\delta)\widetilde{A}^{\T} = \left [ \begin{array}{cc}
   \delta  & 0 \\
   -(1-\delta)r  & I-(1-\delta)R^{\T}  \\
\end{array} \right ],
\end{equation}
which implies
\begin{align}
    (I-(1-\delta)\widetilde{A}^{\T})^{-1} = \left [ \begin{array}{cc}
   \dfrac{1}{\delta}  &  0 \\
   \\
   \dfrac{1-\delta}{\delta}(I-(1-\delta)R^{\T})^{-1} r  & (I-(1-\delta)R^{\T})^{-1}  \\
\end{array} \right ].
\end{align}
Inserting this into \eqref{eq:limiting_exp_asl_int}, we arrive at
\begin{align}
    \widetilde{\lambda}_{\infty}(\theta) &= \frac{\beta}{1-\delta} \left [ \begin{array}{cc}
   \dfrac{1-\delta}{\delta}  &  0 \\
   \\
   \dfrac{1-\delta}{\delta}(I-(1-\delta)R^{\T})^{-1} r  & (I-(1-\delta)R^{\T})^{-1}-I  \\
\end{array} \right ] \left [ \begin{array}{c}
     \dfrac{\delta}{\beta}\log\dfrac{\mu_{1}(\theta^{\circ})}{\mu_{1}(\theta)}\\
     \\
     d_{-m} (\theta) 
     \\
\end{array} \right ] 
\end{align}
which again verifies that \(\widetilde{\lambda}_{1,\infty}(\theta) =   \log\dfrac{\mu_{1}(\theta^{\circ})}{\mu_{1}(\theta)} \) and proves relation \eqref{eq:lambda_m_infty_asl}.

\newpage
\section{Proof of \eqref{eq:asl_lambda_fully_connected}}\label{appendix:asl_fully_connected}

Recall from \eqref{eq:fully_connected_a}--\eqref{eq:fully_connected_intervention_a} that for this fully connected network:
\begin{equation}
    A=v\mathds{1}_{K}^{\T}, \quad R= v_{-m}\mathds{1}_{K-1}^{\T}, \quad r=v_1 \mathds{1}_{K-1}.
\end{equation}
As a result, it holds that
\begin{align}\label{eq:fully_asl_inverse}
    (I-(1-\delta)R^{\T})^{-1} &= \sum_{i=0}^{\infty} (1-\delta)^{i}  (R^{\T})^{i} \notag \\
    &\stackrel{\eqref{eq:R_i_expression_fully}}{=} I + (1-\delta) \mathds{1}_{K-1}v_{-m}^{\T} + (1-\delta)^2 (1-v_1) \mathds{1}_{K-1}v_{-m}^{\T} + \dots \notag \\
    &= I+ (1-\delta)  \Big ( \sum_{i=0}^\infty (1-\delta)^i (1-v_1)^i \Big ) \mathds{1}_{K-1}v_{-m}^{\T} \notag \\
    &= I + \frac{1-\delta}{1-(1-\delta)(1-v_1)} \mathds{1}_{K-1}v_{-m}^{\T}.
\end{align}
This implies that
\begin{align}\label{eq:asl_fully_matrix_1}
    (I-(1-\delta)R^{\T})^{-1} r &= \Big ( I + \frac{1-\delta}{1-(1-\delta)(1-v_1)} \mathds{1}_{K-1}v_{-m}^{\T} \Big ) v_1 \mathds{1}_{K-1} \notag \\
    &\stackrel{(a)}{=} v_1 \Big (\mathds{1}_{K-1} + \frac{1-\delta}{1-(1-\delta)(1-v_1)} (1-v_1)  \mathds{1}_{K-1} \Big) \notag \\
    &= \frac{v_1}{1-(1-\delta)(1-v_1)} \mathds{1}_{K-1}
\end{align}
where \( (a) \) follows from the fact that \( \sum_{\ell=2}^{K} v_{\ell} = 1 - v_1\). Eq.~\eqref{eq:fully_asl_inverse} also implies that
\begin{equation}\label{eq:asl_fully_matrix_2}
    \Big ((I-(1-\delta)R^{\T})^{-1}-I \Big )  d_{-m} (\theta) = \frac{1-\delta}{1-(1-\delta)(1-v_1)} \Big ( \sum_{\ell=2}^K v_{\ell} d_{\ell} (\theta)\Big )\mathds{1}_{K-1} 
\end{equation}
Incorporating \eqref{eq:asl_fully_matrix_1} and \eqref{eq:asl_fully_matrix_2} into \eqref{eq:lambda_m_infty_asl} concludes the proof.

\section{Proof of \eqref{eq:asl_lambda_ring}}\label{appendix:asl_ring}

Recall expressions \eqref{eq:ring_A_nbsl} and \eqref{eq:ring_R_nbsl} for \(A\) and \(R\) in the ring special case. Also observe that
\begin{equation}
    r = \Big [1-\alpha, 0, \dots, 0 \Big ]^{\T}.
\end{equation}
Accordingly, it holds that
\begin{equation}
    (I-(1-\delta)R^{\T}) = \left [ \begin{array}{ccccc}
    1-(1-\delta)\alpha & 0   &  \cdots  & 0 & 0 \\
    -(1-\delta)(1-\alpha) & 1-(1-\delta)\alpha   &  \cdots & 0 & 0 \\
    \vdots & \vdots  &  \ddots & \vdots&  \vdots \\
    0 & 0  &  \cdots   &1-(1-\delta)\alpha  & 0 \\
    0 & 0  &  \cdots   &-(1-\delta)(1-\alpha) & 1-(1-\delta)\alpha \\
\end{array} \right ]
\end{equation}
The inverse of this Toeplitz matrix is given by the following lower diagonal matrix:
\begin{align}
    &(I-(1-\delta)R^{\T})^{-1} = \notag \\& \left [ \begin{array}{ccccc}
    \dfrac{1}{1-(1-\delta)\alpha} & 0   &  \cdots  & 0 & 0 \\[12pt]
    \dfrac{(1-\delta)(1-\alpha)}{(1-(1-\delta)\alpha)^2} & \dfrac{1}{1-(1-\delta)\alpha}   &  \cdots & 0 & 0 \\[12pt]
    \vdots & \vdots  &  \ddots & \vdots&  \vdots \\[12pt]
    \dfrac{(1-\delta)^{K-3}(1-\alpha)^{K-3}}{(1-(1-\delta)\alpha)^{K-2}} & \dfrac{(1-\delta)^{K-4}(1-\alpha)^{K-4}}{(1-(1-\delta)\alpha)^{K-3}}  &  \cdots   &\dfrac{1}{1-(1-\delta)\alpha}  & 0 \\[12pt]
    \dfrac{(1-\delta)^{K-2}(1-\alpha)^{K-2}}{(1-(1-\delta)\alpha)^{K-1}} & \dfrac{(1-\delta)^{K-3}(1-\alpha)^{K-3}}{(1-(1-\delta)\alpha)^{K-2}}  &  \cdots   &\dfrac{(1-\delta)(1-\alpha)}{(1-(1-\delta)\alpha)^2} & \dfrac{1}{1-(1-\delta)\alpha} \\[12pt]
\end{array} \right ] .
\end{align}
Then, the matrix-vector products in the general formula \eqref{eq:lambda_m_infty_asl} become
\begin{equation}
     (I-(1-\delta)R^{\T})^{-1}r = \Bigg [ \dfrac{1-\alpha}{1-(1-\delta)\alpha}, \dfrac{(1-\delta)(1-\alpha)^2}{(1-(1-\delta)\alpha)^2}, \dots,  \dfrac{(1-\delta)^{K-2}(1-\alpha)^{K-1}}{(1-(1-\delta)\alpha)^{K-1}}\Bigg ]^{\T}
\end{equation}
and
\begin{equation}
   \Big ((I-(1-\delta)R^{\T})^{-1}-I \Big )  d_{-m} (\theta) = \left [ \begin{array}{c}
   \dfrac{\alpha (1-\delta)}{1-(1-\delta)\alpha} d_2(\theta) \\[12pt]
   \dfrac{\alpha (1-\delta)}{1-(1-\delta)\alpha} d_3(\theta) + \dfrac{(1-\alpha) (1-\delta)}{(1-(1-\delta)\alpha)^2} d_2(\theta) \\[12pt]
   \vdots \\[12pt]
   \dfrac{\alpha (1-\delta)}{1-(1-\delta)\alpha} d_K(\theta)+ \dots + \dfrac{(1-\alpha)^{K-2} (1-\delta)^{K-2}}{(1-(1-\delta)\alpha)^{K-1}} d_2(\theta)
   \end{array} \right]
\end{equation}
Inserting these into the general expression \eqref{eq:lambda_m_infty_asl} concludes the proof.

\section{Proof of Theorem~\ref{th:inform_learning}}\label{appendix:inform_learning}

If we denote the error in estimating the combination matrix by $E \triangleq \widehat{A}-A$, then combining \eqref{eq:gsp_recursion_linear} and \eqref{eq:llr_estimation_M} yields:
\begin{equation}
    \widehat{\bm{D}} = \dfrac{1}{\beta M} \sum_{i=1}^M \Big (\beta \bX_i - (1-\delta) E^{\T} \bLambda_{i-1} \Big) .
\end{equation}
Therefore, the error in informativeness can be decomposed as
\begin{align}\label{eq:inform_proof_decomposition}
    \e \|\widehat{\bm{D}}- D\|_\textup{F}^2 &= \frac{1}{\beta^2 M^2} \e \Big \| \sum_{i=1}^M \beta (\bX_i - D) - (1-\delta) E^{\T} \bLambda_{i-1}  \Big \|_\textup{F}^2 \notag \\
    &= \frac{1}{M^2}  \e \Big \| \sum_{i=1}^M  (\bX_i - D) \Big \|_\textup{F}^2 + \frac{(1-\delta)^2}{\beta^2 M^2} \e \Big \|\sum_{i=1}^M E^{\T} \bLambda_{i-1} \Big \|_\textup{F}^2 \notag \\
    & \qquad \qquad - \frac{2(1-\delta)}{\beta M^2} \e \Bigg [\textup{Tr} \Big (\sum_{i=1}^M (\bX_i - D) \Big ) \Big (\sum_{i=1}^M ( E^{\T} \bLambda_{i-1}) \Big )^{\T} \Bigg ]
\end{align}
where $\|\cdot\|_\textup{F}$ denotes the Frobenius norm and $D$ denotes the informativeness matrix with entries corresponding to $[D]_{kj} \triangleq d_k (\theta_j)$. Here, $(i)$ because of the i.i.d.\ assumption on data over time, and $(ii)$ under sufficiently small $\delta$ values that keep the probability of the event `$\widehat{\btheta}^\circ \neq \theta^\circ$' sufficiently small with increasing $M$ (see Theorem~\ref{th:asl_conv_dist} and also refer to \citep{bordignon_2021}), the first term satisfies
\begin{align}
     \frac{1}{M^2}  \e \Big \| \sum_{i=1}^M  (\bX_i - D) \Big \|_\textup{F}^2 &= \frac{1}{M} \e \big \|   \bX_i - D \big \|_\textup{F}^2 + \frac{2}{M^2} \sum_{i=1}^M \sum_{j < i} \e \textup{Tr} \Big [ \Big ( \bX_i - D \Big ) \Big (\bX_j - D \Big )^{\T} \Big ] \notag \\
     & = \frac{1}{M} \textup{Tr} (\mathcal{R}),
\end{align}
where we also have defined the covariance matrix of the log-likelihood ratios $\mathcal{R}$:
    \begin{equation}
        \mathcal{R} \triangleq \e \Big ( (\bX_i - D) (\bX_i -D)^{\T} \Big ) .
    \end{equation}
Moreover, expanding recursion \eqref{eq:gsp_recursion_linear}, it holds that
\begin{equation}
    \sum_{i=1}^M \bLambda_{i} = \sum_{i=1}^M (1-\delta)^i (A^i)^{\T} \bLambda_0 + \beta \sum_{i=1}^M \sum_{j=0}^{i-1} (1-\delta)^j (A^j)^{\T} \bX_{i-j}.
\end{equation}
Furthermore, by the triangle inequality, it follows that
\begin{align}
    \Big \|\sum_{i=1}^M \bLambda_{i} \Big \|_{\textup{F}} &\leq \Big \| \sum_{i=1}^M (1-\delta)^i (A^i)^{\T} \bLambda_0 \Big \|_\textup{F} + \beta \Big \| \sum_{i=1}^M \sum_{j=0}^{i-1} (1-\delta)^j (A^j)^{\T} \bX_{i-j} \Big \|_\textup{F} \notag \\
    &= O \Big (\min \Big \{M^2,\frac{M}{\delta} \Big \} \Big ) .
\end{align}
Hence, the second term in \eqref{eq:inform_proof_decomposition} satisfies
\begin{align}
     \frac{(1-\delta)^2}{\beta^2 M^2} \e \Big \|\sum_{i=1}^M E^{\T} \bLambda_{i-1} \Big \|_\textup{F}^2 = O \Bigg ( \left \|E \right \|^2 \min \Big \{M^2, \frac{1}{\delta^2} \Big\} \Bigg )
\end{align}
and the third term in \eqref{eq:inform_proof_decomposition} satisfies
\begin{align}
   \frac{(1-\delta)}{\beta M^2} \e \Bigg [\textup{Tr} \Big (\sum_{i=1}^M (\bX_i - D) \Big ) \Big (\sum_{i=1}^M ( E^{\T} \bLambda_{i-1}) \Big )^{\T} \Bigg ] = O \Bigg ( \dfrac{\left \|E \right \|}{M} \min \Big \{M, \frac{1}{\delta} \Big\} \Bigg )
\end{align}
Consequently, for sufficiently small combination matrix errors, namely,
\begin{equation}
    E = o \Big (\max \Big \{M^{-3/2}, \delta M^{-1/2} \Big  \} \Big ),
\end{equation}
it holds that
    \begin{equation}\label{eq:informativeness_error_in_expectation}
       \e \|\widehat{\bm{D}}- D\|_\textup{F}^2 \leq \dfrac{1}{M} \textup{Tr} (\mathcal{R}) + o (1/M  ).
    \end{equation}
Next, recall \eqref{eq:lambda_m_infty_asl} which establishes log-belief ratios under interventions for ASL. For the first term in \eqref{eq:lambda_m_infty_asl}, the matrix estimation error $E$ implies errors for the matrix components $R$ and $r$ that are also proportional to $E$, which means
\begin{align}\label{eq:matrix_error_1}
    &(I-(1-\delta)(R+ O(E))^{\T})^{-1} (r+O(E)) \notag \\
    & \qquad \qquad = (I-(1-\delta)R^{\T})^{-1} \Big (I- \big (I-(1-\delta)R^{\T} \big )^{-1} (1-\delta) O(E) \Big )^{-1} \big (r + O(E) \big ) \notag \\
    & \qquad \qquad \stackrel{(a)}{=} (I-(1-\delta)R^{\T})^{-1} \Big (I + (1-\delta) O(E) \Big ) \big (r + O(E) \big ) \notag \\
    & \qquad \qquad = \Big ((I-(1-\delta)R^{\T})^{-1} + (1-\delta) O (E) \Big ) \big (r + O(E) \big ) \notag \\
    & \qquad \qquad = (I-(1-\delta)R^{\T})^{-1} r + O (E) 
\end{align}
where $(a)$ holds as long as 
\begin{align}\label{eq:e_o_align_exp}
    E &= o \big ((I-(1-\delta)R^{\T})^{-1} \big ) \notag \\ &= o \big (1/ (1-(1-\delta) \rho (R) ) \big ) \notag \\
    & = o \Big (1/ \big (1- \rho (R) + \delta \rho (R) \big ) \Big ) \notag \\ &= o \big (1/ (1-\rho(R)) \big )
\end{align}
In \eqref{eq:e_o_align_exp} we assume that the graph topology is fixed (i.e., it is not changing with decreasing $\delta$). By the same arguments, an error in the matrix expression of the second term of \eqref{eq:lambda_m_infty_asl} implies
\begin{align}\label{eq:matrix_error_2}
    \Big (\big (I-(1-\delta)(R+O(E))^{\T} \big )^{-1}\!-\!I \Big ) & = \Big ((I-(1-\delta)R^{\T})^{-1} + (1-\delta) O (E) - I\Big )  \notag \\
    & = \Big ((I-(1-\delta)R^{\T})^{-1} - I \Big ) + O (E),
\end{align}
and an error in the matrix expression of the pre-intervention case of \eqref{eq:expected_logbelief_asl_idle} implies
\begin{align}\label{eq:matrix_error_3}
    \Big ((I-(1-\delta)\widehat{A}^{\T})^{-1} - I \Big ) &= \Big ((I-(1-\delta)(A+E)^{\T})^{-1} - I \Big ) \notag \\ &= \Big ((I-(1-\delta)A^{\T})^{-1} - I \Big ) + O(E).
\end{align}
Therefore, considering the fact that true causal effect can be written as
\begin{align}\label{eq:appendix_asl_true_cmk}
   \wcasl &= \mu_{k,\infty}(\theta^\circ)-\widetilde{\mu}_{k,\infty}(\theta^\circ) \notag \\
   & \stackrel{\eqref{eq:expected_log_belief_trans}}{=} \dfrac{1}{1 + \!\sum_{\theta \in \Theta \setminus \{ \theta^\circ\}}   \exp \{ - \lambda_{k,\infty}(\theta) \}} - \dfrac{1}{1 + \!\sum_{\theta \in \Theta \setminus \{ \theta^\circ\}}   \exp \{ - \wS_{k,\infty}(\theta) \}} \notag \\
   & \!\!\!\!\stackrel{\eqref{eq:expected_logbelief_asl_idle},\eqref{eq:lambda_m_infty_asl}}{=} \!\!\dfrac{1}{1 + \!\sum_{\theta \in \Theta \setminus \{ \theta^\circ\}}   \exp \Bigg \{ -  \Big [\dfrac{\beta}{1-\delta} \Big ((I-(1-\delta)A^{\T})^{-1} - I \Big ) d(\theta) \Big ]_k \Bigg \}} \notag \\ & \qquad     - \dfrac{1}{1 + \!\!\!\!\!\!\!\sum\limits_{\theta \in \Theta \setminus \{ \theta^\circ\}} \!\!\! \Big (\dfrac{\mu_m(\theta)}{\mu_m(\theta^\circ)} \Big )^{[(I-(1-\delta)R^{\T})^{-1} r]_k }  \!\!\!\exp \Bigg \{\! -  \!\Big  [\dfrac{\beta}{1-\delta}  \Big ((I-(1-\delta)R^{\T})^{-1}\!-\!I \Big )  d_{-m} (\theta) \Big ]_k \Bigg \}},
\end{align}
by \eqref{eq:matrix_error_1}, \eqref{eq:matrix_error_2}, and \eqref{eq:matrix_error_3}, the causal effect estimate satisfies 
\begin{align}\label{eq:appendix_asl_cmk_estimate}
    &\wcasles \notag \\ &=  \!\!\dfrac{1}{1 + \!\sum_{\theta \in \Theta \setminus \{ \theta^\circ\}}   \exp \Bigg \{ -  \Big [\dfrac{\beta}{1-\delta} \Big ((I-(1-\delta)A^{\T})^{-1} - I + O(E) \Big ) \widehat{\bm{d}}(\theta) \Big ]_k \Bigg \}} \notag \\ & \:     - \dfrac{1}{1 + \!\!\!\!\!\!\!\sum\limits_{\theta \in \Theta \setminus \{ \theta^\circ\}} \!\!\! \Big (\dfrac{\mu_m(\theta)}{\mu_m(\theta^\circ)} \Big )^{[(I-(1-\delta)R^{\T})^{-1} r]_k + O(E) }  \!\!\!\exp \Bigg \{\! -  \!\Big  [\dfrac{\beta}{1-\delta}  \Big ((I-(1-\delta)R^{\T})^{-1}\!-\!I + O(E) \Big )  \widehat{\bm{d}}_{-m} (\theta) \Big ]_k \Bigg \}}.
\end{align}
Here, $\widehat{\bm{d}}(\theta)$ and $\widehat{\bm{d}}_{-m} (\theta)$ are the estimates for $d(\theta)$ and $d_{-m} (\theta)$ respectively, which can be formed from informativeness estimation matrix $\widehat{\bm{D}}$. In addition, the scalar function $g$ with vector argument $\lambda$ (over $\Theta \setminus \{ \theta^\circ\}$) defined as
\begin{equation}
    g (\lambda) \triangleq \dfrac{1}{1+ \sum\limits_{\theta \in \Theta \setminus \{ \theta^\circ\}} \exp \{ - \lambda(\theta) \}} 
\end{equation}
has the partial derivatives
\begin{equation}
    \frac{\partial g(\lambda)}{\partial \lambda(\theta)} = \frac{\exp \{ -\lambda(\theta) \}}{\Big (1+\sum\limits_{\theta \in \Theta \setminus \{ \theta^\circ\}} \exp \{ -\lambda(\theta) \} \Big )^2},
\end{equation}
which implies that
\begin{align}\label{eq:l1_norm_smaller}
    \|\nabla g(\lambda)\|_1 &= \sum_{\theta \in \Theta \setminus \{ \theta^\circ\}} \left| \frac{\partial g(\lambda)}{\partial \lambda(\theta)} \right| \notag \\ &= \sum_{\theta \in \Theta \setminus \{ \theta^\circ\}} \frac{\exp \{ -\lambda(\theta) \}}{\Big (1+\sum\limits_{\theta \in \Theta \setminus \{ \theta^\circ\}} \exp \{ -\lambda(\theta) \} \Big )^2} \: \: \leq 1 
\end{align}
By the mean-value theorem, this further implies that for any two vectors $\lambda_1$ and $\lambda_2$, there exists a vector $\lambda_3$ such that
\begin{equation}\label{eq:mvt_naive}
    g(\lambda_1) - g(\lambda_2) = \nabla g(\lambda_3)^{\T} (\lambda_1-\lambda_2).
\end{equation}
Then, using Hölder's inequality on \eqref{eq:mvt_naive} yields that 
\begin{align}\label{eq:contraction_final_result}
    \Big |g(\lambda_1) - g(\lambda_2) \Big | &\leq \big \|\nabla g(\lambda_3) \big \|_1 \big \|\lambda_1-\lambda_2 \big \|_{\infty} \notag \\ & \!\!\stackrel{\eqref{eq:l1_norm_smaller}}{\leq} \big \|\lambda_1-\lambda_2 \big \|_{\infty}.
\end{align}
Note that the norm $\| \cdot \|_{\infty}$ chooses the absolute maximum over all arguments induced from different hypotheses $\theta \in \Theta \setminus \{ \theta^\circ\}$. If we apply \eqref{eq:contraction_final_result} to the differences of first and second terms in \eqref{eq:appendix_asl_true_cmk} and \eqref{eq:appendix_asl_cmk_estimate}, we get
\begin{align}\label{eq:wo_expectation_causal_est_dif}
   &\Big |  \wcasl - \wcasles \Big | \notag \\ &\leq \dfrac{\beta}{1-\delta} \Bigg \|  \Big [ \Big ((I-(1-\delta)A^{\T})^{-1} - I \Big ) d(\theta) - \Big ((I-(1-\delta)A^{\T})^{-1} - I + O(E) \Big ) \widehat{\bm{d}}(\theta) \Big ]_k \Bigg \|_{\infty} \notag \\ &  + \Bigg \| O(E) \log \dfrac{\mu_m(\theta)}{\mu_m(\theta^\circ)} + \dfrac{\beta}{1-\delta}  \!\Big  [  \Big ((I-(1-\delta)R^{\T})^{-1}\!-\!I \Big )  d_{-m} (\theta) - \Big ((I-(1-\delta)R^{\T})^{-1}\!-\!I + O(E) \Big )  \widehat{\bm{d}}_{-m} (\theta) \Big ]_k  \Bigg \|_{\infty} \notag \\
   &\leq \dfrac{\beta}{1-\delta} \Bigg \|  \Big [ \Big ((I-(1-\delta)A^{\T})^{-1} - I \Big ) \Big ( d(\theta) - \widehat{\bm{d}}(\theta) \Big )  + O(E) \widehat{\bm{d}}(\theta) \Big ]_k \Bigg \|_{\infty} \notag \\
   & \quad + \Bigg \| O(E) + \dfrac{\beta}{1-\delta}  \!\Big  [  \Big ((I-(1-\delta)R^{\T})^{-1}\!-\!I \Big ) \! \Big ( d_{-m} (\theta) - \widehat{\bm{d}}_{-m} (\theta) \Big ) - O(E)  \widehat{\bm{d}}_{-m} (\theta) \Big ]_k  \Bigg \|_{\infty} \notag \\
   & \stackrel{(a)}{\leq} \dfrac{\beta}{1-\delta} \Bigg \|  \Big [ \Big ((I-(1-\delta)A^{\T})^{-1} - I \Big ) \Big ( d(\theta) - \widehat{\bm{d}}(\theta) \Big ) \Big ]_k \Bigg \|_{\infty} \notag \\ & \qquad +  \dfrac{\beta}{1-\delta} \Big \| \Big ( (I-(1-\delta)R^{\T})^{-1}-I \Big )  \Big ( d_{-m} (\theta) - \widehat{\bm{d}}_{-m} (\theta) \Big ) \Big  \|_{\infty} + O(E) \notag \\
   & \stackrel{(b)}{\leq} \dfrac{\beta}{1-\delta} \Big \| (I-(1-\delta)A^{\T})^{-1} - I \Big \|_{\infty} \Big \| d(\theta) - \widehat{\bm{d}}(\theta) \Big \|_{\infty} \notag \\ & \qquad +  \dfrac{\beta}{1-\delta} \Big \| (I-(1-\delta)R^{\T})^{-1}-I \Big \|_{\infty}  \Big \| d_{-m} (\theta) - \widehat{\bm{d}}_{-m} (\theta) \Big \|_{\infty} + O(E)
\end{align}
where $(a)$ follows from the triangle inequality, and $(b)$ follows from the sub-multiplicity of the matrix norms. Finally,
since \eqref{eq:informativeness_error_in_expectation} implies
\begin{equation}\label{eq:d_theta_dif_exp_1}
    \e \Big \| d(\theta) - \widehat{\bm{d}}(\theta) \Big \|_{\infty} = O\left(\dfrac{1}{\sqrt{M}} \right)
\end{equation}
and
\begin{equation}\label{eq:d_theta_dif_exp_2}
    \e \Big \| d_{-m} (\theta) - \widehat{\bm{d}}_{-m} (\theta) \Big \|_{\infty}  = O \left ( \dfrac{1}{\sqrt{M}} \right ),
\end{equation}
taking the expectation of both sides in \eqref{eq:wo_expectation_causal_est_dif} and incorporating \eqref{eq:d_theta_dif_exp_1} and \eqref{eq:d_theta_dif_exp_2} concludes that
\begin{equation}
    \e \: \Big |  \wcasl - \wcasles \Big | =  O(1/\sqrt{M}).
\end{equation}
In the case of NBSL, the only source of disagreement between $\wcnb$ and $\wcnbsles$ is at post-intervention, since pre-intervention steady state beliefs are fixed and identical (see \eqref{eq:cmknb_def}), thus their difference equals 0. Also, the proof for the post-intervention difference proceeds along the same lines of the ASL case above.

\section{Discussion of Computational Complexity}\label{appendix:complexity}

First, assuming that the matrix \(A\) and the informativeness vector \(d(\theta)\) for each hypothesis are known, the computational tasks for calculating the derived causal effect expressions \eqref{eq:general_lambda} and \eqref{eq:lambda_m_infty_asl} involve the following:
\begin{itemize}
    \item Calculating \((I - (1-\delta) R^{\T})^{-1}\) for one agent can be achieved in \(O(K^3)\) time with naive matrix operations, where $K$ is the size of the network. Performing this operation for all agents over the network \textit{naively} results in \(O(K^4)\) worst case complexity. However, as we show in the sequel, doing such operations independently for each agent is highly inefficient, and in fact, all calculations can be done with \(O(K^3)\) instead of \(O(K^4)\). 
    \item For the matrix-vector multiplications, the complexity for each agent is \(O(K^2)\), which leads to a complexity of \(O(K^3)\) when we consider all \(K\) agents.
    \item The computation to find the eigenvectors of the $K \times K$ causal influence matrix also involves a computation time of \(O(K^3)\). 
\end{itemize}
Therefore, the total computational complexity with known $A$ and $d(\theta)$ is \(O(K^3)\). 
If, however, \(A\) and \(d(\theta)\) need to be estimated first as in Alg.~\ref{alg:social_causal_learning}, we need to perform the following additional operations:
\begin{itemize}
    \item Estimating \(A\), based on the averaging rule, is \(O(K)\) or \(O(K^2)\) depending on the utilized graph data structure.
    \item Estimating \(d(\theta)\) involves learning from the average over \(M\) observations. Here, each computation involves a matrix multiplication between \(A\) and \(\Lambda\) (i.e., the observational data) which has \(O(K^2)\) complexity.
\end{itemize}
Consequently, considering causal discovery from observational data phase, the total computational complexity for causal ranking of the network is \(O(K^3 + M \cdot K^2)\). 

Now, we return to the calculation of \((I - (1-\delta) R^{\T})^{-1}\) for each agent and discuss how $O(K^3)$ can be obtained by exploiting the fact that the $K$ inverse computations are not independent. 
For each agent $k=1,\dots,K$, let $R_k$ denote the $(K-1)\times (K-1)$ submatrix obtained by deleting the $k$-th row and the $k$-th column of $A$. We treat the ASL case $\delta>0$ and the NBSL case $\delta=0$ separately.

\paragraph{1) ASL:}  Since $\delta>0$, $I-(1-\delta)A^{\T}$ is invertible. Let $G_{\delta}\triangleq (I-(1-\delta)A^{\T})^{-1}$. The Schur complement identity hence implies that
\begin{equation}\label{eq:schurs_identity}
\big(I-(1-\delta)R_k^{\T}\big)^{-1}
=
(G_{\delta})_{-k,-k}
-
\frac{(G_{\delta})_{-k,k}(G_{\delta})_{k,-k}}{[G_{\delta}]_{kk}},
\end{equation}
where $(G_{\delta})_{-k,-k}$ denotes the submatrix of $G_{\delta}$ obtained by deleting the $k$-th row and column, and $(G_{\delta})_{-k,k}$ and $(G_{\delta})_{k,-k}$ denote the corresponding column and row blocks. Hence, it is sufficient to invert the full matrix $I-(1-\delta)A^{\T}$ once, which is $O(K^3)$. Recovering each individual inverse from the above formula costs $O(K^2)$ (the complexity of outer product) and therefore recovering all $K$ inverses costs $O(K^3)+K\cdot O(K^2)=O(K^3)$. 
\paragraph{2) NBSL:} In this case, $\delta=0$ and $I-A^{\T}$ is singular. 
Therefore, Schur complement argument for ASL cannot be applied directly. 
However, all required inverses can still be recovered in total $O(K^3)$ time as we show next. 

\noindent For each agent $k$, $e_k$ denoting the $k$-th canonical basis vector,
\[
B_k \triangleq I-(I-e_ke_k^{\T})A^{\T}
\]
replaces the $k$-th row of $I-A^{\T}$ by $e_k^{\T}$. Furthermore, let $P_k$ be a permutation matrix that moves the $k$-th coordinate to the first position. Then, it holds that
\[
P_k^{\T}B_kP_k
=
\begin{bmatrix}
1 & 0 \\
-r_k^{\T} & I-R_k^{\T}
\end{bmatrix},
\]
where $r_k$ is the $k$-th row of $A$ with its $k$-th entry removed, written as a $(K-1)\times 1$ column vector. Therefore,
\[
(P_k^{\T}B_kP_k)^{-1}
=
\begin{bmatrix}
1 & 0 \\
(I-R_k^{\T})^{-1}r_k^{\T} & (I-R_k^{\T})^{-1}
\end{bmatrix},
\]
and thus $(I-R_k^{\T})^{-1}$ is exactly the lower-right $(K-1)\times (K-1)$ block of $(P_k^{\T}B_kP_k)^{-1}$. Now, if we fix a reference agent, say $k=1$ as before, computing $B_1^{-1}$ once is $O(K^3)$. For any other agent $k$, by definition, we have
\[
B_k
=
B_1 + e_k(Ae_k)^{\T}-e_1(Ae_1)^{\T}
=
B_1 + U_kV_k^{\T},
\]
where
\[
U_k \triangleq \big[\, e_k \;\; -e_1 \,\big],
\qquad
V_k \triangleq \big[\, Ae_k \;\; Ae_1 \,\big].
\]
Thus, $B_k$ differs from $B_1$ by a rank-$2$ update. By the Woodbury identity,
\[
B_k^{-1}
=
B_1^{-1}
-
B_1^{-1}U_k
\big(I+V_k^{\T}B_1^{-1}U_k\big)^{-1}
V_k^{\T}B_1^{-1}.
\]
Since $U_k$ and $V_k$ each have only two columns, computing $B_k^{-1}$ from $B_1^{-1}$ costs $O(K^2)$ per agent. Hence, after one initial $O(K^3)$ inversion for the reference agent, all $K$ matrices $\{(I-R_k^{\T})^{-1}\}_{k=1}^K$ can be recovered in total
\[
O(K^3)+K\cdot O(K^2)=O(K^3)
\]
time. 

\section{Additional Experiments}
\label{appendix:additional_experiments}
In this section, we report additional experiments that complement the results in Section~\ref{sec:computer_sims}.

\paragraph{Robustness to follower-inflation attacks} 
To test whether the robustness observed in the main text persists beyond the already considered setting, we repeat the follower-inflation experiment on Erd\H{o}s--Renyi networks of different sizes. For each $K \in \{20,40,60,80\}$, we generate $50$ independent Erd\H{o}s--R\'enyi networks with edge probability $0.3$. In each network, one agent is selected uniformly at random as the attacked user and we progressively add bot accounts that act as artificial followers of this agent. 

The results are shown in Figs.~\ref{fig:bot_attack_raw_rank_20}--\ref{fig:bot_attack_raw_rank_80}. Across all considered network sizes, the behavior is consistent with that reported in the main text. In contrast to the other methods, {\sf \small CausalRank} remains stable as bots are added. 
These experiments therefore prove that the robustness of {\sf \small CausalRank} is not specific to a single network size, but persists across substantially larger Erd\H{o}s--R\'enyi graphs as well.

\begin{figure}
  \centering
  \includegraphics[width=.63\linewidth]{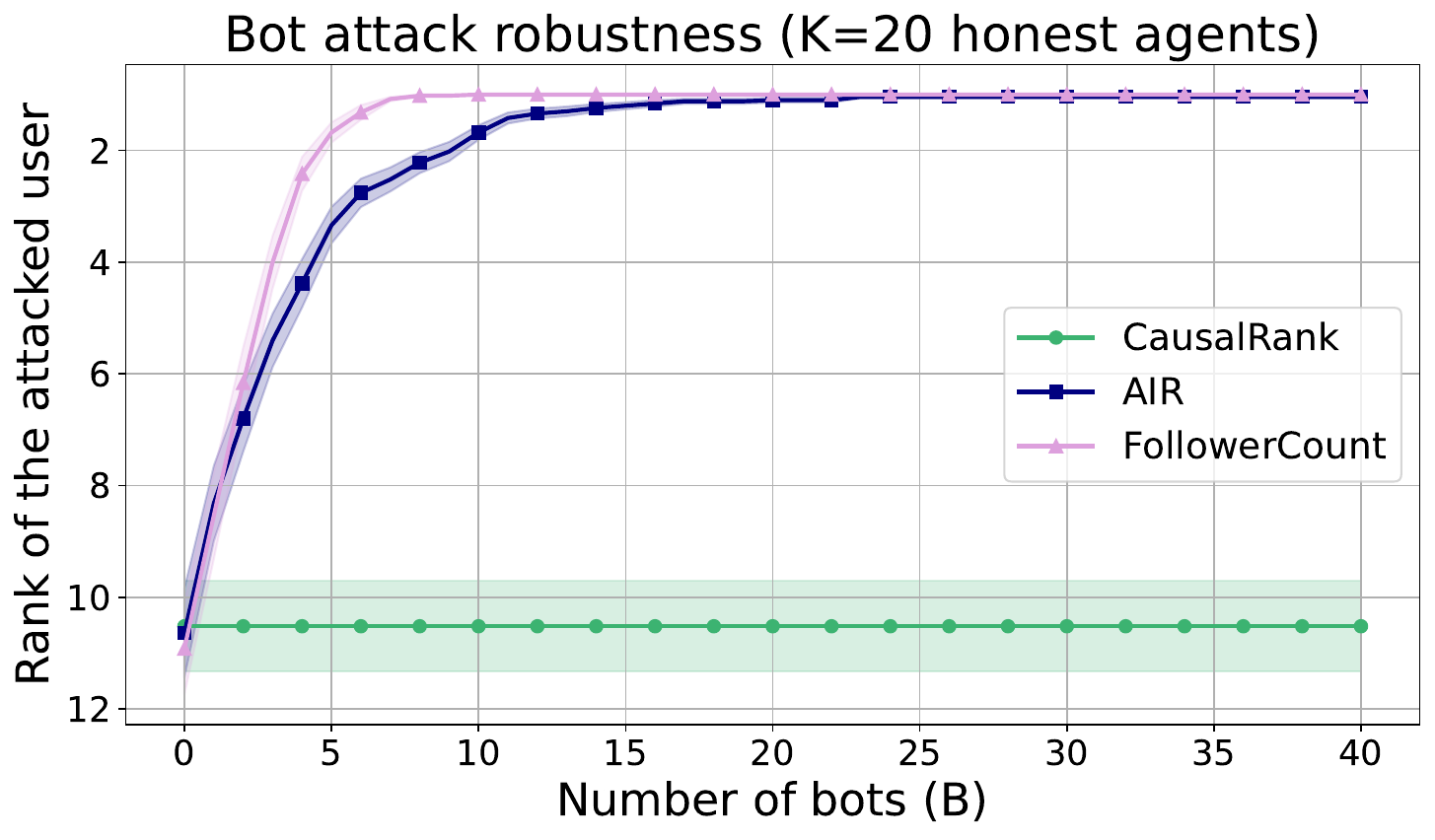}
  \caption{Average rank of the attacked agent versus the number of added bots for Erd\H{o}s--R\'enyi networks with $K=20$ and edge probability $0.3$, averaged over $50$ independent realizations.}
  \label{fig:bot_attack_raw_rank_20}
\end{figure}

\begin{figure}
  \centering
  \includegraphics[width=.63\linewidth]{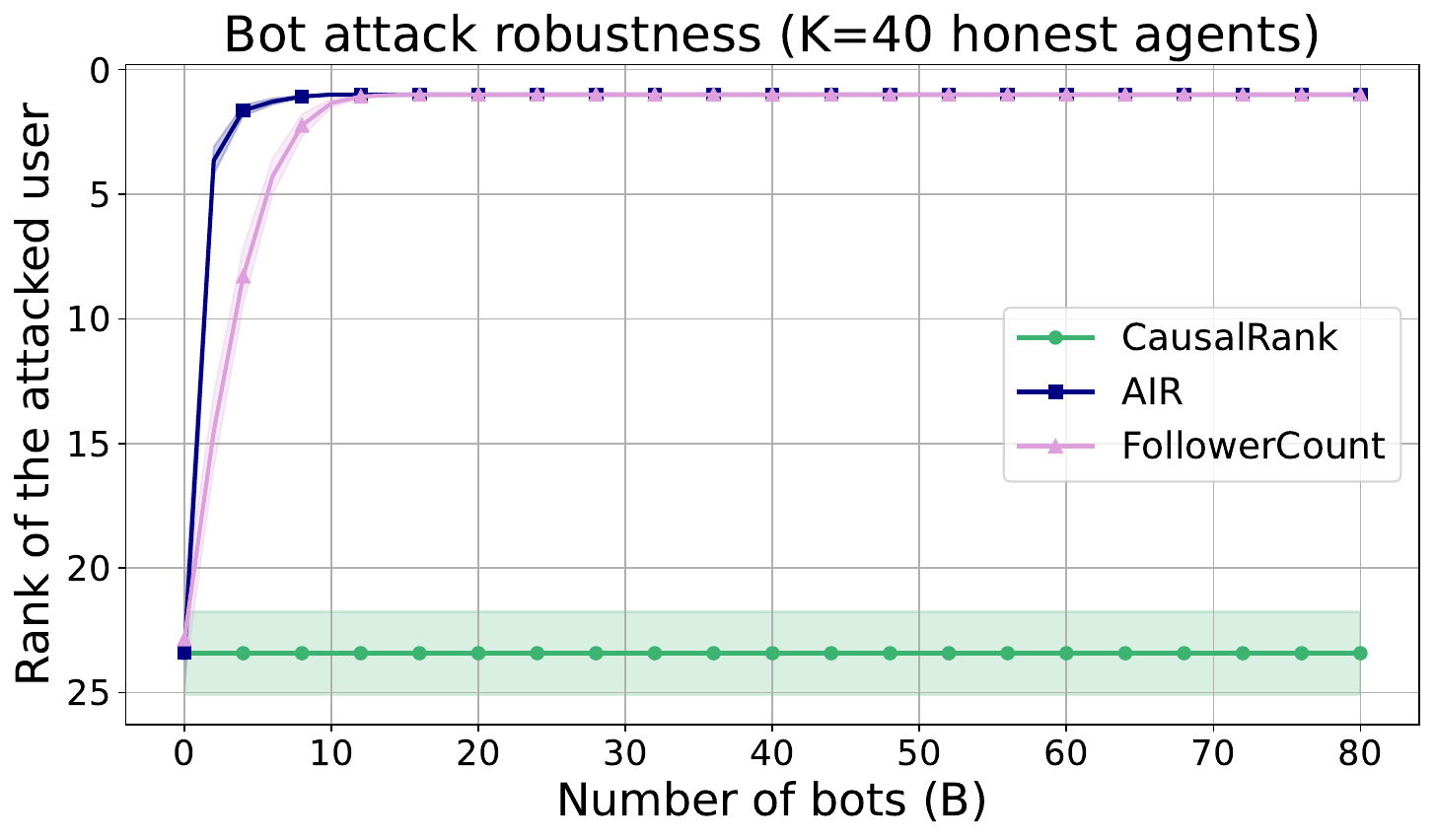}
  \caption{Average rank of the attacked agent versus the number of added bots for Erd\H{o}s--R\'enyi networks with $K=40$ and edge probability $0.3$, averaged over $50$ independent realizations.}
  \label{fig:bot_attack_raw_rank_40}
\end{figure}

\begin{figure}
  \centering
  \includegraphics[width=.63\linewidth]{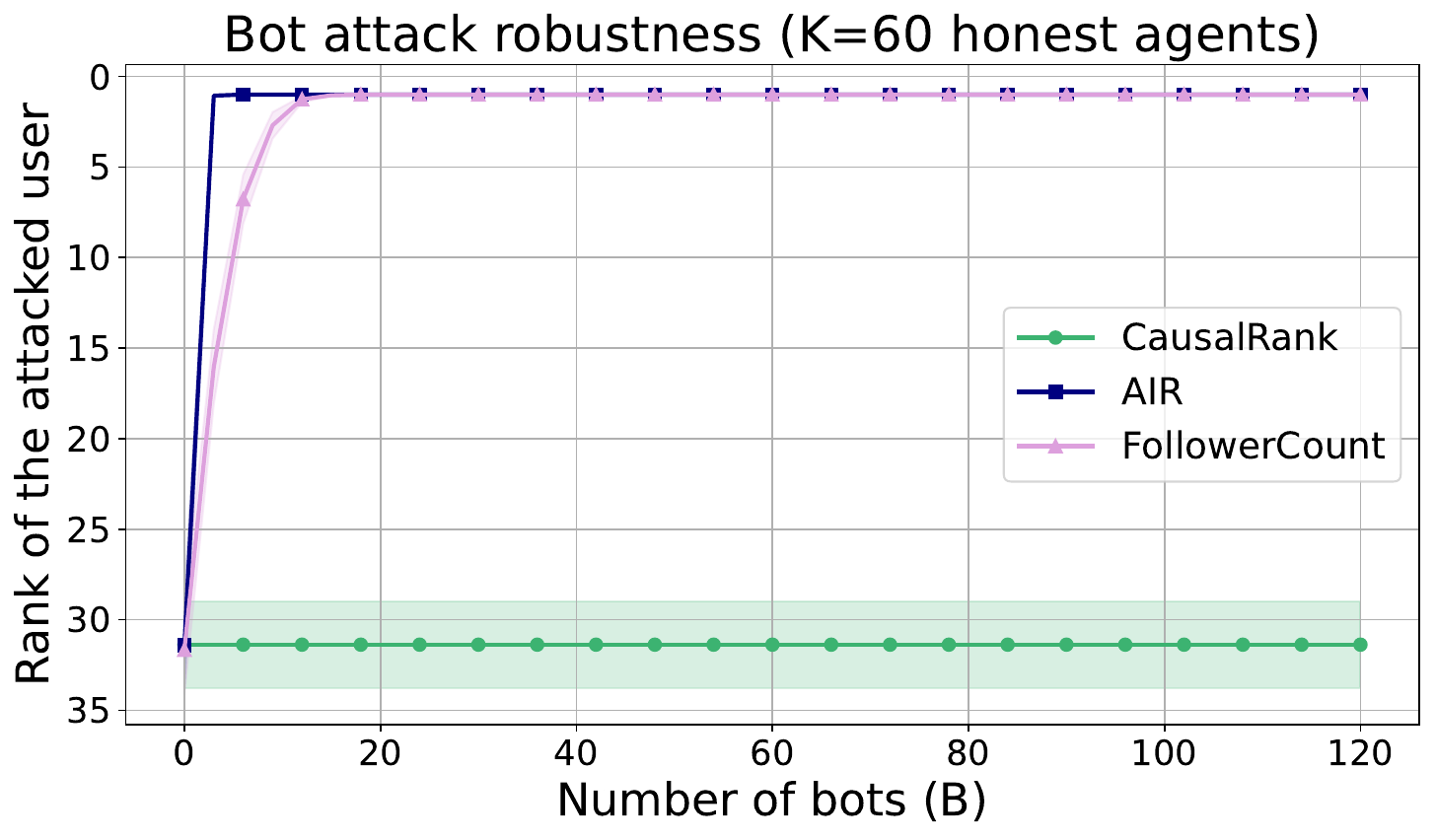}
  \caption{Average rank of the attacked agent versus the number of added bots for Erd\H{o}s--R\'enyi networks with $K=60$ and edge probability $0.3$, averaged over $50$ independent realizations.}
  \label{fig:bot_attack_raw_rank_60}
\end{figure}

\begin{figure}
  \centering
  \includegraphics[width=.63\linewidth]{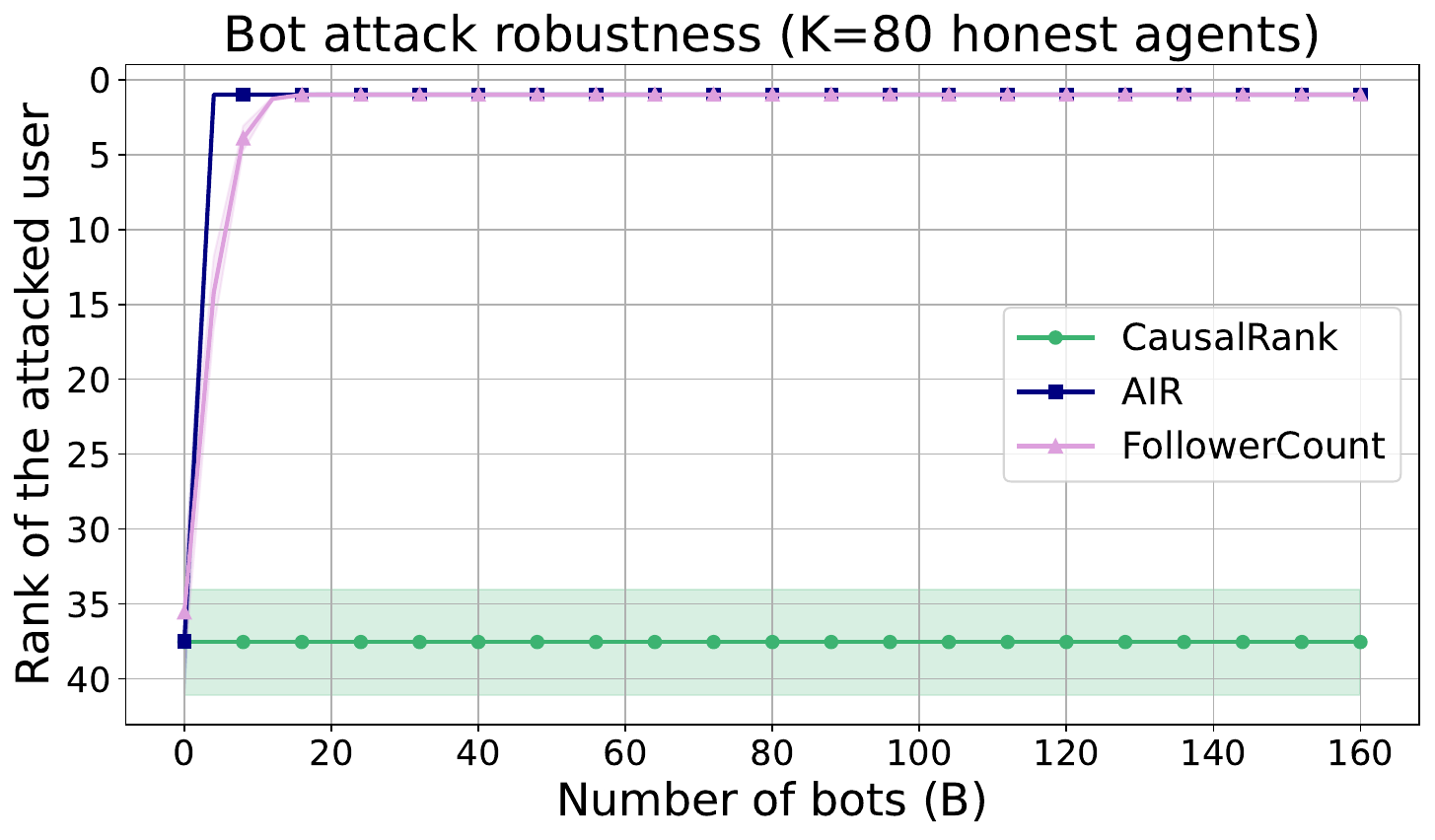}
  \caption{Average rank of the attacked agent versus the number of added bots for Erd\H{o}s--R\'enyi networks with $K=80$ and edge probability $0.3$, averaged over $50$ independent realizations.}
  \label{fig:bot_attack_raw_rank_80}
\end{figure}

\paragraph{CPU-time comparisons} 
We next report additional CPU-timing results for the computation of the causal influence matrix and the rankings. 
In Fig.~\ref{fig:cpu_time_asl}, we show the total CPU time for ASL with $\delta = 0.1$ as a function of network size $K$, which shows that the behavior with NBSL is present also with ASL. 

In Figs.~\ref{fig:cpu_comparison_nbsl} and \ref{fig:cpu_comparison_asl}, we compare the naive $O(K^4)$ implementation with the accelerated $O(K^3)$ formulation derived in Appendix~\ref{appendix:complexity}.
As we discussed in Appendix~\ref{appendix:complexity}, the naive implementation computes $K$ independent matrix inversions, whereas the accelerated version recovers all $K$ required inverses from a single $K \times K$ inversion by exploiting the Woodbury identity in the NBSL case ($\delta = 0$) and the Schur's complement in the ASL case ($\delta > 0$). The empirical results verify that the speedup relative to the naive approach can be substantial (reaching to approximately $10\times$ speedup by $K=80$).

\begin{figure}
  \centering
  \includegraphics[width=.62\linewidth]{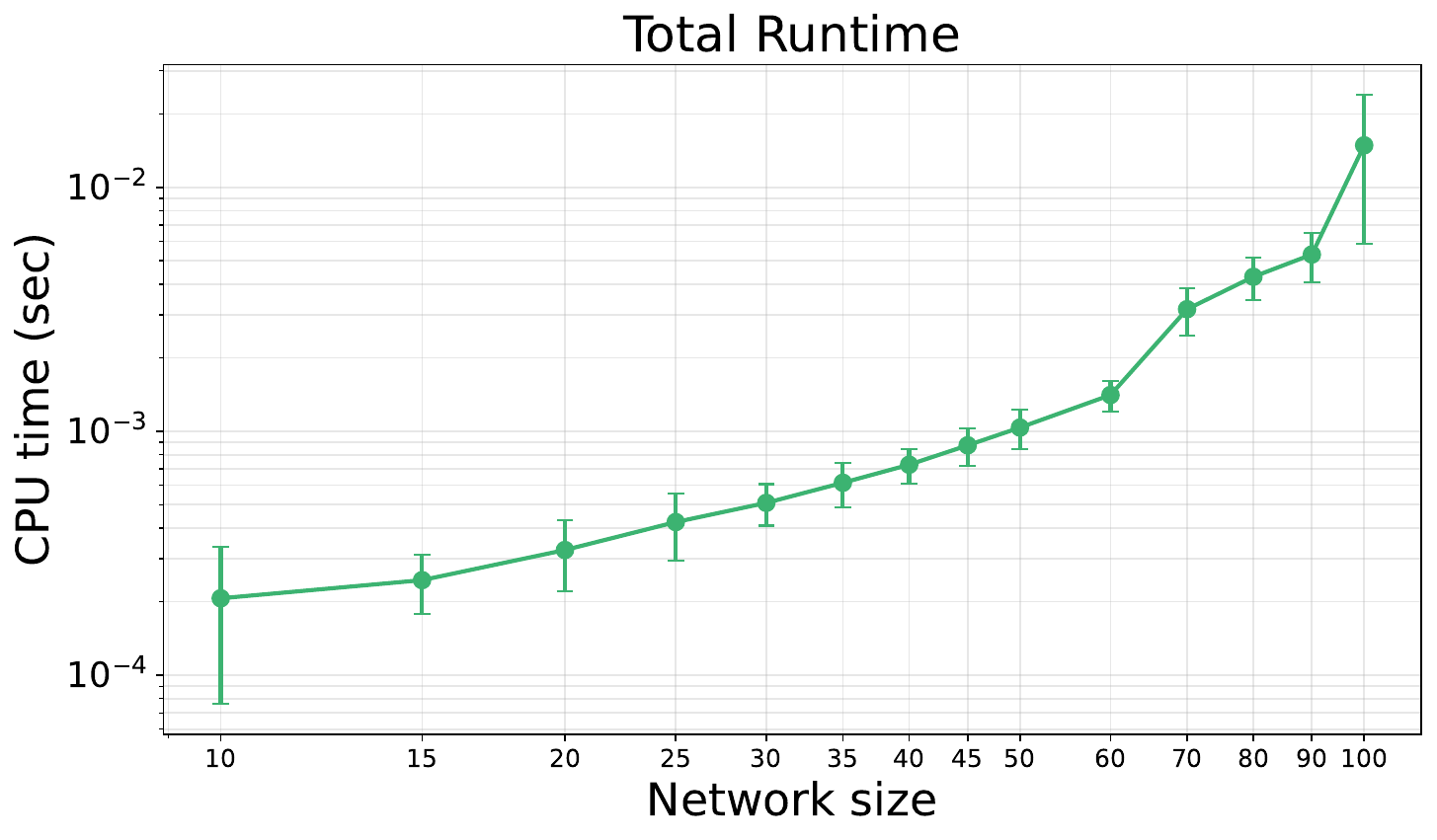}
  \caption{Total CPU time for ASL with $\delta=0.1$ as a function of network size $K$.}
  \label{fig:cpu_time_asl}
\end{figure}

\begin{figure}
  \centering
  \includegraphics[width=.62\linewidth]{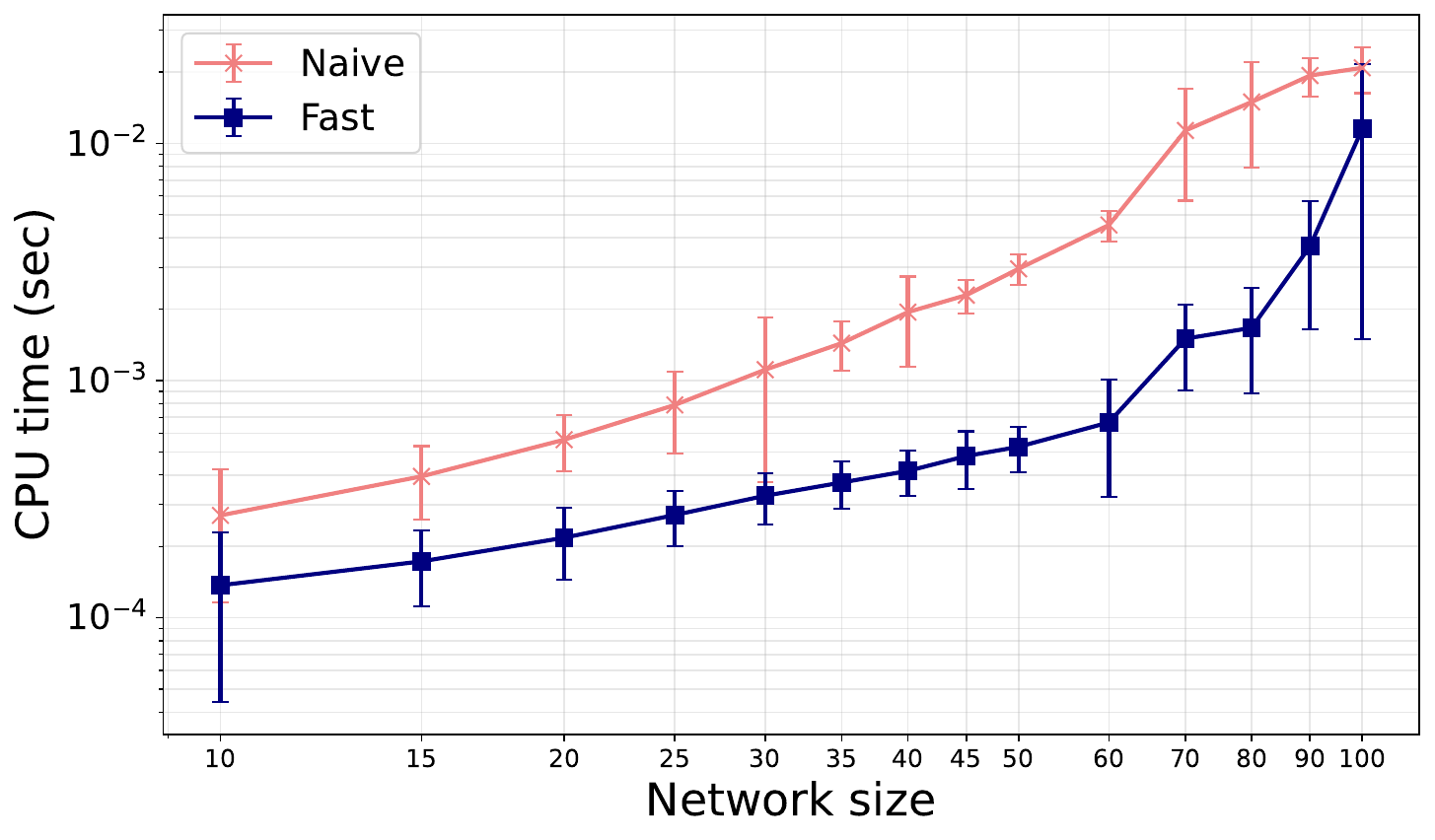}
  \caption{CPU-time comparison between the naive $O(K^4)$ and accelerated $O(K^3)$ implementations for NBSL.}
  \label{fig:cpu_comparison_nbsl}
\end{figure}

\begin{figure}
  \centering
  \includegraphics[width=.62\linewidth]{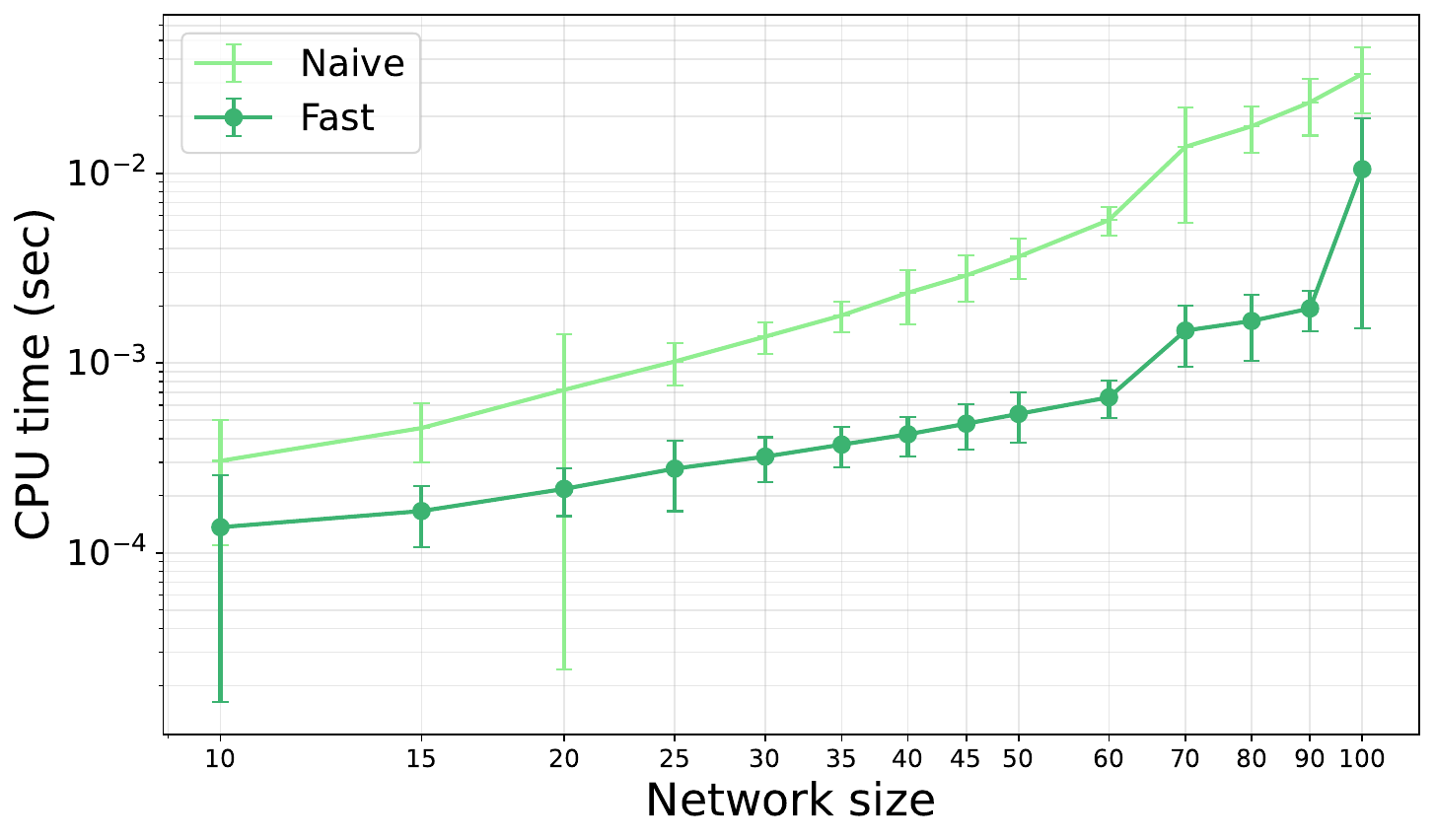}
  \caption{CPU-time comparison between the naive $O(K^4)$ and accelerated $O(K^3)$ implementations for ASL.}
  \label{fig:cpu_comparison_asl}
\end{figure}

\newpage

\bibliography{sample}

@article{hunter2022optimizing,
  title={Optimizing opinions with stubborn agents},
  author={Hunter, D. and Zaman, T.},
  journal={Operations Research},
  volume={70},
  number={4},
  pages={2119--2137},
  year={2022}
}

@book{resnick2013adventures,
  title={Adventures in Stochastic Processes},
  author={Resnick, S. I.},
  year={2013},
  publisher={Springer}
}

@ARTICLE{smith2014,
  author={Smith, S. T. and Kao, E. K. and Senne, K. D. and Bernstein, G. and Philips, S.},
  journal={IEEE Transactions on Signal Processing}, 
  title={Bayesian Discovery of Threat Networks}, 
  year={2014},
  volume={62},
  number={20},
  pages={5324-5338}}

@article{ghaderi2014,
	author = {J. Ghaderi and R. Srikant},
	journal = {Automatica},
	number = {12},
	pages = {3209-3215},
	title = {Opinion dynamics in social networks with stubborn agents: Equilibrium and convergence rate},
	volume = {50},
	year = {2014}}

@article{acemoglu2010spread,
  title={Spread of (mis) information in social networks},
  author={Acemoglu, D. and Ozdaglar, A. and ParandehGheibi, A.},
  journal={Games and Economic Behavior},
  volume={70},
  number={2},
  pages={194--227},
  year={2010},
  publisher={Elsevier}
}

@INPROCEEDINGS{li2019_malicious,
  author={Li, Y. and Qiu, B. and Chen, Y. and Zhao, H. V.},
  booktitle={Proc. IEEE ICASSP}, 
  title={Analysis of Information Diffusion with Irrational Users: A Graphical Evolutionary Game Approach}, 
  year={2019},
  volume={},
  number={},
  pages={2527-2531}}

@article{molavi2018theory,
  title={A theory of non-{Bayesian} social learning},
  author={Molavi, P. and Tahbaz-Salehi, A. and Jadbabaie, A.},
  journal={Econometrica},
  volume={86},
  number={2},
  pages={445--490},
  year={2018}
}

@inproceedings{loureiro2022timelms,
  title={{TimeLMs}: Diachronic Language Models from {Twitter}},
  author={Loureiro, D. and Barbieri, F. and Neves, L. and Anke, L. E. and Camacho-Collados, J.},
  booktitle={Proc. Annual Meeting of the Association for Computational Linguistics},
  pages={251--260},
  year={2022}
}

@inproceedings{twitter_quercia2011mood,
  title={In the mood for being influential on {T}witter},
  author={Quercia, D. and Ellis, J. and Capra, L. and Crowcroft, J.},
  booktitle={Proc. IEEE International Conference on Social Computing},
  pages={307--314},
  year={2011},
  address={Boston, MA}
}

@article{twitter_influence_measure,
title = {Identification of influential users on {T}witter: A novel weighted correlated influence measure for {Covid-19}},
journal = {Chaos, Solitons \& Fractals},
volume = {139},
pages = {1--8},
year = {2020},
author = {S. Jain and A. Sinha},
}

@article{lennart_bitcoin_musk,
	author = {L. Ante},
	doi = {https://doi.org/10.1016/j.techfore.2022.122112},
	journal = {Technological Forecasting and Social Change},
	pages = {1--14},
	title = {How {Elon Musk's} {Twitter} activity moves cryptocurrency markets},
	volume = {186},
	year = {2023}}

@article{sussman2017elements,
  title={Elements of estimation theory for causal effects in the presence of network interference},
  author={Sussman, D. L. and Airoldi, E. M.},
  journal={arXiv:1702.03578},
  year={2017}
}

@article{agarwal2022network,
  title={Network Synthetic Interventions: A Framework for Panel Data with Network Interference},
  author={Agarwal, A. and Cen, S. and Shah, D. and Yu, C. L.},
  journal={arXiv:2210.11355},
  year={2022}
}

@article{wright1934method,
  title={The method of path coefficients},
  author={Wright, S.},
  journal={The Annals of Mathematical Statistics},
  volume={5},
  number={3},
  pages={161--215},
  year={1934}
}

@article{lee2023finding,
  title={Finding influential subjects in a network using a causal framework},
  author={Lee, Y. and Buchanan, A. L. and Ogburn, E. L. and Friedman, S. R. and Halloran, M. E. and Katenka, N. V. and Wu, J. and Nikolopoulos, G.},
  journal={Biometrics},
  year={2023},
  pages={1--13}
}

@inproceedings{toulis2013estimation,
  title={Estimation of causal peer influence effects},
  author={Toulis, P. and Kao, E.},
  booktitle={Proc. ICML},
  pages={1489--1497},
  year={2013}
}

@INPROCEEDINGS{smith2018,
  author={Smith, S. T. and Kao, E. K. and Shah, D. C. and Simek, O. and Rubin, D. B.},
  booktitle={IEEE Statistical Signal Processing Workshop (SSP)}, 
  title={Influence Estimation on Social Media Networks Using Causal Inference}, 
  year={2018},
  volume={},
  number={},
  pages={328-332}}

@article{eichler2010granger,
  title={On {Granger} causality and the effect of interventions in time series},
  author={Eichler, M. and Didelez, V.},
  journal={Lifetime data analysis},
  volume={16},
  pages={3--32},
  year={2010}
}

@InProceedings{soni2019,
  title = 	 {Detecting Social Influence in Event Cascades by Comparing Discriminative Rankers},
  author =       {Soni, S. and Ramirez, S. L. and Eisenstein, J. J.},
  booktitle = 	{Proc. Machine Learning Research},
  pages = 	 {78--99},
  year = 	 {2019},
  editor = 	 {},
  volume = 	 {104}
  }

@InProceedings{sridhar2022,
  title = 	 {Estimating Social Influence from Observational Data},
  author =       {Sridhar, D. and Bacco, C. D. and Blei, D.},
  booktitle = 	 {Proc. Conference on Causal Learning and Reasoning},
  pages = 	 {712--733},
  year = 	 {2022},
  volume = 	 {177},
  month = 	 {11--13 Apr}
}

@article{bond2012,
  title={A 61-million-person experiment in social influence and political mobilization},
  author={Bond, R. M. and Fariss, C. J. and Jones, J. J. and Kramer, A. D.I. and Marlow, C. and Settle, J. E. and Fowler, J. H.},
  journal={Nature},
  volume={489},
  number={7415},
  pages={295--298},
  year={2012}
}

@article{aral2012identifying,
  title={Identifying influential and susceptible members of social networks},
  author={Aral, S. and Walker, D.},
  journal={Science},
  volume={337},
  number={6092},
  pages={337--341},
  year={2012}
}

@ARTICLE{shah2011,
  author={Shah, D. and Zaman, T.},
  journal={IEEE Transactions on Information Theory}, 
  title={Rumors in a Network: Who's the Culprit?}, 
  year={2011},
  volume={57},
  number={8},
  pages={5163-5181}}

@article{lagree2018algorithms,
  title={Algorithms for online influencer marketing},
  author={Lagr{\'e}e, P. and Capp{\'e}, O. and Cautis, B. and Maniu, S.},
  journal={ACM Transactions on Knowledge Discovery from Data},
  volume={13},
  number={1},
  pages={1--30},
  year={2018}
}

@article{smith2021automatic,
  title={Automatic detection of influential actors in disinformation networks},
  author={Smith, S. T. and Kao, E. K. and Mackin, E. D. and Shah, D. C. and Simek, O. and Rubin, D. B.},
  journal={Proc. National Academy of Sciences},
  volume={118},
  number={4},
  year={2021}
}

@article{banerjee2013diffusion,
  title={The diffusion of microfinance},
  author={Banerjee, A. and Chandrasekhar, A. G. and Duflo, E. and Jackson, M. O.},
  journal={Science},
  volume={341},
  number={6144},
  pages={1236498},
  year={2013}
}

@inproceedings{kempe2003maximizing,
  title={Maximizing the spread of influence through a social network},
  author={Kempe, D. and Kleinberg, J. and Tardos, {\'E}.},
  booktitle={Proc. ACM SIGKDD},
  pages={137--146},
  year={2003}
}

@article{golub2010naive,
  title={Naive learning in social networks and the wisdom of crowds},
  author={Golub, B. and Jackson, M. O.},
  journal={American Economic Journal: Microeconomics},
  volume={2},
  number={1},
  pages={112--149},
  year={2010}
}

@article{dablander2019node,
  title={Node centrality measures are a poor substitute for causal inference},
  author={Dablander, F. and Hinne, M.},
  journal={Scientific Reports},
  volume={9},
  number={1},
  pages={6846},
  year={2019}}

@article{chandrasekhar2020testing,
  title={Testing models of social learning on networks: Evidence from two experiments},
  author={Chandrasekhar, A. G. and Larreguy, H. and Xandri, J. P.},
  journal={Econometrica},
  volume={88},
  number={1},
  pages={1--32},
  year={2020}}

@article{granger1969investigating,
  title={Investigating causal relations by econometric models and cross-spectral methods},
  author={Granger, C. W. J.},
  journal={Econometrica},
  pages={424--438},
  year={1969}}

@article{rubin1974estimating,
  title={Estimating causal effects of treatments in randomized and nonrandomized studies},
  author={Rubin, D. B.},
  journal={Journal of Educational Psychology},
  volume={66},
  number={5},
  pages={688},
  year={1974},
  publisher={American Psychological Association}
}

@article{brin1998,
title = {The anatomy of a large-scale hypertextual {Web} search engine},
journal = {Computer Networks and ISDN Systems},
volume = {30},
number = {1},
pages = {107-117},
year = {1998},
author = {S. Brin and L. Page}
}

@INPROCEEDINGS{hare2019,
  author={Hare, J. Z. and Uribe, C. A. and Kaplan, L. M. and Jadbabaie, A.},
  booktitle={Proc. Int. Conf. on Information Fusion}, 
  title={On Malicious Agents in Non-{Bayesian} Social Learning with Uncertain Models}, 
  year={2019},
  volume={},
  number={},
  pages={1-8}}

@book{easley2010networks,
  title={Networks, Crowds, and Markets: Reasoning About a Highly Connected World},
  author={Easley, D. and Kleinberg, J.},
  year={2010},
  publisher={Cambridge University Press}
}

@ARTICLE{valentina2023discovering,
  author={Shumovskaia, V. and Kayaalp, M. and Cemri, M. and Sayed, A. H.},
  journal={IEEE Open Journal of Signal Processing}, 
  title={Discovering Influencers in Opinion Formation Over Social Graphs}, 
  year={2023},
  volume={4},
  pages={188-207}
  }

@article{krishnamurthy2022dynamics,
  title={Dynamics of Social Networks: Multi-agent Information Fusion, Anticipatory Decision Making and Polling},
  author={Krishnamurthy, V.},
  journal={arXiv:2212.13323},
  month={Dec.},
  year={2022}
}

@article{eichler2012causal,
  title={Causal inference in time series analysis},
  author={Eichler, M.},
  journal={Causality: Statistical Perspectives and Applications},
  pages={327--354},
  year={2012},
  publisher={Wiley Online Library}
}

@article{mossel2015strategic,
  title={Strategic learning and the topology of social networks},
  author={Mossel, E. and Sly, A. and Tamuz, O.},
  journal={Econometrica},
  volume={83},
  number={5},
  pages={1755--1794},
  year={2015}
  }

@article{pillai2005perron,
  title={The {Perron-Frobenius} theorem: some of its applications},
  author={Pillai, S. U. and Suel, T. and Cha, S.},
  journal={IEEE Signal Processing Magazine},
  volume={22},
  number={2},
  pages={62--75},
  year={2005},
  publisher={IEEE}
}

@article{molavi2013reaching,
  title={Reaching consensus with increasing information},
  author={Molavi, P. and Jadbabaie, A. and Rad, K. R. and Tahbaz-Salehi, A.},
  journal={IEEE Journal of Selected Topics in Signal Processing},
  volume={7},
  number={2},
  pages={358--369},
  year={2013},
  publisher={IEEE}
}

@article{ying2016information,
  title={Information exchange and learning dynamics over weakly connected adaptive networks},
  author={Ying, B. and Sayed, A. H.},
  journal={IEEE Transactions on Information Theory},
  volume={62},
  number={3},
  pages={1396--1414},
  year={2016},
  publisher={IEEE}
}

@book{woodward2004book,
    author = {Woodward, J.},
    title = "{Making Things Happen: A Theory of Causal Explanation}",
    publisher = {Oxford University Press},
    year = {2004}
}

@book{pearl2018book,
  title={The Book of Why: The New Science of Cause and Effect},
  author={Pearl, J. and Mackenzie, D.},
  year={2018},
  publisher={Basic Books}
}

@book{peters2017elements,
  title={Elements of Causal Inference: Foundations and Learning Algorithms},
  author={Peters, J. and Janzing, D. and Sch{\"o}lkopf, B.},
  year={2017},
  publisher={The MIT Press}
}

@book{pearl2009causality,
  title={Causality},
  author={Pearl, J.},
  year={2009},
  publisher={Cambridge University Press}
}

@article{zellner1988,
 author = {A. Zellner},
 journal = {The American Statistician},
 number = {4},
 pages = {278--280},
 title = {Optimal Information Processing and {Bayes's} Theorem},
 volume = {42},
 year = {1988}
}

@book{oaksford2007bayesian,
  title={Bayesian Rationality: The Probabilistic Approach to Human Reasoning},
  author={Oaksford, M. and Chater, N.},
  year={2007},
  publisher={Oxford University Press}
}

@article{friston2005theory,
  title={A theory of cortical responses},
  author={Friston, K.},
  journal={Philosophical Transactions of the Royal Society B: Biological Sciences},
  volume={360},
  number={1456},
  pages={815--836},
  year={2005}
}

@article{mossel2017opinion,
  title={Opinion exchange dynamics},
  author={Mossel, E. and Tamuz, O.},
  journal={Probability Surveys},
  volume={14},
  pages={155--204},
  year={2017}
}

@article{conlisk1996bounded,
  title={Why bounded rationality?},
  author={Conlisk, J.},
  journal={Journal of Economic Literature},
  volume={34},
  number={2},
  pages={669--700},
  year={1996}
}

@article{hkazla2021bayesian,
  title={Bayesian decision making in groups is hard},
  author={Hazla, J. and Jadbabaie, A. and Mossel, E. and Rahimian, M. A.},
  journal={Operations Research},
  volume={69},
  number={2},
  pages={632--654},
  year={2021}
}

@ARTICLE{chamley2013,
  author={Chamley, C. and Scaglione, A. and Li, L.},
  journal={IEEE Signal Processing Magazine}, 
  title={Models for the Diffusion of Beliefs in Social Networks: An Overview}, 
  year={2013},
  volume={30},
  number={3},
  pages={16-29}
  }

@ARTICLE{vial2021local,
  author={Vial, D. and Subramanian, V.},
  journal={IEEE Transactions on Control of Network Systems}, 
  title={Local Non-{Bayesian} Social Learning With Stubborn Agents}, 
  year={2022},
  volume={9},
  number={3},
  pages={1178-1188}}

@ARTICLE{djuric2012,
  author={Djurić, P. M. and Wang, Y.},
  journal={IEEE Signal Processing Magazine}, 
  title={Distributed {Bayesian} learning in multiagent systems: Improving our understanding of its capabilities and limitations}, 
  year={2012},
  volume={29},
  number={2},
  pages={65-76},
}

@article{degroot1974,
 author = {M. H. DeGroot},
 journal = {Journal of the American Statistical Association},
 number = {345},
 pages = {118--121},
 publisher = {[American Statistical Association, Taylor & Francis, Ltd.]},
 title = {Reaching a Consensus},
 volume = {69},
 year = {1974}
}

@ARTICLE{lalitha_2018,
  author={Lalitha, A. and Javidi, T. and Sarwate, A. D.},
  journal={IEEE Transactions on Information Theory}, 
  title={Social Learning and Distributed Hypothesis Testing}, 
  year={2018},
  volume={64},
  number={9},
  pages={6161-6179},
}

@ARTICLE{nedic_2017,
  author={Nedić, A. and Olshevsky, A. and Uribe, C. A.},
  journal={IEEE Transactions on Automatic Control}, 
  title={Fast Convergence Rates for Distributed Non-{Bayesian} Learning}, 
  year={2017},
  volume={62},
  number={11},
  pages={5538-5553},
}

@article{sayed_2014,
month = {July},
year = {2014},
volume = {7},
journal = {Foundations and Trends in Machine Learning},
title = {{Adaptation, learning, and optimization over networks}},
number = {4-5},
pages = {311-801},
author = {A. H. Sayed}
}

@article{jadbabaie_2012,
title = {Non-{Bayesian} social learning},
journal = {Games and Economic Behavior},
volume = {76},
number = {1},
pages = {210-225},
year = {2012},
author = {A. Jadbabaie and P. Molavi and A. Sandroni and A. Tahbaz-Salehi}
}

@article{acemoglu_2011,
title = {Bayesian Learning in Social Networks},
journal = {The Review of Economic Studies},
volume = {78},
number = {4},
pages = {1201-1236},
year = {2011},
author = {D. Acemoglu and M. A. Dahleh and I. Lobel and A. Ozdaglar}
}

@ARTICLE{bordignon_2021,
  author={Bordignon, V. and Matta, V. and Sayed, A. H.},
  journal={IEEE Transactions on Information Theory}, 
  title={Adaptive Social Learning}, 
  year={2021},
  volume={67},
  number={9},
  pages={6053-6081}
}

\end{document}